\documentclass[bj,authoryear]{imsart}

\usepackage{dsfont}
\usepackage{bm}
\usepackage{algorithm}

\startlocaldefs
\newtheorem{Proposition}{Proposition}

\newtheorem{Corollary}{Corollary}
\newtheorem{Theorem}{Theorem}
\newtheorem{Lemma}{Lemma}

\newtheorem{Example}{Example}
\newtheorem{Assumption}{Assumption}
\DeclareMathOperator{\re}{\mathbb{R}}
\DeclareMathOperator{\na}{\mathbb{N}}

\newcommand{\E}{\mathbb{E}}
\def\proba{\mathbb{P}}
\newcommand{\ind}{\mathds{1}}
\newcommand{\ee}{\mathrm{e}}
\newcommand{\Prob}{\mathbb{P}}

\newcommand{\neigh}{\mathbf{N}}
\newcommand{\ulambda}{\underline{\lambda}}
\newcommand{\Xset}{\bm{\mathcal{X}}}
\newcommand{\Xalg}{\bm{\mathsf{X}}}
\def\tXset{\tilde{\Xset}}
\newcommand{\itXset}{\tXset^{\circ}}
\def\bx{\mathbf{x}}
\def\bbx{\bar{\mathbf{x}}}

\def\bX{\mathbf{X}}
\def\bY{\mathbf{Y}}
\def\tbX{\tilde{\bX}}
\def\tbY{\tilde{\bY}}

\def\tlambda{\tilde\lambda}

\def\by{\mathbf{y}}
\def\bby{\bar{\mathbf{y}}}

\newcommand{\R}{\bm{\mathcal{R}}}
\newcommand{\nset}{\mathbb{N}}
\newcommand{\rset}{\mathbb{R}}
\newcommand{\1}{\mathds{1}}

\def\bz{\mathbf{z}}

\def\tpi{\tilde{\pi}}
\def\tP{\tilde{P}}
\def\cov{\mathbb{C}\mathrm{ov}}
\def\rev{\mathrm{rev.}}
\def\MH{\mathrm{MH}}

\def\vara{\mathrm{var}}
\def\var{\mathbb{V}\mathrm{ar}}
\def\d{\mathrm{d}}
\newcommand{\pscalpi}[2]{\left\langle #1,#2\right\rangle}

\def\Ltwo{{\mathcal{L}^2}}
\def\Ltwoz{\mathcal{L}_{0,1}^2}
\def\Ltwozstar{\mathcal{L}_{0,1}^{2, *}}

\def\esp{\mathbb{E}}
\newcommand{\pscal}[2]{\langle\,#1,#2\,\rangle}
\def\bpi{\bar{\pi}}

\def\tv{\mathrm{tv}}

\def\eps{\epsilon}
\def\bdelta{\bar{\delta}}

\def\btheta{\boldsymbol\theta}
\def\tneigh{\tilde{\neigh}}

\endlocaldefs

\allowdisplaybreaks

\begin{document}

\def\figureautorefname{Figure}
\def\algorithmautorefname{Algorithm}
\def\sectionautorefname{Section}
\def\subsectionautorefname{Section}
\def\subsubsectionautorefname{Section}
\def\Propositionautorefname{Proposition}
\def\Theoremautorefname{Theorem}
\def\Lemmaautorefname{Lemma}
\def\Corollaryautorefname{Corollary}
\def\Exampleautorefname{Example}
\def\Remarkautorefname{Remark}
\def\Assumptionautorefname{Assumption}

\begin{frontmatter}
\title{An asymptotic Peskun ordering and its application to lifted samplers}
\runtitle{An asymptotic Peskun ordering and its application to lifted samplers}

\begin{aug}
\author[A]{\fnms{Philippe} \snm{Gagnon}\ead[label=e1,mark]{philippe.gagnon.3@umontreal.ca}}
\author[A]{\fnms{Florian} \snm{Maire}\ead[label=e2,mark]{florian.maire@umontreal.ca}}
\address[A]{Department of Mathematics and Statistics, Universit\'{e} de Montr\'{e}al, \printead{e1,e2}}
\end{aug}

\begin{abstract}
A Peskun ordering between two samplers, implying a dominance of one over the other, is known among the Markov chain Monte Carlo community for being a remarkably strong result. It is however also known for being a result that is notably difficult to establish. Indeed, one has to prove that the probability to reach a state $\by$ from a state $\bx$, using a sampler, is greater than or equal to the probability using the other sampler, and this must hold for all pairs $(\bx, \by)$ such that $\bx \neq \by$. We provide in this paper a weaker version that does not require an inequality between the probabilities for all these states: essentially, the dominance holds asymptotically, as a varying parameter grows without bound, as long as the states for which the probabilities are greater than or equal to belong to a mass-concentrating set. The weak ordering turns out to be useful to compare \textit{lifted} samplers for \textit{partially-ordered} discrete state-spaces with their Metropolis--Hastings counterparts. An analysis in great generality yields a qualitative conclusion: they asymptotically perform better in certain situations (and we are able to identify them), but not necessarily in others (and the reasons why are made clear). A quantitative study in a specific context of graphical-model simulation is also conducted.
\end{abstract}

\begin{keyword}[class=MSC2020]
\kwd[primary ]{62-08}
\kwd[; secondary ]{62F15}
\end{keyword}

\begin{keyword}
\kwd{Bayesian statistics}
\kwd{binary random variables}
\kwd{Ising model}
\kwd{Markov chain Monte Carlo methods}
\kwd{variable selection}
\end{keyword}

\end{frontmatter}

\section{Introduction}\label{sec:intro}

\subsection{Peskun ordering: context, original version and some variants}\label{sec:Peskun_and_variants}

Let us consider the situation where one is interested in sampling from $\pi$, a probability distribution defined on a measurable space $(\Xset, \Xalg)$, with $\Xset$ finite and assumed to correspond to the support of $\pi$, and $\Xalg$ a sigma-algebra on $\Xset$. In a sampling context, $\pi$ is often referred to as the \textit{target distribution}. Let us consider that, to sample from $\pi$, one has access to two Markov chain Monte Carlo (MCMC) algorithms and wonders which one is best. Establishing a Peskun ordering \citep{peskun1973optimum} is possibly the most sought-after route when one wants to prove that a given MCMC algorithm is superior in terms of statistical efficiency to another. The statistical efficiency is measured in terms of asymptotic variances: for any Markov kernel $P$ acting on $(\Xset,\Xalg)$ and for any $f:\Xset\to\rset$, we denote by $\vara(f, P)$ the asymptotic variance in a central limit theorem for a MCMC estimator of $\pi f$, the expectation of $f(\bX)$ under $\bX\sim\pi$. In this paper, all considered Markov kernels are assumed to be irreducible and aperiodic, so that the associated samplers are valid\footnote{By valid, we mean that a law of large numbers and a central limit theorem hold for time-averages of functionals of Markov chains.}. The original ordering is presented in \autoref{thm:peskun}.

\begin{Theorem}[\citealp{peskun1973optimum}]\label{thm:peskun}
 Let $P_1$ and $P_2$ be two Markov kernels that are reversible with respect to $\pi$. If $P_1(\bx, \by) \geq P_2(\bx, \by)$ for all $(\bx, \by) \in \Xset^2$ with $\bx \neq \by$, then $\vara(f, P_1) \leq \vara(f, P_2)$ for all $f : \Xset \to \rset$.
\end{Theorem}

The strength of this result lies in its universality: the order between the asymptotic variances holds for \emph{all} functions $f$, which explains why we say that a sampler associated with $P_1$ is superior to a sampler associated with $P_2$, for the problem at hand. This ordering is however known to be rather challenging to establish. It is indeed only in specific situations that one can establish that the probability to reach $\by$ from $\bx$ with $P_1$ is greater than or equal to that with $P_2$, and this for all $(\bx, \by) \in \Xset^2$ with $\bx \neq \by$.

The result of \cite{peskun1973optimum} was generalized in several ways. First, \cite{tierney1998note} extended it to general state-spaces. \cite{andrieu2018uniform} then provided a quantitative form requiring that the order on the Markov kernels holds, but up to a multiplicative factor, that is $P_1(\bx, \by) \geq \omega \, P_2(\bx, \by)$ for all $(\bx, \by) \in \Xset^2$ with $\bx \neq \by$, for some $\omega > 0$, while yielding similar conclusions:
\[
 \vara(f, P_1) \leq \frac{\vara(f, P_2)}{\omega} + \frac{1 - \omega}{\omega} \, \var[f(\bX)].
\]
These results are valid for reversible Markov chains only. Recently, \cite{andrieu2019peskun} went beyond the reversible scenario. These authors consider a specific type of non-reversibility for which the chains can be seen as being ``almost'' reversible; they are reversible, up to an involution. This type of non-reversibility nevertheless covers a remarkably large number of known non-reversible MCMC algorithms, including \textit{lifted} algorithms \citep{horowitz1991generalized, gustafson1998guided, chen1999lifting, diaconis2000analysis}.

\subsection{Our proposal: a weaker and asymptotic version}\label{sec:intro_our_ordering}

With a result as strong as the original ordering, it is somewhat expected to be difficult to establish it. The main result of this paper is that a weaker version of this ordering can lead to similar, but weaker, conclusions. This weaker ordering\footnote{For brevity, we will use ``weaker ordering'' or ``weak ordering'' to refer to the proposed weaker version of Peskun's ordering. As will be seen, using such expressions is however an abuse of terminology because the binary relation defined by our ``weak ordering''  does not establish an order on the set of reversible Markov kernels in the mathematical sense.} is particularly well suited for situations where the two Markov chains of interest are well understood, but only on some subsets of the state-space. We believe that this weaker version will allow to compare samplers in situations in which it was not possible before. Indeed, we believe that the difficulty in establishing the original ordering comes from the verification of $P_1(\bx, \by) \geq P_2(\bx, \by)$ for \emph{all} $(\bx, \by) \in \Xset^2$ with $\bx \neq \by$.

Recent concepts such as \textit{approximate spectral gaps} introduced in \cite{atchade2021approximate} and \textit{large sets} proposed in \cite{yang2017complexity} have shown that bounds on the convergence time of Markov chains can be obtained by exploiting the particular behaviour of the process on some subset of the state-space. When the process is particularly efficient on such a subset, resulting bounds can be tighter than traditional ones that account for the whole state-space. We here leverage similar ideas.

Consider that an order on the probabilities $P_1(\bx, \by) \geq P_2(\bx, \by)$ can be established but only on a subset $\tXset^2 \subset \Xset^2$. It is natural to expect that if the mass concentrates on $\tXset$ and under some mixing guarantees (to guarantee that when the chains exit $\tXset$ they do not take too long to come back), then $\vara(f, P_1) \lessapprox \vara(f, P_2)$ for a class of functions $f$, where the approximation is a consequence of working under a limiting regime to represent a phenomenon of mass concentration. In the following, we prove a result essentially corresponding to that just described. We now provide an overview of a motivating application which is explored in the manuscript.

\subsection{Lifted samplers: a motivating application}\label{sec:intro_lifted}

\textit{Lifting} the state-space is a generic technique which yields what are referred to as \textit{lifted} samplers. The state-space is \textit{lifted} (i.e., extended) to incorporate auxiliary variables. The idea is to think of the random variables we want to sample as \textit{position} variables and to associate to them \textit{direction} variables, to \textit{guide} the Markov chains so as to avoid backtracking, a behaviour often exhibited by reversible schemes that is suspected to increase the autocorrelation of the process. Consider for instance that $\Xset = \{1, \ldots, K\}$, where $K$ is a positive integer. We associate to the variable $\bX$ a direction variable $\nu \in \{-1, +1\}$. A Markov chain is defined on the lifted state-space $\Xset \times \{-1, +1\}$. The lifted sampler proceeds as a Metropolis--Hastings (MH, \citep{metropolis1953equation, hastings1970monte}) algorithm in the sense that a proposal is accepted with a given probability, but in this case the proposal is deterministic and given by $\by = \bx + \nu$ when $(\bx, \nu)$ is the current state. The randomness thus comes from the decision to accept or reject the proposal; in the latter case, the direction is reversed. The lifting technique can be seen as a way to equip the resulting stochastic process with some memory of its past (the direction it comes from), while retaining the Markov property. It can be shown that the resulting Markov chains admit $\pi \otimes \mathcal{U}\{-1, 1\}$ as invariant distribution, where $\mathcal{U}\{-1, 1\}$ denotes the uniform distribution over the set $\{-1, 1\}$ and $\pi \otimes \mathcal{U}\{-1, 1\}$ is the product measure. The sampler is thus valid and expectations under $\pi$ can be approximated by considering functions $f : \Xset \times \{-1, +1\} \to \rset$ of solely the first argument.

Let $P_{\text{lifted}}$ be the Markov kernel associated to this algorithm, and let $P_{\text{MH}}$ be the Markov kernel associated to its non-lifted counterpart, which is a MH algorithm proposing  $\by = \bx + 1$ or $\by = \bx - 1$, each with probability $1/2$. Theorem 7 in \cite{andrieu2019peskun} allows to establish that $\vara(f, P_{\text{lifted}}) \leq \vara(f, P_{\text{MH}})$,  for \emph{any} $f$ of solely the first argument and \emph{any} distribution $\pi$. As Peskun's, this result is universal. It is however remarkable that it holds, not only for any $f$, but also for any $\pi$. It is also remarkable to obtain such a result given that the lifted sampler is implemented at no additional computational cost over its non-lifted counterpart, and also with no additional implementation difficulty (lifted samplers often possess these qualities). The result on the order between the asymptotic variances is essentially a consequence of having the same acceptance probabilities in both algorithms. There is thus no lost in terms of acceptance probabilities by using the lifting technique, while there is a potential gain in terms of persistent movement.

The superiority of $P_{\text{lifted}}$ over $P_{\text{MH}}$ for any $\pi$ at no additional computational cost motivates an investigation of lifted samplers for other types of discrete state-spaces, especially given the limited number (or rather the absence) of real-world models where the state-space is of the form $\Xset = \{1, \ldots, K\}$. This latter set is totally ordered; a natural first step in the investigation is thus to consider \textit{partially-ordered} discrete state-spaces. A definition of partially-ordered sets as well as a generic lifted algorithm to sample from distributions defined on such a set are presented in \autoref{sec:lifted_sampler}. Important applications of such an algorithm include simulation of systems formed from binary variables, such as those simulated using the Ising model, and Bayesian variable selection when the posterior model probabilities can be evaluated, up to a normalizing constant.

In the case of partially-ordered discrete state-spaces, Theorem 7 of \cite{andrieu2019peskun} still allows to prove the superiority of the lifted algorithm over its non-lifted counterpart, which is a reversible sampler; however in this case, the non-lifted counterpart does not correspond to the MH algorithm over which we wish to establish a superiority. This is essentially because, contrarily to the totally-ordered case, the acceptance probabilities in the MH and lifted algorithms are in general different. In certain situations, they can be quite unbalanced in some area of the state-space in the lifted algorithm, while they are not in the MH sampler. In contrast, in some other area of the state-space, the acceptance probabilities are similar. When the mass concentrates on the latter area, we explore the possibility of applying our weak ordering to compare the non-lifted counterpart and the MH algorithm to prove that the lifted sampler is superior to the MH algorithm.

\subsection{Organization of the paper}\label{sec:organization}

We now describe how the rest of the paper is organized. We introduce our asymptotic Peskun ordering in \autoref{sec:weak_peskun}. We next use this result to identify situations in which the lifted samplers for partially-ordered discrete state-spaces are expected to outperform (or not) their MH counterparts. Regarding the organization of this part, we first present the lifted samplers in \autoref{sec:lifted_sampler} and then carry out in \autoref{sec:specific_samplers} an analysis in great generality. Given that the analysis is carried out in great generality, we are not in a position to verify the assumptions under which our asymptotic ordering holds. Rather, the analysis aims to establish the results that can be obtained whenever the assumptions are verified. We next conduct in \autoref{sec:simple_Ising} a thorough study in a context of a simulation of a simple graphical model. With this simple model, we are in a good position to verify the assumptions of our theoretical result; \autoref{sec:simple_Ising} serves as a user guide for applying our asymptotic Peskun ordering. The model corresponds to a Ising model with an external field, but without spatial correlation. The target distribution thus factorizes and the components of $\bx$ are independent; the external field defines the marginal distributions. The model can be seen as an approximation to that with weak spatial correlation, referred to in the literature as a \textit{high temperature} model. We will refer to the model studied in \autoref{sec:simple_Ising} as \textit{the simple Ising model}. The main part of the manuscript finishes in \autoref{sec:discussion} with retrospective comments and possible directions for future research. In \autoref{sec:num_experiments}, we study more complex problems for which a verification of the assumptions is not possible. The first problem is about the simulation of a Ising model which is more complex than that considered in \autoref{sec:simple_Ising} (with spatial correlation). Numerical results are provided and they are consistent with the theoretical ones presented in \autoref{sec:simple_Ising}. The second problem is that of variable selection in a real-life situation. All proofs of theoretical results are deferred to \autoref{sec:proofs}. While the paper is concerned with efficient sampling of distributions defined on discrete state-spaces, we stress that numerous results and elements of our analyses translate immediately to general state-space contexts.

\section{A weaker and asymptotic version of Peskun's ordering}\label{sec:weak_peskun}

Before presenting the theoretical result, we provide the intuition behind it (while being more precise than in \autoref{sec:intro_our_ordering}). This will help justify the assumptions, allow to highlight its relevance, and in fact allow to present a sketch of the proof. Beforehand, we introduce required notation.

In all this section, we consider that the distribution of interest $\pi$ is parameterized by some $n \in \nset$, that is $\pi \equiv \pi_n$. The state-space may also be parameterized by $n$ and is thus denoted by $\Xset_n$; we assume that, for each $n \in \nset$, $\Xset_n$ is finite. We define two collections of Markov kernels, $\{P_{1, n}\}$ and $\{P_{2, n}\}$, for which $P_{1, n}$ and $P_{2, n}$ are $\pi_n$-reversible for all $n$. We define a collection of subsets $\{\tilde{\Xset}_n \subset \Xset_n\}$ which we refer to as \textit{control subsets}. We introduce two collections of restricted kernels $\{\tilde{P}_{1, n}\}$ and $\{\tilde{P}_{2, n}\}$ which, for all $n$, are defined for any $(\bx, \by) \in \tilde{\Xset}_n^2$ by
\[
\tP_{i, n}(\bx,\by) := P_{i, n}(\bx,\by) + P_{i, n}(\bx,\Xset_n \backslash \tilde{\Xset}_n) \1_{\by=\bx}, \qquad i \in \{1, 2\}.
\]
The form of states like $\bx$ and $\by$ may depend on $n$, but we make this dependence implicit to simplify. We let $\tpi_n$ be the probability measure defined as $\tpi_n:=\pi_n(\,\cdot\,\cap \tXset_n)/\pi_n(\tXset_n)$. It can be readily checked that  $\tP_{1, n}$ and $\tP_{2, n}$ are both $\tpi_n$-reversible, for all $n$. We define what we call (with some abuse of terminology) the \emph{interior} and the \emph{boundary} of $\tXset_n$ as $\itXset_n := \{\bx \in \tXset_n: P_{i,n}(\bx, \tXset_n^\mathsf{c}) = 0\}$ and $\partial \tXset_n := \tXset_n \setminus \itXset_n$, respectively, where we assume that the definition of $\itXset_n$ is the same for $i = 1, 2$. The functions for which we want to approximate the expectations may also depend on $n$ and are thus denoted by $f_n$. The $\pi_n$-weighted scalar product and $p$-norm are defined as $\pscal{f_n}{g_n}_{\pi_n} := \sum_{\bx\in\Xset}f_n(\bx) \, g_n(\bx) \, \pi_n(\bx)$ and $\|f_n\|_{\pi_n, p} := [\sum_{\bx\in\Xset} |f_n(\bx)|^p \pi_n(\bx)]^{1/p}$, respectively, with $\|f_n\|_{\pi_n}$ for the $2$-norm. In this section, we consider that the functions are standardized, meaning that $f_n \in \Ltwoz(\pi_n)$, where $\Ltwoz(\pi_n) := \{f_n: \pi_n f_n = 0 \text{ and } \|f_n\|_{\pi_n} = 1\}$. This should not be seen as a restriction given that the magnitude of asymptotic variances, which is proportional to $\|f_n\|_{\pi_n}^2$, is irrelevant when it is of interest to establish an order among them. We note that since for each $n$, $\Xset_n$ is finite, $P_{1, n}$ and $P_{2, n}$ admit a non-trivial \emph{right spectral gap} in $\Ltwoz(\pi_n)$, whose variational expression is given by
\begin{equation}\label{spec_gap}
\lambda_i(n)
:= \inf_{f_n\in\Ltwoz(\pi_n):\|f_n\|_{\pi_n} > 0}\frac{\pscal{f_n}{(I_n - P_{i, n}) f_n}_{\pi_n}}{\|f\|_{\pi_n}^2}\,,\quad i \in \{1, 2\},
\end{equation}
where $I_n$ is the identity on $\Ltwoz(\pi_n)$. In particular, it can be proved that $\lambda_i(n)\in(0,2)$. We analogously define the right spectral gaps of $\tP_{i, n}$ and denote them by $\tlambda_i(n)$, $i = 1, 2$, and we define $\ulambda(n) := \min\{\lambda_1(n), \lambda_2(n), \tlambda_1(n), \tlambda_2(n)\}$. In the following, we refer to the \emph{right spectral gap} of a kernel simply as the \emph{spectral gap} to simplify. Finally, we will use $o$ for the little-o notation.

Consider that one wants to establish a Peskun-type ordering between two kernels, but one is only able to establish a (suitable) order on the kernels on a subset of the state-space in the following sense: $P_{1, n}(\mathbf{x}, \mathbf{y}) \geq \omega(n) P_{2, n}(\mathbf{x}, \mathbf{y})$ for all $(\bx,\by) \in \tXset_n^2$ with $\mathbf{x} \neq \mathbf{y}$ where $\omega(n)$ is a (suitable) positive constant which may depend on $n$. This ordering implies that $\tP_{1, n}(\mathbf{x}, \mathbf{y}) \geq \omega(n) \tP_{2, n}(\mathbf{x}, \mathbf{y})$ for all $(\bx,\by) \in \tXset_n^2$ with $\mathbf{x} \neq \mathbf{y}$, which in turn implies that
\begin{align}\label{eqn:Peskun_quan}
\vara(f_n, \tP_{1, n}) \leq \frac{1}{\omega(n)} \vara(f_n, \tP_{2, n}) + \frac{1}{\omega(n)} - 1,
\end{align}
by, as mentioned in \autoref{sec:Peskun_and_variants}, \cite{andrieu2018uniform} (Lemma 33).

Let us consider that $\pi_n$ concentrates on $\itXset_n$. The notion of concentration of $\pi_n$ naturally implies that we are interested by a certain asymptotic regime, which justifies that we consider a limit $n \rightarrow \infty$. Under this regime, $\pi_n(\itXset_n)\rightarrow 1$, implying that $\pi_n(\tXset_n) \rightarrow 1$. One can imagine that, if the Markov chains associated with $P_{1, n}$ and $P_{2, n}$ do not behave ``too badly'' outside of $\tXset_n$, meaning that when they reach the complement $\tXset_n^\mathsf{c}$ they do not stay there for ``too long'', then $\vara(f_n, \tP_{1, n})$ and $\vara(f_n, \tP_{2, n})$ should be similar to $\vara(f_n, P_{1, n})$ and $\vara(f_n, P_{2, n})$. This is what we show in order to prove our theoretical result. In fact, if we think of $P_{1, n}, P_{2, n}, \tP_{1, n}$ and $\tP_{2, n}$ as samplers, it is seen in the proof that in order to establish a connection between the asymptotic variances, it simplifies to assume that the performance of the worst of these samplers, measured through $\ulambda(n)$, is not ``too poor'',  which is a stronger assumption than a performance assumption on $P_{1, n}$ and $P_{2, n}$ only. Under these assumptions, we are able to establish that $\vara(f_n, P_{i, n})$ is equal to $\vara(f_n, \tP_{i, n})$, up to an error term that depends on $n$ and that vanishes in the large $n$ regime, $i \in \{1, 2\}$, which essentially yields our result. The concentration assumption is reasonable given that in practice the mass often concentrates on a subset of the state-space. This is especially true in high dimensions or when the sample size is large in Bayesian statistics contexts (see, e.g., \cite{VanderVaart2000} and \cite{kleijn2012}).

In light of the above, it is understood that three assumptions are required: the order on the kernels on the control subset, the concentration of $\pi_n$ and a performance guarantee on the samplers. We now state formally the first two assumptions and then present a simplified version of the theoretical result with a strong performance guarantee. We next present a more general version. To simplify the results, yet keeping the focus on most important cases, we consider in the following that $\omega(n) \leq 1$, meaning that we exclude cases where $P_{1, n}$ is overly dominant on $\tXset_n$.


\begin{Assumption}[Kernel ordering]\label{ass:1}
    For each $n$, $P_{1, n}(\mathbf{x}, \mathbf{y}) \geq \omega(n) P_{2, n}(\mathbf{x}, \mathbf{y})$ for all $(\bx,\by) \in \tXset_n^2$ with $\mathbf{x} \neq \mathbf{y}$, where $\omega(n)$ admits a limit, that is $\lim_{n \rightarrow \infty} \omega(n)  =: \overline{\omega} > 0$.
\end{Assumption}

\begin{Assumption}[Mass concentration]\label{ass:2}
    The mass concentrates on $\itXset_n$: $\lim_{n \rightarrow \infty} \pi_n(\itXset_n)  = 1$.
\end{Assumption}

Given that together Assumptions \ref{ass:1} and \ref{ass:2} correspond to the assumptions of a classic Peskun ordering as in \cite{andrieu2018uniform} in the limit, one can only hope to establish, under Assumptions \ref{ass:1} and \ref{ass:2}, a version of this ordering that holds in some limiting sense.

\begin{Theorem}[A simple asymptotic Peskun ordering]\label{thm:1} Suppose that Assumptions \ref{ass:1} and \ref{ass:2} hold. Assume that the spectral gaps of $P_{i, n}$ and $\tP_{i, n}$ are bounded away from zero for all $n$, $i =1, 2$.
Assume also that the sequence $\{f_n\}$ is such that $f_n \in \Ltwoz(\pi_n)$ for all $n$ and such that there exist $\delta > 0$ and $\gamma \in (0, \delta / (2 + \delta))$ with
\begin{equation}\label{eq:ass_thm1}
 \|f_n\|_{\pi_n,2+\delta}=o\left(\frac{1}{\left(1-\pi_n(\itXset_n)\right)^{\gamma}}\right).
\end{equation}
Then, for any $\eps \in(0, \overline{\omega})$, there exists $n^\ast \in \nset$, such that for any $n > n^\ast$
\begin{equation*}
 \vara(f_n, P_{1,n}) \leq \frac{1}{\overline{\omega}-\eps}\vara(f_n, P_{2,n}) +\frac{1}{2}\left(\frac{1}{\overline{\omega} - \eps} + \frac{1}{\overline{\omega}}\right) - 1  + \frac{\eps}{2}.
\end{equation*}
\end{Theorem}

We now make a few remarks about \autoref{thm:1}. It allows to retrieve \eqref{eqn:Peskun_quan} in the limit with $\eps \rightarrow 0$ and $\omega(n) \rightarrow \overline{\omega}$. \autoref{thm:1} will be seen to be a special case of the next one in which the spectral gaps are allowed to decrease with $n$, which is usually the case when $n$ is the dimension of the state-space. As mentioned, considering that the spectral gaps are bounded away from zero simplifies the assumptions, at the price of requiring a strong performance guarantee.

In addition to the three assumptions mentioned earlier, another one is made in \eqref{eq:ass_thm1}. This assumption  essentially states that the class of functions that satisfies \eqref{eq:ass_thm1} have a $(2 + \delta)$-norm that is allowed to grow with $n$, but not faster (in fact slightly slower) than $1 / (1 - \pi_n(\itXset_n))$. It is thus not all sequences $\{f_n\}$ that are admissible. It could be tempting to consider a collection of large subsets $\tXset_n$ to encourage a fast concentration of $\pi_n$ on these sets, thus allowing for a large class of admissible sequences of functions in \autoref{thm:1}; however, the larger are the subsets, the more difficult it becomes to obtain a suitable order on the kernels (\autoref{ass:1}).

Different values of the limit of $\omega(n)$, that is $\overline{\omega}$, yield different interpretations of the result. The most important case is when $\overline{\omega} = 1$ for which we can state that the sampler associated with $P_{1, n}$ \textit{asymptotically dominates} that associated with $P_{2, n}$ (for the functions that are admissible). When $\overline{\omega} < 1$, \autoref{thm:1} allows to state that $P_{1, n}$ is \textit{asymptotically comparable} to $P_{2, n}$, in the sense that we have a guarantee that the sampler associated with $P_{1, n}$ will asymptotically produce estimators with variances that are at worst roughly $1 / \overline{\omega}$ larger than the sampler associated with $P_{2, n}$ (again for the functions that are admissible).

We now present the general asymptotic Peskun ordering.

\begin{Theorem}[A general asymptotic Peskun ordering]\label{thm:2} Suppose that \autoref{ass:1} holds. Consider a sequence $\{f_n\}$ such that $f_n \in \Ltwoz(\pi_n)$ for all $n$. Assume that there exist $\delta > 0$ and $\gamma \in (0, \delta / (2 + \delta))$ that satisfy
 \begin{equation}\label{eq:ass1_thm3}
 \|f_n\|_{\pi_n,2+\delta}=o\left(\frac{1}{\left(1-\pi_n(\itXset_n)\right)^{\gamma}}\right),
\end{equation}
and
\begin{equation}\label{eq:ass2_thm3}
 1 - \pi_n(\itXset_n) = o\left(\ulambda(n)^{3 / (\bdelta - \gamma)}\right),
\end{equation}
where $\bdelta := \delta / (2 + \delta)$. Then, for any $\eps \in (0, \overline{\omega})$, there exists $n^\ast \in \nset$, such that for any $n > n^\ast$
\begin{equation*}
 \vara(f_n, P_{1,n}) \leq \frac{1}{\overline{\omega}-\eps}\vara(f_n, P_{2,n}) +\frac{1}{2}\left(\frac{1}{\overline{\omega} - \eps} + \frac{1}{\overline{\omega}}\right) - 1  + \frac{\eps}{2}.
\end{equation*}
\end{Theorem}

We see that the difference between \autoref{thm:2} and \autoref{thm:1} is that \autoref{ass:2} is replaced by \eqref{eq:ass2_thm3}, an assumption connecting $\pi_n(\itXset_n)$ to $\ulambda(n)$, where the latter is now allowed to decrease. After having selected a sequence $\{f_n\}$ and then $\delta$ and $\gamma$ that satisfy \eqref{eq:ass1_thm3} (which is equivalent to \eqref{eq:ass_thm1} in \autoref{thm:1}), one has to verify that the choice of $\delta$ and $\gamma$ also allows to verify \eqref{eq:ass2_thm3}. This equation states that the concentration of $\pi_n$ on $\itXset_n$ has to be faster than $\ulambda(n)^{3 / (\bdelta - \gamma)}$. Note that when the spectral gaps are bounded away from zero, \eqref{eq:ass2_thm3} is equivalent to \autoref{ass:2}, showing that \autoref{thm:1} is indeed a special case of \autoref{thm:2}.

We acknowledge the fact that estimating certain rates appearing in the conditions of Theorems \ref{thm:1} and \ref{thm:2}, especially the rates of spectral quantities, may constitute a problem in itself. We also acknowledge that our sets of assumptions are probably not optimal, but rather a consequence of our proof technique, and may possibly be improved. However, as mentioned, it is understood that the important aspects (the order on the kernels on the control subset, the mass concentration and performance guarantees) together represent necessary conditions. Given the importance of Peskun-type orderings, we believe it is scientifically interesting to understand under which conditions we can establish a result on the asymptotic variances when an order between $P_{1, n}$ and $P_{2, n}$ holds only on a subset of the state-space.

One may be tempted to assume a (non-trivial) relationship between $\tlambda_i(n)$ and $\lambda_i(n)$ given that $\tP_{i, n}$ is a restriction of $P_{i, n}$ on a subset of the state-space $\Xset_n$. It turns out that counterexamples show that it is not possible to obtain an interesting result in the general case. In regular sampling contexts, we expect the rates at which $\tlambda_i(n)$ and $\lambda_i(n)$ decrease to be in the same regime (i.e., both exponential, both polynomial, etc.). For instance, our analysis in a specific context of graphical-model simulation in \autoref{sec:simple_Ising} shows that the decay is polynomial for $\tlambda_i(n)$ and $\lambda_i(n)$, $i = 1, 2$. The analysis also shows that we can select $\{\tXset_n\}$ such that the mass concentrates on $\itXset_n$ exponentially quickly, implying that \autoref{thm:2} applies, provided $\|f_n\|_{\pi_n,2+\delta}$ does not grow too rapidly.

\section{Lifted samplers for partially-ordered discrete state-spaces}\label{sec:lifted_sampler}

In this section, we start by providing a definition of partially-ordered state-spaces in \autoref{sec:posets}. We next present in \autoref{sec:algo} a generic lifted MCMC algorithm for sampling from distributions on partially-ordered discrete sets. In that section, we make another contribution: we make clear that the implementation of lifted samplers for discrete state-spaces is straightforward, as long as a partial order can be established. We put in contrast this contribution with some of other authors by reviewing the literature about sampling on discrete state-spaces in \autoref{sec:literature}. Note that, in order to match the classical MCMC framework, we consider in this section the target distribution, state-space, and so on, to be fixed, and will thus denote them without a subscript to simplify.

\subsection{Partially-ordered state-spaces}\label{sec:posets}

In set theory, a partial order on a set $\Xset$ is a binary relation defined through a set $\R\subset\Xset^2$  which is reflexive, anti-symmetric, and transitive. A set $\Xset$ on which a partial order can be defined, is called \textit{partially ordered}. For such a set, pairs $(\bx,\by)\in\Xset^2$ with $\bx\neq \by$ are \textit{comparable} when either $(\bx,\by)\in\R$ or $(\by,\bx)\in\R$ and are said \textit{incomparable} otherwise. This represents the difference with a totally-ordered set such as $\na$ or $\re$  in which every pair of different elements is comparable. We denote $\mathbf{x} \prec \mathbf{y}$ whenever $(\bx,\by)\in\R$ and $\bx\neq \by$, implying that $\bx$ and $\by$ are comparable. Of course, this is not the only way to have comparable $\bx$ and $\by$ as we can instead have $\mathbf{y} \prec \mathbf{x}$, that is  $(\by,\bx)\in\R$ and $\bx\neq \by$.

An important example of such sets is when any $\mathbf{x} \in \Xset$ can be written as a vector $\mathbf{x} = (x_1, \ldots, x_n)$ for which each component $x_i$ can be of two types, say Type A or Type B, denoted by $x_i \in \{\text{A}, \text{B}\}$. In this case, an inclusion-based partial order on $\Xset$ can be defined through
\begin{equation}
\label{eq:partial_order}
\R = \left\{(\mathbf{x}, \mathbf{y}) \in \Xset \times \Xset : \{i\,:\,x_i=\text{A}\} \subset
\{i\,:\,y_i=\text{A}\}\right\}.
\end{equation}
It can be readily checked that $\R $ is reflexive, anti-symmetric and transitive.  Moreover, defining  $n_{\text{A}}(\mathbf{x})$ to be the number of Type A components in $\mathbf{x}$, i.e.\ $n_{\text{A}}(\mathbf{x}) = \sum_{i = 1}^n \ind_{x_i = \text{A}}$, we have that a pair $(\bx,\by)\in\Xset^2$ such that $\bx\neq \by$ and $n_A(\bx)=n_A(\by)$ is incomparable.

Partially-ordered sets are encountered in many important areas of statistics including the modelling of binary data using networks or graphs and in variable selection. Indeed, for the former, $\Xset$ can be parameterized such that $\Xset=\{-1,+1\}^n$, where for example for an Ising model, $x_i \in \{-1, +1\}$ represents the state of a spin. For variable selection, $\Xset=\{0,1\}^n$ and $x_i\in\{0,1\}$ indicates whether or not the $i$-th covariate is included in the model employed.

\subsection{Generic algorithm}\label{sec:algo}

Let us assume that a neighbourhood structure $\{\neigh(\mathbf{x})\,:\,\bx\in\Xset\}$ and a partial order $\R$ have been specified on $\Xset$. The sampler that we present is a MCMC algorithm that relies on the lifting technique. The state-space is thus extended: we add a direction variable $\nu \in \{-1, +1\}$ to which we assign a uniform distribution $\mathcal{U}\{-1, +1\}$. The target distribution becomes $\pi \otimes \mathcal{U}\{-1, +1\}$. The idea is to generate proposals belonging to a specific subset of the neighbourhood $\neigh(\mathbf{x})$, where the subset is defined through $\R$ and chosen according to the direction $\nu$, when the current state of the chain is $(\mathbf{x}, \nu)$. In particular, the proposal belongs to $\neigh_{+1}(\mathbf{x}) := \{\mathbf{y} \in \neigh(\mathbf{x}) : \mathbf{x} \prec \mathbf{y}\} \subset \neigh(\mathbf{x})$ when the current state of the direction variable is $\nu =+1$ and to $\neigh_{-1}(\mathbf{x}) := \{\mathbf{y} \in \neigh(\mathbf{x}) : \mathbf{y} \prec \mathbf{x}\} \subset \neigh(\mathbf{x})$ when $\nu = -1$.  The partial order is thus used to induce directions to follow in the state-space. We assume that $\neigh(\bx)$ is formed only of states that are comparable to $\bx$ so that $\neigh_{-1}(\mathbf{x})\cup\neigh_{+1}(\mathbf{x})=\neigh(\mathbf{x})$. Note that $\neigh_{-1}(\mathbf{x})\cap\neigh_{+1}(\mathbf{x})=\varnothing$. The underlying assumption $\bx\not\in \neigh(\bx)$ implies that, strictly speaking, $\neigh$ is not a neighbourhood in a topological sense. We nevertheless carry on with this abuse of terminology.

Recently, successful applications of the lifting technique have been carried out in contexts where the state-space exhibits a one-dimensional discrete parameter which plays a central role in the sampling scheme: the temperature variable in simulated tempering \citep{sakai2016irreversible} and in parallel tempering \citep{syed2019non}, and the model indicator in selection of nested models \citep{gagnon2019NRJ}. When such a one-dimensional feature does not exist, there is no straightforward way of lifting the state-space and inducing directions without facing issues of reducibility or the risk of obtaining inefficient samplers. Leveraging what can be regarded as a directional neighbourhood structure induced by the partial order on $\Xset$ allows to break free from the requirement of resorting to an existing one-dimensional parameter to guide the chain.

In what follows, for each $(\mathbf{x},\nu)\in\Xset\times\{-1, +1\}$, $\neigh_{\nu}(\mathbf{x}) $ shall be referred to as the \textit{$\nu$-directional neighbourhood} of state $\mathbf{x}$. The proposal distribution, denoted by $q_{\mathbf{x}, \nu}$, where $(\mathbf{x}, \nu)$ represents the current state of the Markov chain, is assumed to have its support restricted to $\neigh_\nu(\mathbf{x})$. It will be noticed that the implementation of the generic algorithm is straightforward provided that a partial ordering has been established. Indeed, the required inputs are:
\begin{itemize}
 \item[(i)] a neighbourhood structure $\{\neigh(\mathbf{x}): \mathbf{x}\in\Xset\}$,
 \item[(ii)] a partial ordering $\R$ on $\Xset$,
 \item[(iii)] proposal distributions $q_{\mathbf{x}, \nu}$,
\end{itemize}
and there exist natural candidates for the proposal distributions, as will be explained in \autoref{sec:literature} and, in most cases, for the neighbourhood structure as well.

The MCMC algorithm, which bares a strong resemblance with the guided walk \citep{gustafson1998guided}, is presented in \autoref{algo1}. We use $x \wedge y$ to denote $\min\{x, y\}$. In \autoref{sec:trans-dimensional}, we consider that $\Xset$ is a model space and propose a trans-dimensional version of \autoref{algo1} that can be used for, among others, variable selection when it is not possible to integrate out the model parameters.

\begin{algorithm}[ht]
\caption{A lifted sampler for partially-ordered discrete state-spaces} \label{algo1}

 \begin{itemize}

  \item[1.] Generate $\mathbf{y} \sim q_{\mathbf{x}, \nu}$ and $u \sim \mathcal{U}\mathcal[0, 1]$.

  \item[2.] If
  \begin{align}\label{eq_acc_prob}
   u \leq \alpha_\nu(\mathbf{x}, \mathbf{y}) := 1 \wedge \frac{\pi(\mathbf{y}) \, q_{\mathbf{y}, -\nu}(\mathbf{x})}{\pi(\mathbf{x}) \, q_{\mathbf{x}, \nu}(\mathbf{y})},
  \end{align}
  set the next state of the chain to $(\mathbf{y}, \nu)$. Otherwise, set it to $(\mathbf{x}, -\nu)$.

  \item[3.] Go to Step 1.

 \end{itemize}

\end{algorithm}

Given that $\Xset$ is finite, there exists a boundary, in the sense that, for some $(\bx, \nu)$, $\neigh_{\nu}(\bx)$ is the empty set and there is thus no mass beyond state $\mathbf{x}$ when the direction followed is $\nu$. This is for instance the case in the context of variable selection when $\mathbf{x} = (1, \ldots, 1)$, meaning that the current model is the full model, and the direction is $\nu = +1$. \autoref{algo1} may thus seem incomplete: it does not explicitly specify how the algorithm behaves on the boundary. We can consider that for any $\mathbf{x}\in\Xset$ on the boundary, the support of $q_{\mathbf{x}, \nu}$ is not $\neigh_{\nu}(\bx)$ (because this is the empty set), but instead given by a fictive state outside $\Xset$. Given that the support of $\pi$ is $\Xset$, then any state outside $\Xset$ has zero mass under $\pi$ and such a fictive state is automatically rejected at Step 2. As a consequence, when such a state is proposed, the chain remains at $\mathbf{x}$ and the direction is reversed.  Note that this is a technical requirement. In practice, one can simply skip Step 1 when $\mathbf{x}$ is on the boundary and directly set the next state to $(\mathbf{x}, -\nu)$.


 It is possible to establish that the Markov chain defined by  \autoref{algo1} is $\pi \otimes \mathcal{U}\{-1, +1\}$-invariant by casting it into a more general algorithm framework presented in \cite{andrieu2019peskun}. We present below the associated generalization of \autoref{algo1} which has interesting theoretical features. Beforehand, we introduce necessary notation. Let $\rho_\nu: \Xset \to [0,1]$, for $\nu \in \{-1, +1\}$, be a user-defined function for which we require that  for all $(\mathbf{x}, \nu) \in \Xset \times \{-1, +1\}$:
 \begin{eqnarray}
\label{eq:cdt1_rho}
   && 0\leq \rho_\nu(\mathbf{x})\leq 1 - T_\nu(\mathbf{x}, \Xset), \\
\label{eq:cdt2_rho}  && \rho_{\nu}(\mathbf{x}) - \rho_{-\nu}(\mathbf{x}) = T_{-\nu}(\mathbf{x}, \Xset) - T_{\nu}(\mathbf{x}, \Xset),
\end{eqnarray}
where, for all $(\mathbf{x}, \nu) \in \Xset \times \{-1, +1\}$,
\begin{align*}
 T_\nu(\mathbf{x}, \Xset) := \sum_{\mathbf{x}' \in \Xset} q_{\mathbf{x},\nu}(\mathbf{x}') \, \alpha_\nu(\mathbf{x},\mathbf{x}') = \sum_{\mathbf{x}' \in \neigh_{\nu}(\mathbf{x})} q_{\mathbf{x},\nu}(\mathbf{x}') \, \alpha_\nu(\mathbf{x},\mathbf{x}').
\end{align*}
These conditions  are considered satisfied in the sequel as they guarantee, as established in \autoref{prop_invariance} below, that the Markov chain $\{(\bX,\nu)_k\}$ is $\pi \otimes \mathcal{U}\{-1, +1\}$-invariant and thus that the marginal process $\{\bX_k\}$ is $\pi$-invariant. Let $Q_{\mathbf{x}, \nu}$ be the probability mass function (PMF) defined through $Q_{\mathbf{x}, \nu}(\mathbf{x}') \propto q_{\mathbf{x}, \nu}(\mathbf{x}') \, \alpha_\nu(\mathbf{x}, \mathbf{x}')$. The generalization of \autoref{algo1} is presented in \autoref{algo2}.

\begin{algorithm}[ht]
\caption{A generalization of \autoref{algo1}} \label{algo2}
 \begin{itemize}

  \item[1.] Generate $u \sim \mathcal{U}[0, 1]$.

 \begin{itemize}

    \item[(i)] If $u \leq T_\nu(\mathbf{x}, \Xset)$, generate $\mathbf{y} \sim Q_{\mathbf{x}, \nu}$ and set the next state of the chain to $(\mathbf{y}, \nu)$;

    \item[(ii)] if $T_\nu(\mathbf{x}, \Xset) < u \leq T_\nu(\mathbf{x}, \Xset) + \rho_\nu(\mathbf{x})$, set the next state of the chain to $(\mathbf{x}, -\nu)$;

    \item[(iii)] if $u > T_\nu(\mathbf{x}, \Xset) + \rho_\nu(\mathbf{x})$, set the next state of the chain to $(\mathbf{x}, \nu)$.

 \end{itemize}

  \item[2.] Go to Step 1.
 \end{itemize}
\end{algorithm}

\begin{Proposition}\label{prop_invariance}
 The transition kernel of the Markov chain $\{(\mathbf{X}, \nu)_k\}$ simulated by \autoref{algo2} admits $\pi \otimes \mathcal{U}\{-1, 1\}$ as invariant distribution.
\end{Proposition}

One may notice that $T_\nu(\mathbf{x}, \Xset)$ represents the probability to leave the current state $(\mathbf{x}, \nu)$. In \autoref{algo2}, we thus first decide if we move on from $\bx$, in which case, in Step 1.(i), we randomly select the value of $\by$, the state to move to (using the conditional distribution). It can be readily checked that valid choices for $\rho_\nu$ include  $\rho_\nu(\mathbf{x}) = 1 - T_\nu(\mathbf{x}, \Xset)$  and $\rho_\nu(\mathbf{x}) = \max\{0, T_{-\nu}(\mathbf{x}, \Xset) - T_{\nu}(\mathbf{x}, \Xset)\}$. If $\rho_\nu(\mathbf{x}) = 1 - T_\nu(\mathbf{x}, \Xset)$, the condition for Case (iii) of Step 1 is never satisfied, and the algorithm either accepts the proposal and keeps the same direction, or the proposal is rejected and the direction is reversed. In this case, one can show that \autoref{algo2} corresponds to \autoref{algo1}, which is why \autoref{prop_invariance} allows ensuring the correctness of \autoref{algo1} as well. Setting $\rho_\nu(\mathbf{x})$ otherwise than $\rho_\nu(\mathbf{x}) = 1 - T_\nu(\mathbf{x}, \Xset)$ allows in Case (iii) of Step 1 to keep following the same direction, even when the proposal is rejected. Intuitively, this is desirable when the rejection is due to ``bad luck'', and not because there is low mass in the direction followed. The function $\rho_\nu(\mathbf{x})$ aims to incorporate this possibility in the sampler.

In a typical MCMC framework with continuous state-spaces, the function $\bx\mapsto T_\nu(\mathbf{x}, \Xset)$ is intractable. In such a case, it is therefore usually not possible to set $\rho_\nu(\mathbf{x})$ otherwise than $1 - T_\nu(\mathbf{x}, \Xset)$. This contrasts with our discrete state-space framework in which it is often possible to directly compute $T_\nu(\mathbf{x}, \Xset)$. Theorem~6 in \cite{andrieu2019peskun} states that the best choice of function  $\rho_\nu$ in terms of a mathematical object related to the asymptotic variance is
\begin{align}\label{eqn_rho_star}
 \rho_\nu^*(\mathbf{x}) := \max\{0, T_{-\nu}(\mathbf{x}, \Xset) - T_{\nu}(\mathbf{x}, \Xset)\},
\end{align}
and that the worst choice is $\rho_\nu^{\text{w}}(\mathbf{x}) := 1 - T_\nu(\mathbf{x}, \Xset)$. \autoref{cor_bestrho} below establishes an order on the asymptotic variances in the context of finite state-spaces of this paper. Denote by $P_\rho$ the transition kernel corresponding to \autoref{algo2} for a given function $\rho_\nu : \Xset\to[0,1]$.

\begin{Corollary}\label{cor_bestrho}
 If $\Xset$ is finite, then for any function $\rho_\nu$ satisfying \eqref{eq:cdt1_rho}-\eqref{eq:cdt2_rho} and for any function $f: \Xset \times \{-1, + 1\} \to \rset$ such that $f(\bx, -1) = f(\bx, +1)$, we have $ \vara(f, P_{\rho^*}) \leq \vara(f, P_{\rho}) \leq \vara(f, P_{\rho^{\text{w}}})$.
\end{Corollary}

The price to pay for using $\rho_\nu^*$ instead of $\rho_\nu^{\text{w}}$ is that the algorithm is more complicated to implement because it is required to systematically compute $T_{\nu}(\mathbf{x}, \Xset)$ at each iteration (it is also sometimes required to compute $T_{-\nu}(\mathbf{x}, \Xset)$). Using $\rho_\nu^*$ thus also comes with an additional computational cost. We observed in some numerical experiments that, if we account for this increased computational cost, there is no gain in efficiency of using \autoref{algo2} with $\rho_\nu^*$ over \autoref{algo2} with $\rho_\nu^{\text{w}}$ (corresponding to \autoref{algo1}). One may thus opt for simplicity and implement \autoref{algo1}. Note that the latter and its MH counterpart have essentially the same computational cost.

\subsection{Related work about sampling on discrete state-spaces}\label{sec:literature}

Sampling on discrete state-spaces is typically performed using uniform proposal distributions in reversible samplers. If we consider for instance that $\mathbf{x} = (x_1, \ldots, x_n)$ with $x_1, \ldots, x_n \in \{\text{A}, \text{B}\}$, Glauber dynamics for graphical models or the tie-no-tie sampler for network models selects uniformly at random one of the coordinate, say $x_i$, and proposes to change its value from A to B (B to A) when $x_i = \text{A}$ ($x_i = \text{B}$). Such moves are often rejected when the mass concentrates on a subset of the state-space. To address this issue, \cite{zanella2019informed} recently proposed a \textit{locally-balanced} generic approach for which the probability to select the $i$-th coordinate depends on the relative mass of the resulting proposal, that is $\pi(\mathbf{y}) / \pi(\mathbf{x})$, aiming to propose less ``naive'' moves. \cite{zanella2019informed} proves that the acceptance probabilities converge to 1 in a high-dimensional regime. This property suggests that locally-balanced samplers are efficient, at least in high dimensions. Indeed, samplers for discrete state-spaces typically use the same neighbourhood structure $\{\neigh(\mathbf{x})\,:\,\bx\in\Xset\}$, implying that the range of the proposal distributions is the same and that higher acceptance probabilities often translate into better mixing properties. \cite{zanella2019informed} in fact empirically shows that locally-balanced samplers perform better than alternative solutions to sample from PMFs, and that the difference is highly marked in the high-dimensional regime. Yet, the samplers are reversible, implying that the chains may often go back to recently visited states, or in other words, that the chains exhibit a random-walk behaviour.

In the presented generic algorithms in \autoref{sec:algo}, there is no restriction on the proposal distributions $q_{\mathbf{x}, \nu}$. In \autoref{sec:zanella}, we set them to locally-balanced proposal distributions, thus combining the strengths of the lifting and locally-balanced approaches. An illustration showing the benefit of this combination is provided in \autoref{fig_traces} in which we measure the performance using the effective sample size (ESS) of a statistic, reported per iteration. ESS per iteration is defined as the inverse of the integrated autocorrelation time. When the chains start in stationarity, integrated autocorrelation time corresponds to the asymptotic variance of a standardized version of the statistic. A small asymptotic variance thus corresponds to a high ESS (and vice versa).

 \begin{figure}[ht]
  \centering
  $\begin{array}{cc}
  \textbf{Random-walk behaviour} & \textbf{Persistent movement} \cr
   \vspace{-0mm}\textbf{ESS = 0.12 per it.} & \textbf{ESS = 0.33 per it.} \cr
 \includegraphics[width=0.48\textwidth]{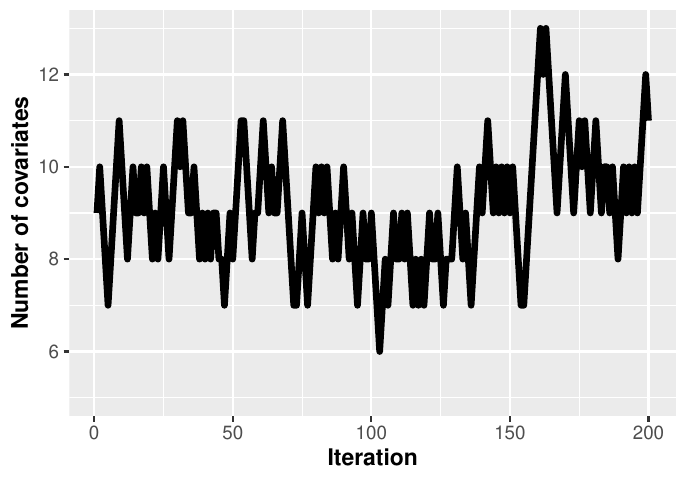} &  \includegraphics[width=0.48\textwidth]{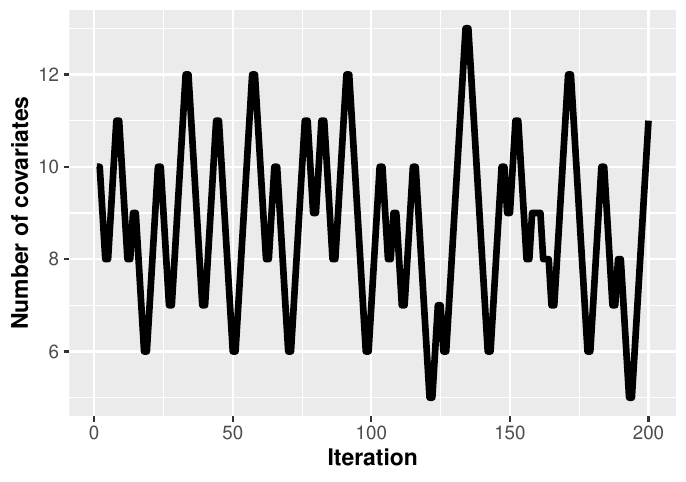}
  \end{array}$\vspace{-5mm}
  \caption{\small Trace plots for the statistic of number of covariates included in a model for a MH sampler with a locally-balanced proposal distribution on the left panel and its lifted counterpart on the right panel, when applied to solve a real variable-selection problem presented in \autoref{sec:variable_selection}.}\label{fig_traces}
 \end{figure}
\normalsize

Other (somewhat) generic approaches to non-reversible sampling on discrete state-spaces are (to our knowledge) all contemporary to ours: \cite{bierkens2016non}, \cite{sakai2016eigenvalue}, \cite{power2019accelerated}, \cite{faizi2020efficient} and \cite{herschlag2020non}. They rely on the lifting technique as well, except \cite{bierkens2016non}. Our work is most closely related to \cite{power2019accelerated} in which the approach of \cite{zanella2019informed} is also exploited. In fact, when $\mathbf{x} = (x_1, \ldots, x_n)$ with $x_1, \ldots, x_n \in \{\text{A}, \text{B}\}$, \autoref{algo1} corresponds to the discrete-time version of a specific sampler independently developed in \cite{power2019accelerated}. \autoref{algo1} can also be seen to be a special case of a sampler presented in \cite{sakai2016eigenvalue} in which a general extended transition matrix is defined from lifting the MH one. A similar approach, described in \cite{faizi2020efficient}, explicitly incorporates the changes in the function $f$ by moving from a state to another in the transition matrix; this latter approach is closely related to ours when $f(\bx)$ decreases or increases every time we change $\bx$ for $\by$ with $\bx \prec \by$. We consequently do not claim originality for the samplers presented here. In those papers, however, the notion of partial ordering is not identified nor exploited; the focus is rather on improving state-space exploration through the exploitation of \emph{any} symmetric or algebraic structure of $\Xset$ identified by users. The focus is the same in \cite{bierkens2016non}, but the non-reversibility is achieved by directly modifying the acceptance probability in MH, using the notion of vorticity matrix; this approach is valid in general state-space contexts. In \cite{herschlag2020non}, the authors generalize non-reversible lifted kernels to \textit{mixed skewed} kernels by means of a series of involutions in a context of undirected graph sampling. In their work, the main application is sampling of districting maps to evaluate the degree of partisan districting. The involutions are created by a series of user-specified vortices that generate non-reversible flows on the state-space. Interestingly, this scheme can be seen as creating directional neighbourhoods.

\section{Two specific lifted samplers and their analysis}\label{sec:specific_samplers}

In this section, we specify two lifted samplers through two different choices of proposal distributions $q_{\mathbf{x}, \nu}$ and provide a theoretical analysis using the asymptotic Peskun ordering. We first present and analyse in \autoref{sec:uniform} a lifted sampler using uniform proposal distributions. As explained in \autoref{sec:literature}, this sampler is often inefficient, especially in high dimensions, but it is simple enough to allow an easy understanding of the reasons why lifted samplers are \emph{not} expected to always dominate their MH counterparts within our framework. We next turn in \autoref{sec:zanella} to a more promising choice of proposal distributions, namely the locally-balanced proposal distributions, and study the resulting lifted sampler.

As mentioned, the study here will be conducted in great generality. More precisely, the target distribution will not be specified; we will thus not be in a position to explicitly estimate the rates appearing in the conditions of our theoretical results presented in \autoref{sec:weak_peskun}. We will make assumptions regarding these rates, but this will not prevent us from defining the control subsets. Making assumptions regarding the rates appearing in the conditions of our theoretical results and defining judiciously the control subsets will allow to gain general insights into the situations in which the lifted samplers are expected to outperform their MH counterparts, and also into those in which there is no guarantee. In \autoref{sec:simple_Ising}, we conduct a thorough study in a specific context of graphical-model simulation. This will allow to have a concrete example of how the assumptions of our theoretical results can be verified in practice. That study will also allow to improve the understanding of the behaviour of lifted samplers through practical results.

For ease of presentation, we consider in this section the setup where $\bx = (x_1, \ldots, x_n)$ and $x_1, \ldots, x_n \in \{\text{A}, \text{B}\}$ with the partial order on $\Xset$ defined in \eqref{eq:partial_order}. We consider, additionally, but without loss of generality, a Ising model context where $\text{A} = -1$ and $\text{B} = +1$. Finally, we consider that the neighbourhood structure, used by all samplers, is the typical one, meaning that the neighbourhoods are set to $\neigh(\mathbf{x}) = \{ \mathbf{y} \in \Xset_n: \sum_i |x_i - y_i| = 2\}$, so that the algorithms propose to flip a single bit at each iteration. Because of the nature of our analysis, we, as in \autoref{sec:weak_peskun}, highlight a dependency on $n$ of the target distribution, the state-space, and so on, by denoting them by $\pi_n$, $\Xset_n$, etc.

\subsection{Uniform proposal distributions}\label{sec:uniform}

In the reversible MH sampler, it is common, as mentioned in \autoref{sec:literature}, to set the proposal distribution, denoted by $q_\mathbf{x}$ for this algorithm, to the uniform distribution over the neighbourhood of the current state $\bx\in\Xset_n$, that is $q_\mathbf{x} = \mathcal{U}\{\neigh(\mathbf{x})\}$. In the framework of Algorithms \ref{algo1} and \ref{algo2}, the analogous proposal distribution is naturally defined as $q_{\mathbf{x}, \nu} = \mathcal{U}\{\neigh_\nu(\mathbf{x})\}$. In this case, the acceptance probability \eqref{eq_acc_prob} of a proposed move becomes
\[
 \alpha_\nu(\mathbf{x}, \mathbf{y}) = 1 \wedge a_\nu(\bx,\by)\,,\qquad a_\nu(\bx,\by):=\frac{\pi_n(\mathbf{y}) \,  |\neigh_{\nu}(\mathbf{x})|}{\pi_n(\mathbf{x}) \, |\neigh_{-\nu}(\mathbf{y})|}\,,
\]
where we refer to $a_\nu$ as the \textit{acceptance ratio}. The function $|\, \cdot \,|$ when applied to a set is the cardinality.

In the MH sampler, given that the neighbourhoods are set to $\neigh(\mathbf{x}) = \{ \mathbf{y} \in \Xset_n: \sum_i |x_i - y_i| = 2\}$, the uniform distribution chooses which bit to flip uniformly at random. Therefore, the size of the neighbourhoods in this sampler is constant for any $\mathbf{x}$ and is given by $n$. This implies that the acceptance probability in this sampler, denoted by $\alpha(\mathbf{x}, \mathbf{y})$, reduces to
\begin{align*}
\alpha(\mathbf{x}, \mathbf{y}) = 1 \wedge a(\mathbf{x}, \mathbf{y})\,,\qquad  a(\mathbf{x}, \mathbf{y}) := \frac{\pi_n(\mathbf{y}) \, q_{\by}(\bx)}{\pi_n(\mathbf{x}) \, q_{\bx}(\by)} = \frac{\pi_n(\mathbf{y})}{\pi_n(\mathbf{x})}.
\end{align*}

In the lifted case, we have that for any $\nu\in\{-1,+1\}$, $n_\nu(\bx)=\sum_{i=1}^n\1_{x_i=\nu}$ and the acceptance probability can thus be rewritten as:
\begin{align}\label{eqn_acc_prob_unif_lifted}
 \alpha_\nu(\mathbf{x}, \mathbf{y}) = 1 \wedge a_\nu(\bx,\by)\,,\qquad a_\nu(\bx,\by)=a(\bx,\by) \, \frac{ n_{-\nu}(\mathbf{x})}{ n_{\nu}(\mathbf{y})}.
\end{align}
Indeed, $\neigh_\nu(\mathbf{x}) = \{ \mathbf{y} \in \Xset_n: \text{there exists one } j \text{ such that } y_j = -x_j = \nu\}$ implies that $|\neigh_\nu(\mathbf{x})| = n_{-\nu}(\mathbf{x})$. The acceptance probability $\alpha_\nu$ thus depends on an additional factor $n_{-\nu}(\mathbf{x}) / n_{\nu}(\mathbf{y})$ compared to $\alpha$ in the MH sampler. While the reversible sampler is allowed to backtrack, which makes the size of the neighbourhoods constant, the size of the neighbourhoods diminishes in the lifted sampler as the chain moves further in a given direction (making the neighbourhoods in the reverse direction bigger and bigger). As a consequence, the longer the acceptance streak, the smaller  $n_{-\nu}(\mathbf{x}) / n_{\nu}(\mathbf{y})$. On an acceptance streak, this factor eventually becomes less than one and thus shrinks $\alpha_\nu$, relatively to the MH acceptance ratio, until the lifted chain switches its direction. To summarize, the price to pay when considering a Markov chain with persistent dynamic is a shrinking factor $n_{-\nu}(\mathbf{x}) / n_{\nu}(\mathbf{y})$ in the acceptance ratio.

An \textit{ideal} situation, which is incompatible with most statistical models, is one where
\begin{equation}\label{eq:ideal_unif}
  |\neigh_{-1}(\mathbf{x})| = |\neigh_{+1}(\mathbf{x})| = |\neigh(\mathbf{x})| / 2 = n / 2, \quad\text{for } \pi_n\text{-almost all } \mathbf{x} \in \Xset_n.
\end{equation}
This implies that if the chain is at state $(\bx, \nu)$, $a(\bx,\by)=a_\nu(\bx,\by)$ for all $\by\in\neigh_{\nu}(\bx)$. Qualitatively, the persistent dynamic of the lifted chain is no longer counter-balanced by the shrinking factor and is thus expected to be more efficient than MH. This fact is made rigourous in \autoref{cor_domination_unif} below, which follows from Theorem 7 of \cite{andrieu2019peskun}. In the rest of this subsection, the transition kernel corresponding to \autoref{algo2} with $q_{\mathbf{x}, \nu} = \mathcal{U}\{\mathcal{N}_\nu(\mathbf{x})\}$ is denoted by $P_{\rho, n}$ and that of the MH sampler with $q_\mathbf{x} = \mathcal{U}\{\mathcal{N}(\mathbf{x})\}$ by $P_{\text{MH}, n}$. Recall that \autoref{algo2} with $\rho_\nu^{\text{w}}$ corresponds to \autoref{algo1}.

\begin{Corollary}\label{cor_domination_unif} Let $n\in\nset$. If $\Xset_n$ is finite and \eqref{eq:ideal_unif} holds, then for any function $f_n: \Xset \times \{-1, + 1\} \to \rset$ such that $f_n(\bx, -1) = f_n(\bx, +1)$, we have $\vara(f_n, P_{\rho, n}) \leq \vara(f_n, P_{\MH, n})$.
\end{Corollary}

The proof of \autoref{cor_domination_unif} is postponed to \autoref{sec:proofs} but its main steps are now presented as they highlight what is important to obtain such an ordering. Central to the proof of \autoref{cor_domination_unif} is the idea that once a lifted sampler is defined, it is possible to identify a non-lifted counterpart which differs from \autoref{algo2} in that the direction is resampled $\nu \sim \mathcal{U}\{-1, +1\}$ at the beginning of each iteration. At each iteration, a choice between $q_{\bx,-1}$ and $q_{\bx,+1}$ is thus first made uniformly at random, and the proposal is next sampled. Non-lifted refers to the fact that, while operating on the extended state-space, the systematic resampling of $\nu$ makes the marginal dynamic $\{\bX_k\}$ Markov again, and reversible. This scheme, when looking at functions $f_n$ with $f_n(\bx, -1) = f_n(\bx, +1)$, makes the extension of state-space to include the direction variable superfluous, explaining how a comparison between $\vara(f_n, P_{\rho, n})$ and $\vara(f_n, P_{\text{MH}, n})$ is possible. Let $P_{\rev, n}$ be the transition kernel of this non-lifted reversible Markov chain. As noted in \cite{andrieu2019peskun}, $P_{\rev,n}$ can indeed be seen as an intermediate kernel through which comparison of the asymptotic variance of $P_{\rho,n}$ and $P_{\MH,n}$ is possible if one can establish, perhaps independently, that $\vara(f,P_{\rho,n})\leq \vara(f,P_{\rev,n})$ and $\vara(f,P_{\rev,n})\leq \vara(f,P_{\MH,n})$. While establishing the former essentially follows from Theorem~7 of \cite{andrieu2019peskun}, the latter may prove more difficult. However, under \eqref{eq:ideal_unif}, it turns out that $P_{\rev,n}=P_{\MH,n}$, trivially establishing that $\vara(f,P_{\rev,n})=\vara(f,P_{\MH,n})$.  Indeed, the sub-stochastic part of $P_{\rev,n}$ associated with accepted proposals is
\begin{equation}\label{eq:prev}
 (1/2)q_{\bx,+1}(\by)\alpha_{+1}(\bx,\by)+(1/2)q_{\bx,-1}(\by)\alpha_{-1}(\bx,\by),
\end{equation}
and it can be readily checked that under \eqref{eq:ideal_unif}, \eqref{eq:prev} indeed coincides with the sub-stochastic part of $P_{\MH,n}$. These are the same mathematical arguments that allow to prove the dominance mentioned in \autoref{sec:intro_lifted} of lifted samplers over their MH counterparts when the state-space is totally ordered.

The incompatibility of the condition \eqref{eq:ideal_unif} with most statistical models motivates us to take our analysis one step further, and this is where the asymptotic Peskun ordering presented in \autoref{sec:weak_peskun} proves useful. Note that in order to find a model such that \eqref{eq:ideal_unif} is satisfied, one has to be quite creative; an example is provided in the supplementary material (\autoref{sec:supp_mat}). The next step in our analysis is to establish if the order on the asymptotic variances 
still holds when \eqref{eq:ideal_unif} is relaxed, and if not, we want to find conditions under which $\vara(f_n, P_{\rho, n})$ and $\vara(f_n, P_{\text{MH}, n})$ can be compared. A modification of our example presented in the supplementary material shows that the order on the asymptotic variances 
\emph{does not} necessarily hold when \eqref{eq:ideal_unif} is relaxed. This should not come as a surprise in the light of the aforementioned observations about the potentially shrinking factor in $\alpha_\nu$. Comparing the efficiency of $P_{\MH, n}$ and $P_{\rho, n}$ beyond the context of \autoref{cor_domination_unif} is not an easy task for several reasons:
\begin{itemize}
\itemsep 0mm
\item $P_{\rho,n}$ is not reversible and most techniques to establish domination results between Markov kernels hold for reversible kernels, \citet{andrieu2019peskun} being a noteworthy exception;
\item the two kernels $P_{\MH,n}$ and $P_{\rho,n}$ are not defined on the same state-space.
\end{itemize}
For these reasons, finding reasonable conditions under which $\vara(f_n, P_{\rev,n})$ and $\vara(f_n, P_{\MH,n})$ can be compared appears to be a suitable route to establish a comparison between $P_{\rho,n}$ et $P_{\MH,n}$ (given that we already know that $\vara(f_n, P_{\rho,n}) \leq \vara(f_n, P_{\rev,n})$ using similar arguments to those used to prove \autoref{cor_domination_unif}). We thus employ Theorems \ref{thm:1} and \ref{thm:2}.

Note that if one manages to design the distributions $q_{\bx, \nu}$ such that $q_{\bx}(\by) = (1/2) q_{\bx, -1}(\by) + (1/2) q_{\bx, +1}(\by)$ for all $\bx, \by$, then one directly has $P_{\rev,n}=P_{\MH,n}$ and thus a comparison between $P_{\rho,n}$ and $P_{\MH,n}$; this is the approach proposed in \cite{kamatani2020non} for general state-spaces, but it is one that cannot in general be applied in the case of discrete state-spaces. Note also that if one is interested in comparing a lifted sampler using proposal distributions $q_{\bx, \nu}$ with a MH sampler using proposal distributions defined as $q_{\bx}(\by) = (1/2) q_{\bx, -1}(\by) + (1/2) q_{\bx, +1}(\by)$, then again $P_{\rev,n}=P_{\MH,n}$ and a comparison between $P_{\rho,n}$ and $P_{\MH,n}$ is direct. In the context of variable selection, the latter MH sampler corresponds to one where it is first chosen to add a covariate or remove one already in the model, and next which covariate to add or delete. In our paper, we focus on the common situation where, in the MH sampler, a proposal is made uniformly at random from $\neigh(\mathbf{x})$ (or using a locally-balanced weight function as described in \autoref{sec:zanella}), and we want to compare a lifted sampler with the MH one.

The idea that we now explore is to consider situations where the mass concentrates on an area where we have a control over the  factor $n_{-\nu}(\mathbf{x}) / n_{\nu}(\mathbf{y})$ in $\alpha_\nu$ \eqref{eqn_acc_prob_unif_lifted}, which translates into the existence of a (non-trivial) relationship between the sub-stochastic part of $P_{\rev,n}$ and that of $P_{\MH,n}$ on this area. To simplify, we consider situations where the mass concentrates on the centre of the domain, i.e., on states where $n_{-1}(\bx)$ and $n_{+1}(\bx)$ are not too far from $n/2$, and set
\begin{align}\label{eq_tX_unif}
\tXset_n := \{\mathbf{x} \in \Xset_n: n / 2 - \beta(n) \leq n_{-1}(\mathbf{x}), n_{+1}(\mathbf{x}) \leq n / 2 + \beta(n)\},
\end{align}
by choosing a specific function $\beta:\nset\to(0,\infty)$. With this definition of $\tXset_n$ and that of the neighbourhood structure (mentioned at the beginning of \autoref{sec:specific_samplers}), we are able to state that the interior of $\tXset_n$ is as follows: $\itXset_n = \{\mathbf{x} \in \Xset_n: n / 2 - \beta(n) + 1 \leq n_{-1}(\mathbf{x}), n_{+1}(\mathbf{x}) \leq n / 2 + \beta(n) - 1\}$. Note that the analysis can be done by considering instead that the mass concentrates on states where the minimum between $n_{-1}(\bx)$ and $n_{+1}(\bx)$ is not too far from $n/\kappa$ with $\kappa \geq 2$. The difference is that, with control subsets defined as in \eqref{eq_tX_unif}, $\omega(n)$ will be seen to converge to $\overline{\omega} = 1$, whereas in the general framework, $\overline{\omega} \leq 1$ and a function of $\kappa$, and the results are more complicated to present. Constructing the control subsets $\{\tXset_n\}$ as in \eqref{eq_tX_unif} implies that, remarkably, the analysis is parameterized by the sole function $\beta$.

\begin{Lemma}
\label{lemma:omega_unif}
Consider the definition of $\tXset_{n}$ in \eqref{eq_tX_unif}. Assume that $\beta$ is such that $\beta(n)=o(n)$. Then for a large enough $n$, it holds that $\tP_{\rev, n}(\bx, \by)\geq \omega(n) \tP_{\MH, n}(\bx,\by)$, for all $(\bx,\by)\in\tXset_{n}^2$ with $\bx\neq \by$, where
\begin{align}\label{eqn_omega1}
 \omega(n) = \left(1 - \frac{\beta(n)}{n / 2}\right) \left(1 + \frac{\beta(n)}{n / 2}\right)^{-2} \rightarrow \overline{\omega} = 1.
\end{align}
\end{Lemma}

Intuitively, if $\beta(n)$ grows like $n$ or faster, then for a large enough $n$ we have $\Xset_n=\tXset_n$ which boils down to the initial Peskun's problem so the assumption $\beta(n)=o(n)$ is sensible. If $\beta(n)$ grows \textit{too} slowly then the control subsets $\tXset_n$ may eventually fail to track the bulk of $\Xset_n$, resulting in that the mass of $\pi_n$ will not concentrate on $\itXset_n$ and that the restricted kernels will be too different from the original ones to allow the machinery of \autoref{sec:weak_peskun} to work. The condition $\beta(n) = o(n)$, together with \eqref{eqn_omega1}, means that \autoref{ass:1} holds with $\overline{\omega} = 1$. If we assume that the spectral gaps are bounded away from zero, which is realistic, for example, when $\beta(n)$ is constant, and that $\pi_n$ concentrates on $\itXset_n$ defined above (implying that \autoref{ass:2} holds), then \autoref{thm:1} can be applied and
\[
 \vara(f_n, P_{\rho,n}) \leq \vara(f_n, P_{\rev,n}) \leq \frac{1}{1 - \epsilon} \vara(f_n, P_{\MH,n}) + \frac{\eps}{2(1-\eps)} + \frac{\epsilon}{2},
\]
for any $\eps > 0$, provided that $n$ is large enough and that we consider functions $f_n \in \Ltwoz(\pi_n)$ satisfying \eqref{eq:ass_thm1} and such that $f_n(\bx, -1) = f_n(\bx, +1)$. The assumption on the spectral gaps can be relaxed and \autoref{thm:2} can be instead applied if we are able to establish a connection between the rate at which $\pi_n$ concentrates on $\itXset_n$ and that at which $\ulambda(n)$ decreases, i.e., if \eqref{eq:ass2_thm3} can be verified.

 To summarize, our analysis suggests that the lifted sampler with uniform proposal distributions dominates its MH counterpart (at least for $n$ large enough and a specific class of functions) when $\pi_n$ concentrates on states in the centre of the domain. If it concentrates elsewhere, then the lifted sampler is expected to be comparable to its MH counterpart as long as $\pi_n$ does not concentrate on areas where the neighbourhoods, and thus the additional factors $n_{-\nu}(\mathbf{x}) / n_{\nu}(\mathbf{y})$ in \eqref{eqn_acc_prob_unif_lifted}, are very unbalanced.

 When $n$ is large, uniform proposal distributions, whether they are used in a lifted or MH sampler, are likely to represent a poor strategy. We will thus not focus on samplers with uniform proposal distributions in our study in a context of graphical-model simulation in \autoref{sec:simple_Ising}. We will rather focus on studying locally-balanced samplers presented in the next subsection which represent efficient alternatives.

\subsection{Locally-balanced proposal distributions}\label{sec:zanella}

In this section, we discuss and analyse samplers using locally-balanced proposal distributions. For simplicity, we will use the same notation as in \autoref{sec:uniform}: $q_{\bx}$ and $q_{\bx, \nu}$ are the proposal distributions in the MH and lifted samplers, respectively, but in this section they are locally-balanced (a definition follows), and $P_{\rho, n}$, $P_{\rev, n}$ and $P_{\MH, n}$ are the Markov kernels associated with \autoref{algo2}, and its non-lifted and MH counterparts, respectively, which are all using locally-balanced proposal distributions. Recall that \autoref{algo1} is a special case of  \autoref{algo2} with $\rho_\nu(\mathbf{x}) = 1 - T_\nu(\mathbf{x}, \Xset)$.

 As defined in \cite{zanella2019informed} in the MH framework, a proposal distribution is locally-balanced if
\begin{align*} 
  q_\mathbf{x}(\mathbf{y}) = g\left(\frac{\pi_n(\mathbf{y})}{\pi_n(\mathbf{x})}\right) \bigg/ c_n(\mathbf{x}), \quad \mathbf{y} \in \neigh(\mathbf{x}),
\end{align*}
where $c_n(\mathbf{x})$ is the normalizing constant, that is $c_n(\mathbf{x}) = \sum_{\mathbf{x}' \in \neigh(\mathbf{x})} g(\pi_n(\mathbf{x}') / \pi_n(\mathbf{x}))$, and $g$ is a positive continuous function such that $g(x) / g(1 / x) = x$ for $x > 0$. Such a function $g$ implies that the acceptance probability in the MH algorithm is given by
\begin{equation}
\label{eqn_acc_prob_inf_MH}
 \alpha(\mathbf{x}, \mathbf{y}) = 1 \wedge \frac{\pi_n(\mathbf{y}) \, q_{\mathbf{y}}(\mathbf{x})}{\pi_n(\mathbf{x}) \, q_{\mathbf{x}}(\mathbf{y})} = 1 \wedge \frac{c_n(\mathbf{x})}{c_n(\mathbf{y})}.
\end{equation}
The name \textit{locally-balanced} comes from the fact that, in the limit, when the state-space becomes larger and larger (but the neighbourhoods have a fixed size and proposed moves are thus local), there is no need for an accept-reject step anymore; the proposal distributions leave the distribution $\pi_n$ invariant. Indeed, as shown in
\cite{zanella2019informed}, $\sup_{(\bx,\by)\in\Xset_n\,:\,\by\in\neigh(\bx)}c_n(\mathbf{x}) / c_n(\mathbf{y}) \rightarrow 1$ as $n \rightarrow \infty$ under some assumptions. The author more precisely considers that $\mathbf{x} = (x_1, \ldots, x_n)$ and that at any given iteration, only a small fraction of the $n$ components is proposed to change values. The result holds when there exists a uniform bound which does not depend on $n$ on $\pi_n(\by)/\pi_n(\bx)$ for all pairs of neighbouring states $(\bx,\by)$ and the random variables $X_1, \ldots, X_n$ exhibit a structure of conditional independence, the latter implying that the normalizing constants $c_n(\mathbf{x})$ and $c_n(\mathbf{y})$ share a lot of terms. Note that $c_n(\mathbf{x})$ and $c_n(\mathbf{y})$ are both sums over the same number of terms, which is crucial in showing that $\sup_{(\bx,\by)\in\Xset_n\,:\,\by\in\neigh(\bx)}c_n(\mathbf{x}) / c_n(\mathbf{y}) \rightarrow 1$.

Two valid choices for $g$ are $g(x) = \sqrt{x}$ and $g(x) = x / (1 + x)$, the latter yielding what is referred to as the \textit{Barker proposal distribution} in reference to \cite{barker1965monte}'s acceptance probability choice. The advantage of the latter choice is that it is a bounded function of $x$, which stabilizes the normalizing constants and thus the acceptance probability, see \cite{zanella2019informed}, and  \cite{livingstone2019robustness} for the continuous-random-variable case.

A locally-balanced proposal distribution in the lifted-sampler framework is naturally defined as
\[
 q_\mathbf{x, \nu}(\mathbf{y}) = g\left(\frac{\pi_n(\mathbf{y})}{\pi_n(\mathbf{x})}\right) \bigg/ c_{n, \nu}(\mathbf{x}), \quad \mathbf{y} \in \neigh_\nu(\mathbf{x}),
\]
where $c_{n, \nu}(\mathbf{x})$ is the normalizing constant and $g$ is as above. In this case,
\begin{align}\label{eqn_acc_prob_inf_lifted}
 \alpha_\nu(\mathbf{x}, \mathbf{y}) = 1 \wedge \frac{\pi_n(\mathbf{y}) \, q_{\mathbf{y}, -\nu}(\mathbf{x})}{\pi_n(\mathbf{x}) \, q_{\mathbf{x}, \nu}(\mathbf{y})} = 1 \wedge \frac{c_{n, \nu}(\mathbf{x})}{c_{n, -\nu}(\mathbf{y})} = 1 \wedge \frac{c_{n}(\mathbf{x})}{c_{n}(\mathbf{y})}\frac{c_{n, \nu}(\mathbf{x})/ c_{n}(\mathbf{x})}{c_{n, -\nu}(\mathbf{y}) / c_{n}(\mathbf{y})}.
\end{align}
As with the uniform proposal distributions in \autoref{sec:uniform}, we see that the acceptance probability in the lifted sampler \eqref{eqn_acc_prob_inf_lifted} differs from that in MH \eqref{eqn_acc_prob_inf_MH}. There is thus again a price to pay to use a lifted sampler: there is no guarantee that $c_{n, \nu}(\mathbf{x}) / c_{n, -\nu}(\mathbf{y}) \rightarrow 1$ for $\by\in\neigh(\bx)$, even when $c_{n}(\mathbf{x}) / c_{n}(\mathbf{y}) \rightarrow 1$. A reason is because the sums $c_{n, \nu}(\mathbf{x})$ and $c_{n, -\nu}(\mathbf{y})$ are in this case not over the same number of terms, a consequence of the nature of the lifted sampler.

As previously, the reversible counterpart to the lifted algorithm chooses at each iteration uniformly at random a proposal distribution between $q_{\mathbf{x}, -1}$ and $q_{\mathbf{x}, +1}$ from which a proposal is sampled. Imagine that $c_n(\mathbf{x}) / c_n(\mathbf{y}) = 1$ for all $\bx, \by$, then one can notice from \eqref{eqn_acc_prob_inf_lifted} that the stability of ratios $c_{n, \nu}(\mathbf{x}) / c_{n, -\nu}(\mathbf{y})$ is crucial to establish a connection between the sub-stochastic parts of $P_{\rev,n}$ and $P_{\MH,n}$ (recall \eqref{eq:prev}). In fact, in an ideal situation, which is again incompatible with most statistical models, one can establish that $P_{\rev,n} = P_{\MH,n}$, guaranteeing a dominance of the lifted sampler.

\begin{Corollary}\label{cor_domination_lifted} Let $n\in\nset$. If $\Xset_n$ is finite and $c_{n, -1}(\mathbf{x}) = c_{n, +1}(\mathbf{x})$, for all $\mathbf{x} \in \Xset_n$, then for any function $f_n: \Xset \times \{-1, + 1\} \to \rset$ such that $f_n(\bx, -1) = f_n(\bx, +1)$, we have $\vara(f_n, P_{\rho, n}) \leq \vara(f_n, P_{\text{MH}, n})$.
\end{Corollary}

Locally-balanced proposal distributions allow to explore the state-space by often proposing points that belong to the subset on which the mass concentrates. \autoref{cor_domination_lifted} tells us that, in order  to compare $P_{\rev,n}$ to $P_{\MH,n}$  (and thus $P_{\rho, n}$ to $P_{\MH,n}$),  the directional neighbourhoods to which these points belong must have similar mass, implying similar normalizing constants $c_{n, -1}(\bx)$ and $c_{n, +1}(\bx)$ over the subset. The analysis can be pushed beyond \autoref{cor_domination_lifted}  by making use of our asymptotic framework. To simplify, we consider, as in \cite{zanella2019informed},  the situation where $\sup c_{n}(\mathbf{x}) / c_{n}(\mathbf{y}) \rightarrow 1$ where the supremum is over all neighbouring states $\bx, \by$, and $\overline{\omega} = 1$.

We now turn to the definition of the control subset:
\begin{align}\label{eq:subset_loc}
\tXset_n &= \{\bx\in\Xset_n: 1-\beta(n)/(c_n(\bx)/2) \leq c_{n,\nu}(\bx)/(c_n(\bx)/2)  \leq 1+\beta(n)/(c_n(\bx)/2)\} \\
&= \{\bx\in\Xset_n:  |c_{n,-1}(\bx) - c_{n,+1}(\bx)|  \leq 2\beta(n)\}, \nonumber
\end{align}
which again is defined through a function $\beta:\nset\to(0,\infty)$. The equivalence between the sets follows from the fact that $c_n(\mathbf{x}) = c_{n, -1}(\mathbf{x}) + c_{n, +1}(\mathbf{x})$. Under assumptions on the target such as those in \cite{zanella2019informed}, the normalizing constants $c_n(\bx)$ scale linearly with $n$ and below we show that lifted and MH samplers can be compared in terms of asymptotic variances when $\beta(n)=o(n)$, because in this case for states in $\tXset_n$, $\beta(n)/(c_n(\bx)/2)$ vanishes and the acceptance probabilities in the lifted sampler are close to 1, as those in MH. Notice that in the case of locally-balanced samplers, we cannot state explicitly what the interior of $\tXset_n$ is without specifying $\pi_n$. With the current level of generality, we cannot go beyond the definition presented in \autoref{sec:weak_peskun}, which in the framework of this section is $\itXset_n := \{\bx \in \tXset_n: q_{\bx}(\tXset_n^\mathsf{c}) = 0\}$.

As in the previous section, the analysis can be done by considering instead that the mass concentrates on states where the minimum between $c_{n,-1}(\bx)$ and $c_{n,+1}(\bx)$ is not too far from $c_n(\bx)/\kappa$ with $\kappa \geq 2$. In this case, $\overline{\omega} \leq 1$ and a function of $\kappa$, and the definition of the control subset and results are more complex. From the definition of $\tXset_n$ in \eqref{eq:subset_loc}, we are able to establish a result analogous to \autoref{lemma:omega_unif}.
\begin{Lemma}
\label{lemma:omega_loc}
Consider the definition of $\tXset_n$ in \eqref{eq:subset_loc} and let $R_n:=\{(\bx,\by)\in\tXset_n^2\,:\,\by\in\neigh(\bx)\}$. Assume that
$$
\inf_{(\bx,\by)\in R_n}g\left(\frac{\pi_n(\by)}{\pi_n(\bx)}\right)\geq m\,,\quad \tau_n:=\sup_{(\bx,\by)\in R_n}\frac{c_n(\bx)}{c_n(\by)}\to 1\,,\quad \beta(n)=o(n)\,,
$$
with $m$ independent of $n$. Then, for a large enough $n$, it holds that $\tP_{\rev, n}(\bx, \by)\geq \omega(n) \tP_{\MH, n}(\bx,\by)$, for all $(\bx,\by)\in\tXset_{n}^2$ with $\bx\neq \by$, where
\begin{align*}
\omega(n)=\left(1+\frac{\beta(n)}{nm/2}\right)^{-1}\left(\frac{1-2\beta(n)/nm}{1+2\tau_n\beta(n)/n m}\right) \rightarrow \overline{\omega} = 1\,.
\end{align*}
\end{Lemma}

Clearly, under the assumptions of \autoref{lemma:omega_loc} and that $\pi_n$ concentrates on $\itXset_n$, Assumptions \ref{ass:1} and \ref{ass:2} are satisfied and we can apply \autoref{thm:1} or \autoref{thm:2}  with $\overline{\omega}=1$, depending on whether the spectral gaps are bounded away from 0 or not. This gives an asymptotic ordering between $P_{\MH,n}$ and $P_{\rev,n}$, and thus between $P_{\MH,n}$ and $P_{\rho, n}$.

It is expected that lifted samplers only have an advantage when there is room for persistent movement, meaning that they can explore the state-space by using paths of considerable lengths.
The analysis conducted in the current section shows that lifted samplers using locally-balanced proposal distributions are expected to have an advantage when, additionally, the mass does not vary much from a directional neighbourhood to another on the subset on which $\pi_n$ concentrates. These samplers are expected to be comparable to their MH counterparts when, on the subset, the normalizing constants
$c_{n,-1}(\bx)$ and $c_{n,+1}(\bx)$ are bounded by $c_n(\bx)/\kappa\pm \beta(n)$ with $\kappa>2$.

\section{Simulation of a simple Ising model: A case study}\label{sec:simple_Ising}

The sampling method developed in \autoref{sec:lifted_sampler} and results presented in \autoref{sec:specific_samplers} are illustrated through several examples. In this section, we proceed by studying a simple Ising model that allows for an explicit definition of $\tXset_n$ and $\itXset_n$ when using locally-balanced samplers, and a verification of the assumptions of \autoref{thm:2}. As mentioned in \autoref{sec:organization}, we study in \autoref{sec:num_experiments} more complex problems (including the simulation of a Ising model which is more complex) for which an explicit definition of $\tXset_n$ and $\itXset_n$ and a veriﬁcation of the assumptions is beyond the scope of the manuscript.

The model that we study here is the following:
\begin{align}\label{eq:simple_Ising}
 \pi_n(\bx) = \frac{1}{Z_n} \exp\left(\sum_{i = 1}^n \alpha_i x_i\right), \quad \bx = (x_1, \ldots, x_n) \in \{-1, +1\}^n,
\end{align}
where $Z_n$ is the normalizing constant and $\alpha_i \in \re$. This model can be thought of as an Ising model with a single parameter $\boldsymbol\alpha_n := (\alpha_1, \ldots, \alpha_n)$ which is often referred to as the \textit{external field}. This parameter essentially tends to polarize each spin. The difference with classical Ising models like that in \autoref{sec:Ising} is that the model defined in \eqref{eq:simple_Ising} does not possess a spatial-correlation parameter. We can think of this model as being defined on a square lattice (with $x_1, \ldots, x_\eta$ being the values of the components on the first line, $x_{\eta + 1}, \ldots, x_{2\eta}$ being the values of the components on the second line, and so on), but by omitting the spatial correlation, the form on which the model is defined is actually not important. As mentioned in \autoref{sec:organization}, this simplified model can be seen as an approximation to the high temperature model. A common problem in statistical physics is to estimate the average \textit{magnetisation} of an Ising model, the magnetisation being defined as the mapping $\bx \mapsto \sum_{i=1}^n x_i$.

For the study conducted here, we consider the following simplified situation: $n$ is even, $\alpha_i = \pm c$ with $c$ a positive constant, and $|\{i: \alpha_i = -c\}| = |\{i: \alpha_i = c\}| = n / 2$, implying that the number of elements in the external field with the value $-c$ is the same as the number of elements with the value $c$. In our study, we focus on locally-balanced samplers and consider to simplify that $g$ is a monotonically increasing function, which is the case for the two functions mentioned in \autoref{sec:zanella}, namely $g(x) = \sqrt{x}$ and $g(x) = x / (1 + x)$.

In the simplified situation described above, we have that
\begin{align}\label{eqn:target_simple_Ising}
 \pi_n(\bx) \propto \exp\left(\sum_{\{i: \alpha_i x_i = +c\}} c + \sum_{\{i: \alpha_i x_i = -c\}} -c\right) &= \exp\left(c(|\{i: \alpha_i x_i = +c\}| - |\{i: \alpha_i x_i = -c\}|)\right) \nonumber \\
 &= \exp\left(c(n - 2|\{i: \alpha_i x_i = -c\}|)\right) \nonumber \\
 &\propto \exp\left(- 2c|\{i: \alpha_i x_i = -c\}|\right).
\end{align}
From the expression in \eqref{eqn:target_simple_Ising}, we easily deduce that the mode, denoted by $\bx^*$, is such that $|\{i: \alpha_i x_i^* = -c\}| = 0$, and that all the other values of $\pi_n$ are characterized by $|\{i: \alpha_i x_i = -c\}|$. Let us define $d(\bx) := |\{i: \alpha_i x_i = -c\}|\in \{0, \ldots, n\}$, which can be seen as a distance from the mode. We make the dependence on $n$ implicit to simplify. With the expression in \eqref{eqn:target_simple_Ising}, we have a better understanding of the model and how to compute probabilities of different events.

To motivate the use of our weak Peskun ordering for a comparison between the lifted and MH samplers, we provide a result about an inequality on the transition probabilities when considering the whole state-space.

\begin{Proposition}\label{prop:simple_ising_1}
  Within the framework described in this section, we have the following lower bound:
  \[
   P_{\rev, n}(\bx, \by)\geq (1 / 2) P_{\MH, n}(\bx,\by),
  \]
  for all $(\bx, \by) \in \Xset_n^2$ with $\bx \neq \by$, and for all $n$. Also, we have the following upper bound:
  \[
   P_{\rev, n}(\bx, \by)\leq (1 / 2) P_{\MH, n}(\bx,\by) \left(1 - \frac{g(\exp(2c)) - g(\exp(-2c))}{n g(\exp(-2c))}\right)^{-1},
  \]
  for certain $(\bx, \by) \in \Xset_n^2$ with $\bx \neq \by$, when $n > \exp(2c) - 1$. It is thus essentially not possible to obtain a better lower bound than that above.
 \end{Proposition}
 \autoref{prop:simple_ising_1} implies that the ordering based on Lemma 33 of \cite{andrieu2018uniform} is the following:
\begin{align}\label{eqn:asym_var_whole}
 \vara(f_n, P_{\rho,n}) \leq \vara(f_n, P_{\rev, n}) \leq 2 \, \vara(f_n, P_{\MH, n}) + 1,
\end{align}
for any $f_n \in \Ltwoz(\pi_n)$.

 We now turn to an analysis with an objective of applying our weak Peskun ordering. Our analysis allows to show that we can obtain tighter bounds on asymptotic variances when focusing on a subset of the state-space. The first step of such an analysis is to define $\tXset_n$ and understand which states belong to $\itXset_n$. We thus start with a result which will motivate a simple and explicit definition of $\tXset_n$ that we will connect to that in \eqref{eq:subset_loc}, and from which an explicit characterization of $\itXset_n$ will be easily deduced.

 \begin{Proposition}\label{prop:simple_ising_2}
  Within the framework described in this section, we have that for any $\bx$ and $n$,
  \[
  1-\frac{\left(g(\exp(2c)) + g(\exp(-2c))\right) d(\bx)/2}{c_n(\bx)/2} \leq \frac{c_{n,\nu}(\bx)}{c_n(\bx)/2}  \leq 1+\frac{\left(g(\exp(2c)) + g(\exp(-2c))\right) d(\bx)/2}{c_n(\bx)/2}.
  \]
 \end{Proposition}

 \autoref{prop:simple_ising_2} indicates that setting $\tXset_n := \{\bx: d(\bx) \leq \lfloor \phi(n) \rfloor\}$ with $\phi$ a monotonically increasing function allows to have a control on the ratio of normalizing constants of $q_{\bx}$ and $q_{\bx, \nu}$, and thus on the difference between $P_{\rev, n}$ and $P_{\MH, n}$. In particular, it allows to verify the inequality in \eqref{eq:subset_loc} with $\beta(n) = \left(g(\exp(2c)) + g(\exp(-2c))\right) \lfloor \phi(n) \rfloor / 2$, even though $\tXset_n$ is not defined as in \eqref{eq:subset_loc}. This is because
  \begin{align*}
  \tXset_n \subset \{\bx\in\Xset_n: 1-\beta(n)/(c_n(\bx)/2) \leq c_{n,\nu}(\bx)/(c_n(\bx)/2)  \leq 1+\beta(n)/(c_n(\bx)/2)\}.
 \end{align*}
 From our definition of $\tXset_n$, we can deduce that $\itXset_n = \{\bx: d(\bx) \leq \lfloor \phi(n) \rfloor - 1\}$. With those characterizations of $\tXset_n$ and $\itXset_n$, we easily understand which states belong to those subsets (comparatively to the definition of $\tXset_n$ in \eqref{eq:subset_loc} and that of $\itXset_n$ that follows from it), and thus how to compute probabilities like $1 - \pi_n(\itXset_n)$.

 Now that we have define the subsets $\{\tXset_n\}$, from which $\{\itXset_n\}$ are deduced, the next step is to verify whether the mass concentrates on $\itXset_n$ (\autoref{ass:2}). In our framework, $\itXset_n$ depends on the definition of $\phi(n)$. We present a result which indicates how to set $\phi(n)$ to obtain a mass concentration on $\itXset_n$ .

  \begin{Proposition}\label{prop:simple_ising_3}
  Within the framework described in this section, if $\lfloor\phi(n)\rfloor \leq n \frac{\exp(-2c)}{1 + \exp(-2c)}$, then $1 - \pi_n(\itXset_n)$ does not converge to 0. If $n > \lfloor\phi(n)\rfloor > n\frac{\exp(-2c)}{1 + \exp(-2c)}$ with $\lfloor\phi(n)\rfloor / n - \frac{\exp(-2c)}{1 + \exp(-2c)}$ converging towards a positive constant, then $1 - \pi_n(\itXset_n)$ converges to 0 at an exponential rate.
\end{Proposition}

\autoref{prop:simple_ising_3} indicates that setting $\phi(n) = o(n)$ does not allow for a mass concentration on $\itXset_n$. The result is thus somewhat negative as it prevents us to apply the results of \autoref{sec:specific_samplers}, in particular \autoref{lemma:omega_loc}, and forces us to exploit the structure of the current problem to establish a refined order between $P_{\rev, n}$ and $P_{\MH, n}$ (\autoref{ass:1}). \autoref{prop:simple_ising_3} indicates that, to obtain a mass concentration (at an exponential rate), we have to enlarge $\tXset_n$ and include states that are further away from the mode. We now establish a refined order between $P_{\rev, n}$ and $P_{\MH, n}$ on $\tXset_n$, when setting $\phi(n) = n\frac{\exp(-2c)}{1 + \exp(-2c)}(1 + \varepsilon) + 1$ with $\varepsilon > 0$ arbitrarily small (which is essentially the best choice of $\phi(n)$ that ensures that $\lfloor\phi(n)\rfloor > n\frac{\exp(-2c)}{1 + \exp(-2c)}$ with $\lfloor\phi(n)\rfloor / n - \frac{\exp(-2c)}{1 + \exp(-2c)}$ converging towards a positive constant).

\begin{Proposition}\label{prop:simple_ising_4}
  Within the framework described in this section and with $\phi(n) = n\frac{\exp(-2c)}{1 + \exp(-2c)}(1 + \varepsilon) + 1$, we have that
  \[
   P_{\rev, n}(\bx, \by)\geq \omega(n) P_{\MH, n}(\bx,\by),
  \]
  with
  \[
   \omega(n) \rightarrow \frac{1}{2} + \frac{1}{2} \frac{1 - 2  \frac{\exp(-2c)(1 + \varepsilon)}{1 + \exp(-2c)}}{1 + 2 \frac{1 + \varepsilon}{1 + \exp(-2c)}} = \overline{\omega},
  \]
  for all $(\bx, \by) \in \tXset_n^2$ with $\bx \neq \by$.
 \end{Proposition}

 \autoref{prop:simple_ising_4} highlights a dependence of $\overline{\omega}$ on the value of $c$: a smaller value of $c$ yields a larger $\phi(n)$ which results in a larger subset $\tXset_n$ and a possibility of more unbalanced ratios of normalizing constants of $q_{\bx}$ and $q_{\bx, \nu}$, and vice versa. We present in \autoref{fig:omega_bar} $\overline{\omega}$ as a function of $c$, without the factors $1+\varepsilon$ that can be made arbitrarily close to 1.

  \begin{figure}[ht]
  \centering
  \includegraphics[width=0.45\textwidth]{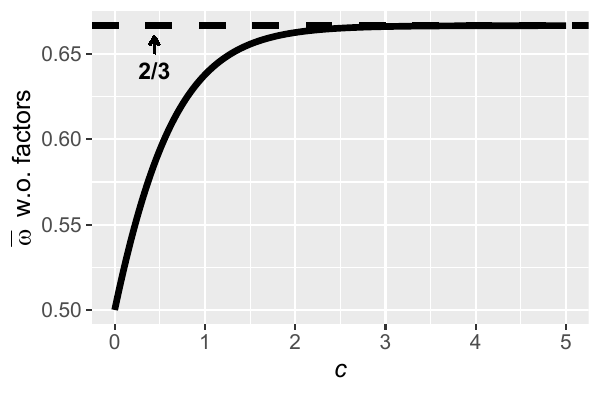}
  \caption{\small $\overline{\omega}$ as a function of $c$, without the factors $1+\varepsilon$.}\label{fig:omega_bar}
 \end{figure}
\normalsize

Provided that the spectral gaps of $P_{\rev, n}$, $\tP_{\rev, n}$, $P_{\MH, n}$ and $\tP_{\MH, n}$ do not decrease too quickly as $n$ increases (a result about that follows), \autoref{prop:simple_ising_3} together with \autoref{prop:simple_ising_4} ensure that \autoref{thm:2} can be applied for a class of functions, yielding
 \[
 \vara(f_n, P_{\rho,n}) \leq \vara(f_n, P_{\rev,n}) \leq \frac{1}{\overline{\omega} - \epsilon} \vara(f_n, P_{\MH,n}) + \frac{1}{2}\left(\frac{1}{\overline{\omega} - \eps} + \frac{1}{\overline{\omega}}\right) - 1  + \frac{\eps}{2},
\]
for any $\eps \in (0, \overline{\omega})$, provided that $n$ is large enough. When $c$ is large enough, we essentially have an upper bound of $ (3/2) \vara(f_n, P_{\MH,n}) + 1/2$, comparatively to what is obtained in \eqref{eqn:asym_var_whole}.

The advantage of this example is that it is standard, easy to understand, and simple enough to prove mass-concentration results, precise orderings between $P_{\rev, n}$ and $P_{\MH, n}$, and spectral-gap bounds. This simplicity follows from an independence between the components of $\bx$ and the steady decrease in mass by a factor of $\exp(-2c)$ as getting away from the mode, regardless of which components of $\bx$ are flipped and become misaligned with the external field. This steady, but relatively slow, decrease in mass forces us to set $\phi$ to be (essentially) proportional to $n$. This in turn leads to large subsets $\tXset_n$ and thus an improvement in terms of orderings between $P_{\rev, n}$ and $P_{\MH, n}$ which is not optimal, i.e., with $\overline{\omega} < 1$. In \autoref{sec:supp_mat} (\autoref{ex:path}), we \emph{construct} an example (thus an example that is less standard and simple) in which we are able to achieve $\overline{\omega} = 1$ by applying the results of \autoref{sec:specific_samplers}, in particular \autoref{lemma:omega_loc}.

We now present the last piece of evidence that \autoref{thm:2} can be applied. More specifically, we present a result about lower bounds on the spectral gaps of $P_{\rev, n}$ and $P_{\MH, n}$.

 \begin{Proposition}\label{prop:simple_ising_5}
  Within the framework described in this section, $P_{\rev, n}$ and $P_{\MH, n}$ have spectral gaps with lower bounds that decrease to 0 as $n$ increases at a rate of $n \log n$.
 \end{Proposition}

 While we do not prove a result about the spectral gaps of $\tP_{\rev, n}$ and $\tP_{\MH, n}$, there is no reason to believe that these decrease in another regime (for instance, with an exponential rate) given the definition of the sequence of subsets $\{\tXset_n\}$. For small values of $n$, we computed the spectral gaps through a spectral decomposition of $\tP_{\rev, n}$ and $\tP_{\MH, n}$ and the observed rate was polynomial.

  To summarize, the analysis in this section shows that \autoref{thm:2} can be applied for any $\delta > 0$ and $\gamma \in (0, \delta / (2 + \delta))$, when considering the class of functions $f_n$ with a $(2+\delta)$-norm that grows polynomially with $n$ or slower. An example of functions which satisfies this condition is the standardized version of the magnetisation $\bx \mapsto \sum_{i=1}^n x_i$, as indicated by \autoref{prop:simple_ising_6} below.

  \begin{Proposition}\label{prop:simple_ising_6}
  Let $f_n$ be the standardized version of the mapping $\bx \mapsto \sum_{i=1}^n x_i$. Within the framework described in this section, $\|f_n\|_{\pi_n, 4} \rightarrow 3^{1/4}$ as $n \rightarrow \infty$.
 \end{Proposition}

\section{Discussion}\label{sec:discussion}

In this paper, we have introduced a weaker version of the celebrated Peskun ordering \citep{peskun1973optimum} and have used it to analyse a class of lifted samplers designed to sample from distributions whose supports are partially-ordered discrete state-spaces. The weaker ordering does not require to establish a relationship between the Markov kernels on the whole state-space; it is only required to establish a relationship on a subset of the state-space, but the order between the asymptotic variances holds asymptotically, as a varying parameter grows without bound, as long as the mass concentrates on the subset (and provided that performance guarantees hold). This weaker requirement turned out to be useful to analyse some aspects of the lifted samplers and in particular how they compare to their MH counterparts. We have also shown that these lifted samplers can be straightforwardly implemented, at no additional computational cost and complexity, whenever a partial ordering on $\Xset_n$ can be established.

The main contribution of our analysis of the lifted samplers in \autoref{sec:specific_samplers} is to provide insights into the situations in which they are expected to outperform their MH counterparts, and also into those in which there is no guarantee. The analysis conducted shows that lifted samplers are expected to have an advantage when the mass does not vary much from a directional neighbourhood to another on the subset on which $\pi_n$ concentrates and when that subset allows the samplers to experience constant-momentum excursions. It is when they experience constant-momentum excursions of considerable lengths that the lifted samplers shine. While this point was reasonably well understood by the MCMC community, the merit of that part of our research presented in \autoref{sec:specific_samplers} has been to provide a rigorous analysis framework, which, de facto, can be used to study similar problems, perhaps some for which one does not have a clear intuition. Our analysis was conducted under a general framework, without focusing on specific statistical models or systems, explaining why we were not in a position to explicitly verify the assumptions of Theorems \ref{thm:1} and \ref{thm:2}. We dug deeper and provided a thorough analysis in a context of simulation of a simple Ising model in \autoref{sec:simple_Ising}, where the normalizing constants $c_{n, \nu}(\bx)$ and $c_n(\bx)$ have simple expressions, to take the study of lifted samplers one step further and to provide a concrete example of verification of the assumptions of \autoref{thm:2}.

One of the shortcomings of the application of our theoretical results to lifted samplers is that it does not give any quantitative measurement of the improvement offered by a lifted sampler over its MH counterpart when estimating $\pi_n f_n$, meaning that they are not such that $\vara(f_n,P_{\rho, n})\leq \omega_n \vara(f_n,P_{\text{MH},n})+ \text{error}$ for some $\omega_n > 1$. Indeed, our analysis only allows to establish an inequality, but in the case where (essentially) $\omega_n \leq 1$. This a consequence of the route we followed to compare the asymptotic variances of the lifted and MH samplers:
$$
\vara(f_n, P_{\rho, n})\leq \vara(f_n,P_{\text{rev}.,n})\leq \frac{1}{\omega_n} \, \vara(f_n,P_{\text{MH},n})+ \frac{1}{\omega _{n}} - 1 + \text{error}.
$$
In particular, no quantitative reduction factor is provided in the first inequality, which is expected given that this inequality holds in great generality (for any $f_n$ and any $\pi_n$). Given that $P_{\text{rev}.,n}$ and $P_{\text{MH},n}$ are, at best, similar and in fact, as mentioned in \autoref{sec:weak_peskun}, $\omega_n$ is usually larger than one, a way to have a quantitative variance improvement factor is to obtain a different inequality between $\vara(f_n, P_{\rho, n})$ and $\vara(f_n,P_{\text{rev}.,n})$ by leveraging an advantageous structure of the target distribution when it exists. We believe that this is possible, yet difficult, as the analysis needs to take into account the time duration of constant-momentum excursions conducted by the lifted sampler. This typically involves an analysis of $k$-step transition kernels with $k >1$ because it is only after $k$ transitions starting from a state $\bx$ that we start to see a significant difference between lifted samplers and their non-lifted and MH counterparts.

Our work can also be extended in another direction: the theoretical result can be generalized to general state-spaces and the lifted samplers can be applied in cases where there exist partial orders on these general state-spaces. However, our proofs implicitly assume that the Markov kernels are uniformly ergodic and it would be interesting to see how this assumption can be relaxed.

A methodological question which has been unaddressed in the paper is that of the choice of the partial order. If a specific state-space admits a partial order, it needs not be unique and its choice may significantly impact the sampler. Indeed, some choices may guarantee more than others those aforementioned constant-momentum excursions. If specifically interested in the estimation of $\pi_n f_n$ for a particular $f_n$, one could also design the partial order based on $f_n$, in the spirit of \citet{faizi2020efficient}.

Finally, in terms of applications of the theoretical work on the weak Peskun ordering, it would be interesting to consider the particular case of Bayesian models where a Bernstein von-Mises theorem holds. Comparing two MCMC methods sampling from the corresponding posterior distribution, our result suggests that one only needs to compare those samplers locally around a realization of a consistent parameter estimator. A question that naturally arises in this context is: is it possible to have a precise estimate of the sample size beyond which the approximate asymptotic-variance ordering holds? From a methodological standpoint this would motivate the design of samplers that are particularly efficient near the parameter estimate, perhaps at the expense of their behaviour in the tails of the distribution.

\section*{Acknowledgements}

The authors thank two anonymous referees for constructive comments that led to an improved manuscript.

\section*{Funding}

Philippe Gagnon acknowledges support from NSERC (Natural Sciences and Engineering Research Council of Canada) and FRQNT (Fonds de recherche du Québec -- Nature et technologies). Florian Maire acknowledges support from NSERC.

\bibliographystyle{imsart-nameyear}
\bibliography{reference}

\appendix

\section{Numerical experiments}\label{sec:num_experiments}

In this section, we conduct numerical experiments that allow to corroborate the theoretical findings presented in \autoref{sec:specific_samplers} about the lifted samplers. We focus on findings about locally-balanced samplers. We first consider in \autoref{sec:Ising} the simulation of an Ising model and use this as a toy example for which we can control the dimension and the roughness of the target. We show that specific combinations of these parameters are favourable for lifted samplers, in the sense that the mass in directional neighbourhoods varies smoothly over a subset of considerable size, suggesting the existence of subsets $\tXset_n$ defined as in \eqref{eq:subset_loc} and interiors on which the mass concentrates. For these combinations of parameters, lifted samplers outperform MH ones. Other combinations are unfavourable, and the opposite happens. The findings are consistent with those of Sections \ref{sec:specific_samplers} and \ref{sec:simple_Ising}. In \autoref{sec:variable_selection}, a real variable-selection problem yields a target which is favourable for lifted samplers (in the same sense as above), and again lifted samplers outperform MH ones.

\subsection{Ising model}\label{sec:Ising}

Let us consider the two-dimensional Ising model. For this model, the state-space $(V_\eta, E_\eta)$ is a $\eta \times \eta$ square lattice regarded here as a square matrix in which each element takes either the value $-1$ or $+1$. We write each state as a vector as before: $\mathbf{x} = (x_1, \ldots, x_n)$, where $n = \eta^2$. The states can be encoded as follows: the values of the components on the first line are $x_1,  \ldots, x_\eta$, those on the second line $x_{\eta + 1}, \ldots, x_{2\eta}$, and so on. The PMF is given by
\[
 \pi(\mathbf{x}) = \frac{1}{Z} \exp\left(\sum_{i} \alpha_i x_i + \lambda \sum_{\langle i j \rangle} x_i x_j\right),
\]
where $\alpha_1, \ldots, \alpha_n \in \re$ and $\lambda > 0$ are fixed parameters, $Z$ is the normalizing constant and the notation $\langle i j \rangle$ indicates that sites i and j are nearest neighbours. The notion of neighbourhood on $(V_\eta, E_\eta)$ should not be confused with that on $\Xset_n$ on which the samplers rely. The neighbourhood of a site $i\in V_\eta$ comprises, when they exist, its North-South-East-West neighbours on the lattice. Note that we make the dependence of the target on the parameters and $n$ implicit to simplify.

The role of the parameters in this Ising model are worth being explained. The parameter $\lambda$ is a spatial correlation parameter: the larger it gets, the larger are the chances that two neighbouring nodes share the same spin state. Realizations from such models when $\lambda$ is large  are thus likely to be lattices featuring large patches of identical spin states. The parameter $\boldsymbol\alpha := (\alpha_1, \ldots, \alpha_n)$, often referred to as the external field, essentially tends to polarize each spin, regardless its neighbours. In particular, when $\alpha_i$ decreases, $x_i$ has an increasing tendency to align with a negative spin, that is $x_i=-1$. If $|\alpha_i| \gg \lambda$ for all $i$, the dependency structure in the lattice is negligible and thus spins tend to align with the external field. Conversely, if $\lambda \gg |\alpha_i|$ for all $i$, spins in a vicinity tend to align with one another.

We first consider a base target distribution for which $n = 50^2$, the spatial correlation is moderate and more precisely $ \lambda = 0.5$, and which has the external field presented in \autoref{fig_base_target_Ising}.
\begin{figure}[ht]
\begin{center}
    \includegraphics[width=0.44\textwidth]{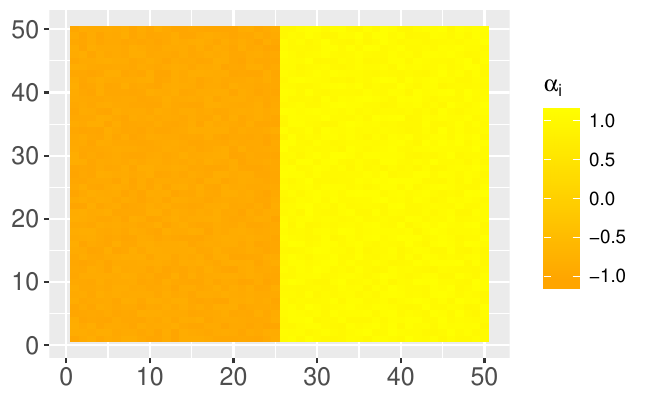}
\end{center}
  \vspace{-2mm}
\caption{\small External field of the base target.}\label{fig_base_target_Ising}
\end{figure}
\normalsize
We generated the $\alpha_i$ independently as follows:  $\alpha_i = -\mu + \epsilon_i$ if the column index is smaller than or equal to $\ell := \lfloor \eta / 2 \rfloor$ and $\alpha_i = \mu + \epsilon_i$ otherwise, where $\mu = 1$, the $\epsilon_i$ are independent uniform random variables on the interval $(-0.1, +0.1)$ and $\lfloor\, \cdot \,\rfloor$ is the floor function. In this setup, while the mild external field tends to push spins on the left-hand side (LHS) of the lattice to $-1$ and those on the right-hand side (RHS) to $+1$, the moderate spatial correlation tends to make likely lattices with $-1$ on the RHS near the centre and $+1$ on the LHS near the centre. This makes the target moderately rough, in the sense that it concentrates on a subset of the state-space with directional neighbourhoods on the subset that have a smoothly varying mass. This subset can be thought of as the subset $\tXset_n$ which is the central ingredient of Theorems \ref{thm:1} and \ref{thm:2}. The characteristic of the subset suggests that $\tXset_n$ satisfies the definition in \eqref{eq:subset_loc}, implying that such a base target represents a favourable scenario for lifted samplers with locally-balanced proposal distributions as described in \autoref{sec:zanella}. We will notice that it is indeed a favourable scenario and observe what happens when modifying target-parameter values.

We now describe the simulation study.
\begin{itemize}
\itemsep 0mm
\item While keeping the other parameters fixed, we first gradually increase $\eta$ from 50 to 500 to observe the impact of dealing with larger systems, for targets that are moderately rough. This will thus lead to longer paths along which the state-space can be explored, which is again favourable for lifted samplers. The numerical experiment will allow to measure an increasing difference in performance between lifted samplers and MH ones, which is not possible with results such as Theorems \ref{thm:1} and \ref{thm:2}.
\item Next, we gradually increase the value of $\mu$ from 1 to 3, while keeping the other parameters fixed (with $\eta = 50$). This increases the contrast in \autoref{fig_base_target_Ising}. When $\mu$ increases, there is less and less chance to observe negative (positive) spins on the RHS (LHS), even near the centre, thus making the target rougher and concentrated on fewer configurations.  In the limit, the set of possible lattices shrinks to the one lattice dictated by the external field with $-1$'s on the LHS and $+1$'s on the RHS. This suggests that in extreme cases, it becomes difficult to define a subset $\tXset_n$ as in \eqref{eq:subset_loc}, while keeping the concentration level reasonable, with an interior on which the mass concentrates because such a $\tXset_n$ is too small implying that its interior is too small as well (or even non-existent), in turn suggesting that the assumptions of \autoref{thm:1} or \autoref{thm:2} do not hold. In the experiment, when the value of $\mu$ is beyond a threshold, MH samplers become more efficient than lifted ones.
\end{itemize}
One could vary $\lambda$ and $\ell$ as well. Varying $\lambda$ also makes the target rougher and concentrated on fewer configurations. We thus do not do it to avoid redundancy. Varying $\ell$ is expected to have a more important impact on the uniform lifted sampler than the other samplers because it modifies the location of the area where the mass concentrates. We do not present the associated results because the graph is uninteresting: the performance is essentially constant for the locally-balanced samplers and that of the uniform ones is so low that we do not see the ESS vary.

We present the simulation results in \autoref{fig_results_Ising} for \autoref{algo1} with uniform and locally-balanced proposal distributions, and their MH counterparts. Locally-balanced samplers use the Barker proposal distribution with $g(x) = x / (1 + x)$. For a simulation study such as that conducted here, it would be simply too long to obtain the results for \autoref{algo2} with $\rho_{\nu}^*$ \eqref{eqn_rho_star}. The results are based on 1,000 independent runs of 100,000 iterations for each algorithm and each value of $\mu$ and $\eta$, with burn-ins of 10,000. For each run, an ESS per iteration is computed for $f(\mathbf{x}, -1) = f(\mathbf{x}, +1) = \sum_i x_i$ and then the results are averaged out. This function is proportional to what is called \textit{magnetisation} in a Ising-model framework. Monitoring such a statistic is relevant as a quicker variation of its value (leading to a higher ESS) indicates that the whole state-space is explored quicker.

For the base target (represented by the starting points on the left of the lines in \autoref{fig_results_Ising}), the mass is, as mentioned, concentrated on a subset of many configurations with, on the subset, a mass that does not vary too much from a directional neighbourhood to another. The locally-balanced lifted sampler takes advantage of this and induces persistent movement on the subset: it is approximately 7 times more efficient than its MH counterpart. The gap widens as $\eta$ increases (\autoref{fig_results_Ising} (a)), a consequence of longer paths that the locally-balanced lifted sampler efficiently follows; it is approximately 20 and 70 times more efficient when $\eta$ is 3.2 and 10 times larger (i.e.\ when $n$ is 10 and 100 times larger), respectively. We evaluated that the ratio of ESSs increases linearly with $\eta$, indicating that the locally-balanced lifted sampler scales better than its MH counterpart. The samplers with uniform proposal distributions perform poorly (the lines are on top of each other).

As $\mu$ increases (\autoref{fig_results_Ising} (b)), the target becomes rougher and concentrated on fewer configurations. When the roughness and concentration level are too severe the performance of the locally-balanced lifted sampler stagnates, whereas that of its MH counterpart continues to improve. When the roughness and concentration level are too severe and the samplers are at the mode, the MH sampler has an advantage. When the chain leaves the mode, it always has the possibility to return to it the following iteration. The chain simulated by the lifted sampler cannot because it is forced to try continuing in the same direction. Also, when the mass is concentrated on few configurations, it leaves not much room for persistent movement for the lifted sampler, and it thus loses its advantage.

\begin{figure}[ht]
\centering
$\begin{array}{cc}
    \hspace{-0mm} \includegraphics[width=0.50\textwidth]{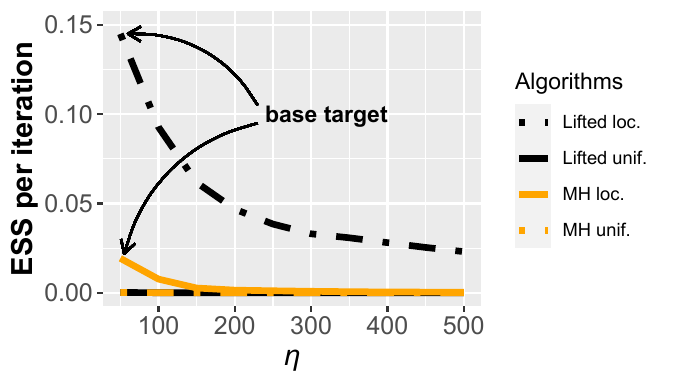} & \hspace{-2mm} \includegraphics[width=0.50\textwidth]{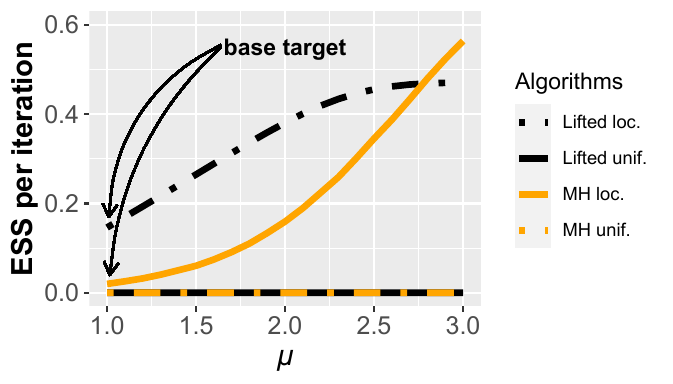} \cr
    \hspace{-11mm} \textbf{(a)} & \hspace{-12mm} \textbf{(b)}
\end{array}$
  \vspace{-2mm}
\caption{\small ESS per iteration of $f(\mathbf{x}, -1) = f(\mathbf{x}, +1) = \sum_i x_i$ for \autoref{algo1} with uniform and locally-balanced proposal distributions and their MH counterparts when: (a) $\eta$ increases from 50 to 500 and the other parameters are kept fixed ($\mu = 1$ and $\lambda = 0.5$); (b) $\mu$ increases from 1 to 3 and the other parameters are kept fixed ($\eta = 50$ and $\lambda = 0.5$).}\label{fig_results_Ising}
\end{figure}
\normalsize

\subsection{Variable selection: US crime data}\label{sec:variable_selection}

In this section, we contrast the performance of the lifted samplers with that of their MH counterparts when applied to solve a real Bayesian variable-selection problem. The data are for a study of crime rate in the United States in 1960. They were aggregated by state and were from 47 states. They were first presented in \cite{erhlich1973participation} and then expanded and corrected in \cite{vandaele1978participation}. These authors were in particular interested in studying the connection between crime rate and 15 covariates (some were added by \cite{vandaele1978participation}) such as percentage of males of age between 14 and 23 and mean years of schooling in a given state. They were analysed in several statistics papers, for instance in \cite{raftery1997bayesian} in a context of model averaging, and are available in the R package \texttt{MASS}.

The data are modelled using a linear regression with normal errors. Here we set the prior distribution of the regression coefficients and scaling of the errors to be, conditionally on a model, the non-informative Jeffreys prior. It can be shown (analogously to in \cite{gagnon2017PCR} in a context of principal component regression) that a simple modification to the uniform prior on the model indicator, represented here by $\mathbf{X}$, yields a consistent model selection procedure, thus effectively preventing the Jeffreys--Lindley \citep{lindley1957paradox, jeffreys1967prob} paradox from arising. The likelihood function and prior density on the parameters allows for the latter to be integrated out. It is thus possible to evaluate the exact marginal posterior probability of any of the $2^{15} = $ 32,768 models, up to a normalizing constant. We are consequently able to implement the MH sampler with the Barker locally-balanced proposal distribution of \cite{zanella2019informed}  and its lifted counterparts, namely \autoref{algo1} and \autoref{algo2} with $\rho_{\nu}^*$ \eqref{eqn_rho_star}, to sample from $\pi$, which is, in this context, a posterior model distribution. In the previous statistical studies (such as in \cite{raftery1997bayesian}), it was noticed that for many models, the mass varies smoothly; the mass in fact concentrates on the resulting subset of the state-space and does not vary too much from a directional neighbourhood to another on the subset. As with the Ising-model example, this suggests the existence of a subset $\tXset_n$ defined as in \eqref{eq:subset_loc} with a significant size and an order on the asymptotic variances of some functions between lifted samplers and MH ones. Lifted samplers indeed outperform MH ones in this example. In particular, the locally-balanced lifted chains exhibit persistent movement, as seen in \autoref{fig_traces}. We do not show the performance of the uniform samplers because, as in the previous section, it is very poor.

The performances of the algorithms are summarized in \autoref{fig_results_var_sel}. The results are based on 1,000 independent runs of 10,000 iterations for each algorithm, with burn-ins of 1,000. Each run is started from a distribution which approximates the target. On average, \autoref{algo1} and \autoref{algo2} with $\rho_{\nu}^*$ are $2.7$ and $3.3$ times more efficient than their MH counterpart, respectively. The benefits of persistent movement thus compensate for a decrease in acceptance rates; the rate indeed decreases from 0.92 for the MH sampler to 0.71 for \autoref{algo1} and \autoref{algo2} with $\rho_{\nu}^*$ \eqref{eqn_rho_star}. This highlights again the difference in stability of neighbourhood mass versus \emph{directional} neighbourhood mass (recall the difference in the acceptance ratios, \eqref{eqn_acc_prob_inf_MH} and \eqref{eqn_acc_prob_inf_lifted}).

\begin{figure}[ht]
\begin{center}
    \includegraphics[width=0.41\textwidth]{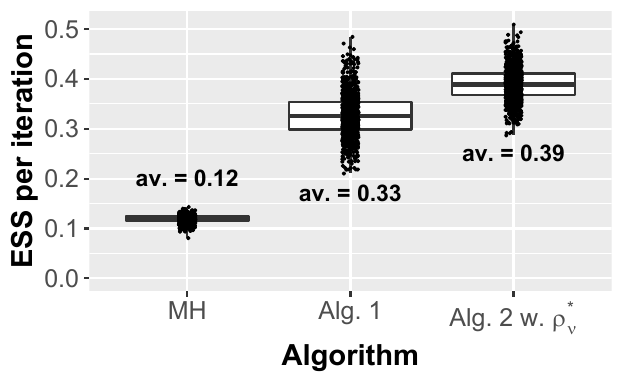}
\end{center}
  \vspace{-2mm}
\caption{\small ESS per iteration for $f(\mathbf{x}, -1) = f(\mathbf{x}, +1) = \sum_i x_i$ of 1,000 independent runs for the MH sampler with the Barker locally-balanced proposal distribution and its lifted counterparts (\autoref{algo1} and \autoref{algo2} with $\rho_{\nu}^*$).}\label{fig_results_var_sel}
\end{figure}
\normalsize

\section{Lifted trans-dimensional sampler}\label{sec:trans-dimensional}

In this section, we introduce a generic sampler that can be used for model selection/averaging in situations where it is not possible to integrate out the parameters, contrarily to the linear regression with normal errors and suitable priors (like in \autoref{sec:variable_selection}). Examples of such situations include analyses based on linear regression with super heavy-tailed errors ensuring \textit{whole robustness} \citep{gagnon2018regression, gagnon2017PCR} and generalized linear models and generalized linear mixed models \citep{forster2012reversible}.

More precisely, in this section, we introduce a trans-dimensional version of \autoref{algo1} which thus represents a non-reversible counterpart to the popular reversible jump (RJ) algorithm introduced by \cite{green1995reversible}. In the same way that \autoref{algo1} can be seen as a modification of a MH algorithm, the non-reversible jump (NRJ) algorithm is constructed from the RJ algorithm. To present our lifted trans-dimensional sampler, it is thus convenient to first provide an overview of the RJ one. A lifted trans-dimensional sampler has been recently introduced in \cite{gagnon2019NRJ}, but it can only be applied when the models can be rearranged in a sequence of nested models, meaning that model 1 is nested in model 2 which is nested in model 3, and so on; in other words, when a total order exists. Only a partial order is sufficient to apply the NRJ proposed here.

In a trans-dimensional framework, we consider that $\Xset$ is a model space and $\mathbf{X}$ a model indicator. The latter indicates, for instance, through a vector of 0's and 1's which covariates are included in the model employed in variable-selection contexts as in \autoref{sec:variable_selection}. In the following, we consider that a neighbourhood structure $\{\neigh(\mathbf{x})\,:\,\bx\in\Xset\}$ is given.  The parameters of a given model $\mathbf{x}$ are denoted by $\boldsymbol\theta_{\mathbf{x}} \in \boldsymbol\Theta_{\mathbf{x}}$. Trans-dimensional algorithms sample from a target distribution $\pi$ defined on a union of sets $\cup_{\mathbf{x} \in \Xset} \{\mathbf{x}\} \times \boldsymbol\Theta_{\mathbf{x}}$, which corresponds in Bayesian statistics to the joint posterior distribution of the model indicator $\mathbf{X}$ and the parameters of model $\mathbf{X}$, that is $\boldsymbol\theta_{\mathbf{X}}$. Such a posterior distribution allows to jointly infer about $(\mathbf{X}, \boldsymbol\theta_{\mathbf{X}})$, or in other words, simultaneously achieve model selection/averaging and parameter estimation. In this section, we assume for simplicity that the parameters of all models are continuous random variables.

We now outline an iteration of a RJ algorithm. Consider that the current state of the Markov chain is given by $(\bx, \btheta_{\bx})$.
 \begin{itemize}

  \item[1.] Sample $u_c \sim \mathcal{U}[0, 1]$.

  \item[2.(a)\hspace{-4.5mm}] \hspace{4mm} If $u_c \leq \tau$, where $0 \leq \tau \leq 1$, attempt a \textit{parameter update}, meaning an update of the parameters of the current model, using a MCMC kernel of invariant distribution $\pi(\, \cdot \mid \mathbf{x})$ while keeping the current value of the model indicator $\mathbf{x}$ fixed.

  \item[2.(b)\hspace{-4.5mm}] \hspace{4.75mm} If $u_c > \tau$, attempt a \textit{model switch}. Sample $\mathbf{y} \sim q_{\mathbf{x}}$ and $u_a \sim \mathcal{U}[0, 1]$, where $q_{\bx}$ is a PMF with support $\neigh(\bx)$. Next, sample $\mathbf{u}_{\bx \mapsto \by} \sim q_{\bx \mapsto \by}$ and compute $\mathcal{D}_{\bx\mapsto \by}(\boldsymbol\theta_{\mathbf{x}}, \mathbf{u}_{\bx \mapsto \by}) =: (\boldsymbol\theta_{\mathbf{y}}', \mathbf{u}_{\by \mapsto \bx})$, where $q_{\bx \mapsto \by}$ is used to denote both the distribution and the probability density function, $\mathcal{D}_{\bx\mapsto \by}$ is a diffeomorphism and $\boldsymbol\theta_{\mathbf{y}}'$ is the proposal for the parameter values of model $\mathbf{y}$. Set the next state of the chain to $(\mathbf{y}, \boldsymbol\theta'_{\mathbf{y}})$ if
\begin{align*}
  u_a \leq \alpha_{\text{RJ}}((\bx, \btheta_{\bx}), (\by, \btheta_{\by}')) := 1 \wedge \frac{q_{\mathbf{y}}(\mathbf{x})}{q_{\mathbf{x}}(\mathbf{y})} \, r((\mathbf{x}, \boldsymbol\theta_{\mathbf{x}}), (\mathbf{y}, \boldsymbol\theta'_{\mathbf{y}})),
\end{align*}
where
\[
 r((\mathbf{x}, \boldsymbol\theta_{\mathbf{x}}), (\mathbf{y}, \boldsymbol\theta'_{\mathbf{y}})) := \frac{\pi(\by, \boldsymbol\theta'_{\mathbf{y}}) \, q_{\by \mapsto \bx}(\mathbf{u}_{\by \mapsto \bx})}{\pi(\bx, \btheta_{\bx}) \, q_{\bx \mapsto \by}(\mathbf{u}_{\bx \mapsto \by}) \, |J_{\mathcal{D}_{\bx \mapsto \by}}(\btheta_{\bx}, \mathbf{u}_{\bx \mapsto \by})|^{-1}},
\]
  and $|J_{\mathcal{D}_{\bx \mapsto \by}}(\btheta_{\bx}, \mathbf{u}_{\bx \mapsto \by})|$ is the absolute value of the determinant of the Jacobian matrix of the function $\mathcal{D}_{\bx \mapsto \by}$; the dependence of the functions $\alpha_{\text{RJ}}$ and $r$ on $\mathbf{u}_{\bx \mapsto \by}$ and $\mathbf{u}_{\by \mapsto \bx}$ is made implicit to simplify. If $u_a > \alpha_{\text{RJ}}((\bx, \btheta_{\bx}), (\by, \btheta_{\by}'))$, set the next state of the chain to $ (\mathbf{x}, \boldsymbol\theta_{\mathbf{x}})$.

  \item[3.] Go to Step 1.
 \end{itemize}

The notation $\bx \mapsto \by$ in subscript is used to highlight a dependence on the model transition that is proposed, which is from model $\bx$ to model $\by$. Recall that a diffeomorphism is a differentiable map having a differentiable inverse. A simple example of a mapping $\mathcal{D}_{\bx \mapsto \by}$ is one where the current parameter value $\boldsymbol\theta_{\mathbf{x}}$ is not involved in the parameter-proposal scheme: $\boldsymbol\theta_{\mathbf{y}}' = \mathbf{u}_{\bx \mapsto \by}$ and $\mathbf{u}_{\by \mapsto \bx} = \boldsymbol\theta_{\mathbf{x}}$, implying that $|J_{\mathcal{D}_{\bx \mapsto \by}}(\btheta_{\bx}, \mathbf{u}_{\bx \mapsto \by})| = 1$.

In the trans-dimensional framework presented above, $\bx \notin \neigh(\bx)$, as before, and $q_{\bx}$ is used conditionally on the fact that a model switch is proposed. The probability of proposing a model switch is $1 - \tau$, $\tau$ representing the probability of proposing a parameter update. In trans-dimensional samplers, the probability of proposing a parameter update is typically allowed to depend on the current state and is incorporated in $q_{\bx}$. By contrast, it is considered constant and not incorporated in $q_{\bx}$ in this framework so as to guarantee the correctness of the non-reversible counterpart of the RJ sampler.

We now consider that a partial order $\R$ has been specified on $\Xset$. In the lifted framework, the state-space is extended to include a direction variable $\nu \in \{-1, +1\}$ to guide the model indicator $\bX$. The state-space and target become $\cup_{\mathbf{x} \in \Xset} \{\mathbf{x}\} \times \boldsymbol\Theta_{\mathbf{x}} \times \{-1, +1\}$ and $\pi \otimes \mathcal{U}\{-1, +1\}$, respectively. Apart from the inclusion of $\nu$ in the algorithm process, there is only one major change made to RJ to yield NRJ: given a current state of $(\bx, \boldsymbol\theta_{\mathbf{x}}, \nu)$ and that a model switch has been proposed, a model $\by$ is proposed using a PMF $q_{\bx, \nu}$ with support $\neigh_\nu(\bx)$, instead of $q_{\bx}$ with support $\neigh(\bx)$. The directional neighbourhoods are defined as before:  $\neigh_{+1}(\mathbf{x}) := \{\mathbf{y} \in \neigh(\mathbf{x}) : \mathbf{x} \prec \mathbf{y}\} \subset \neigh(\mathbf{x})$ and $\neigh_{-1}(\mathbf{x}) := \{\mathbf{y} \in \neigh(\mathbf{x}) : \mathbf{y} \prec \mathbf{x}\} \subset \neigh(\mathbf{x})$. The rest of NRJ is essentially the same as RJ. Given that $q_{\bx, \nu}$ is often defined analogously to $q_{\bx}$, the implementation is thus straightforward for a RJ user that already specified the functions $q_{\bx}$, $\mathcal{D}_{\bx\mapsto \by}$ and $q_{\bx \mapsto \by}$, provided that a partial order can be established on $\Xset$. For a specification of these functions, we refer users to \cite{gagnon2019RJ}, in which a generic procedure yielding fully informed and efficient RJ is presented.

The NRJ algorithm is now presented in \autoref{algo_NRJ} and \autoref{prop_invariance_NRJ} below establishes its correctness. The proof of \autoref{prop_invariance_NRJ} establishes that any valid scheme used for parameter proposals during model switches in the RJ framework, such as those of \cite{karagiannis2013annealed} and \cite{andrieu2018utility}, are also valid in the non-reversible framework.

\begin{algorithm}[ht]
\caption{A lifted trans-dimensional sampler for partially-ordered model spaces} \label{algo_NRJ}
 \begin{itemize}

  \item[1.] Sample $u_c \sim \mathcal{U}[0, 1]$.

  \item[2.(a)\hspace{-4.5mm}] \hspace{4mm} If $u_c \leq \tau$, attempt a parameter update using a MCMC kernel of invariant distribution $\pi(\, \cdot \mid \mathbf{x})$ while keeping the current value of the model indicator $\mathbf{x}$ and direction $\nu$ fixed.

  \item[2.(b)\hspace{-4.5mm}] \hspace{4.75mm} If $u_c > \tau$, attempt a model switch. Sample $\mathbf{y} \sim q_{\mathbf{x},\nu}$, $\mathbf{u}_{\bx \mapsto \by} \sim q_{\bx \mapsto \by}$ and $u_a \sim \mathcal{U}[0, 1]$. Next, compute $\mathcal{D}_{\bx\mapsto \by}(\boldsymbol\theta_{\mathbf{x}}, \mathbf{u}_{\bx \mapsto \by}) = (\boldsymbol\theta_{\mathbf{y}}', \mathbf{u}_{\by \mapsto \bx})$. If
  \begin{align*}\label{eqn_acc_prob_NRJ}
  	\hspace{-3mm} u_a \leq \alpha_{\text{NRJ}}((\mathbf{x}, \boldsymbol\theta_{\mathbf{x}}), (\mathbf{y}, \boldsymbol\theta'_{\mathbf{y}})) := 1 \wedge \frac{q_{\mathbf{y}, -\nu}(\mathbf{x})}{q_{\mathbf{x}, \nu}(\mathbf{y})} \, r((\mathbf{x}, \boldsymbol\theta_{\mathbf{x}}), (\mathbf{y}, \boldsymbol\theta'_{\mathbf{y}})),
  \end{align*}
  set the next state of the chain to $(\mathbf{y}, \boldsymbol\theta'_{\mathbf{y}}, \nu)$. Otherwise, set it to $ (\mathbf{x}, \boldsymbol\theta_{\mathbf{x}}, -\nu)$.

  \item[3.] Go to Step 1.
 \end{itemize}
\end{algorithm}

\begin{Proposition}\label{prop_invariance_NRJ}
 The transition kernel of the Markov chain $\{(\mathbf{X}, \boldsymbol\theta_{\mathbf{X}}, \nu)_k\}$ simulated by \autoref{algo_NRJ} admits $\pi\otimes \mathcal{U}\{-1, 1\}$ as invariant distribution.
\end{Proposition}

In \cite{gagnon2019RJ}, the proposed procedure to specify the functions $q_{\bx}$, $\mathcal{D}_{\bx\mapsto \by}$ and $q_{\bx \mapsto \by}$ is proved to produce a RJ which asymptotically approaches an ideal one which is able to sample $\boldsymbol\theta'_{\mathbf{y}}$ from $\pi(\, \cdot \mid \mathbf{y})$ (the correct conditional distribution) and which sets $q_{\mathbf{x}}$ to locally-balanced distributions (because it has access to the exact ratios of marginal probabilities $\pi(\mathbf{y}) / \pi(\mathbf{x})$), as the sample size goes to infinity in a Bayesian statistics context. The analogous conclusions hold for NRJ, and thus $q_{\mathbf{x}, \nu}$ can be set to be asymptotically locally-balanced following the analogous procedure to that in \cite{gagnon2019RJ}. In the limit, the marginal process $\{(\mathbf{X}, \nu)_k\}$ is the same (if we consider only iterations for which model switches are proposed) as that simulated by \autoref{algo1}. All conclusions previously drawn about the state-space exploration efficiency of \autoref{algo1} compared to its MH counterpart thus hold (at least approximatively) for \autoref{algo_NRJ}, but when compared with its RJ counterpart. In particular, if we were to analyse the same data as in \autoref{sec:variable_selection}, but using the super heavy-tailed regression of \cite{gagnon2018regression} for robust inference and outlier detection, it is likely that the algorithm performance results would be the similar. Indeed, \cite{raftery1997bayesian} verified that nothing points towards a gross violation of the assumptions underlying normal linear regression and the robust method is designed for leading to similar results in the absence of outliers. We thus omit further analysis of \autoref{algo_NRJ} and we do not illustrate how it performs for brevity. We nevertheless mention that, within the trans-dimensional framework, $r((\mathbf{x}, \boldsymbol\theta_{\mathbf{x}}), (\mathbf{y}, \boldsymbol\theta'_{\mathbf{y}}))$ can be seen as an estimator of $\pi(\mathbf{y}) / \pi(\mathbf{x})$ and it is important that this estimator has a low variance in the lifted framework as persistent movement may be interrupted otherwise because significant noise fluctuations may lead to high rejection rates, as shown in \cite{gagnon2019NRJ}. The methods of \cite{karagiannis2013annealed} and \cite{andrieu2018utility} can be used to produce an estimator involved in the acceptance probability like $r((\mathbf{x}, \boldsymbol\theta_{\mathbf{x}}), (\mathbf{y}, \boldsymbol\theta'_{\mathbf{y}}))$, but with a reduced variability.

\section{Proofs of theoretical results and useful lemmas}\label{sec:proofs}

We now present the proofs of all theoretical results in the same order as the results appeared in the paper. We beforehand present and prove three lemmas which are central to the proofs of Theorems \ref{thm:1} and \ref{thm:2}. In the proofs, we will sometimes use a subscript in $\esp$ to make clear with respect to which distribution the expectation is computed. We will do the same with $\proba$.

To prove Theorems \ref{thm:1} and \ref{thm:2}, we actually prove an order on the asymptotic variances of the lazy versions of $P_{1,n}$ and $P_{2,n}$ and then use a result about the relationship between the asymptotic variance associated to a kernel and that of the lazy version to conclude. The reason for this is that we want to use a result that we now present. Let $P$ be a $\pi$-reversible Markov kernel. It operates a contraction on $\Ltwoz(\pi)$ in the sense that for all $f \in \Ltwoz(\pi)$ and all $k\in \nset$,
\begin{equation}\label{eq:contraction}
  \|P^kf\|_\pi\leq (1-\lambda)^k,
\end{equation}
where $\lambda$ is the associated \emph{absolute spectral gap} \citep[Section 12.2]{levin2017markov}. For the lazy version of a Markov kernel, the \emph{absolute spectral gap} corresponds to the \emph{right spectral gap}. Also, it is possible to establish lower bounds on the right spectral gaps of the lazy versions of $P_{1,n}, \tP_{1,n}, P_{2,n}$ and $\tP_{2,n}$ using the order between a Markov kernel and its lazy version \citep[Theorem 2]{zanella2019informed}. We thus now proceed with results about connections between the asymptotic variances of the lazy versions of $P_{1,n}, \tP_{1,n}, P_{2,n}$ and $\tP_{2,n}$. Note that they are valid for the original kernels as well. Before proceeding, we explicitly define what is the lazy version of a Markov kernel. Let $P$ be a $\pi$-reversible Markov kernel. Its lazy version, denoted by $P^{(\text{L})}$, is defined as $P^{(\text{L})} := (P + I) / 2$.

The three lemmas that we now present and prove hold for any fixed $n$. To simplify the presentation, we thus make implicit the dependence on this parameter of the target distribution, state-space, and so on. In particular, we write $p$ for $p(n):=\pi_n(\tXset_n)$. We introduce some notation that are required for the presentation of the lemmas. We define four Markov chains $\{\bX_k\}, \{\tbX_k\}, \{\bY_k\}$ and $\{\tbY_k\}$ with Markov kernels $P_{1}^{(L)}, \tP_{1}^{(L)}, P_{2}^{(L)}$ and $\tP_{2}^{(L)}$, respectively, started in stationarity. Let $\varrho \in \nset$. We define
\[
 A_\varrho := \bigcap_{k < \varrho}\{\bX_k \in \itXset\},
\]
\[
 \tilde{A}_\varrho := \bigcap_{k < \varrho}\{\tbX_k \in \itXset\},
\]
\[
 B_\varrho := \bigcap_{k < \varrho}\{\bY_k \in \itXset\},
\]
and
\[
 \tilde{B}_\varrho := \bigcap_{k < \varrho}\{\tbY_k \in \itXset\}.
\]
Note that the asymptotic variance can be written for a test-function $f$ as
\[
 \vara(f, P_1^{(L)}) = \var[f(\bX_0)] + 2 \sum_{k = 1}^\infty \cov[f(\bX_0), f(\bX_k)].
\]

\begin{Lemma}
\label{lem:1}
For any $f\in\Ltwoz(\pi)$ and any $\varrho\in\nset$,
\begin{multline*}
\vara(f, P_1^{(L)})=p\, \vara(f, \tP_1^{(L)})+\pi (f^2\1_{\tXset^\mathsf{c}})+p(2\varrho-1)(\tpi f)^2\\
+2\sum_{k=1}^{\varrho-1}\left\{\esp[f(\bX_0)f(\bX_{k})\1_{A_\varrho^\mathsf{c}}]-p\esp[f(\tbX_0)f(\tbX_{k})\1_{\tilde{A}_\varrho^\mathsf{c}}]\right\}\\
+2\sum_{k\geq \varrho}\left\{\esp[f(\bX_0)f(\bX_k)]-p\cov[f(\tbX_0),f(\tbX_k)]\right\}\,.
\end{multline*}
\end{Lemma}

Note that the result holds if we replace $P_1^{(L)}$ and $\tP_1^{(L)}$ by $P_2^{(L)}$ and $\tP_2^{(L)}$, $\{\bX_k\}$ and $\{\tbX_k\}$ by $\{\bY_k\}$ and $\{\tbY_k\}$, and $A_\varrho$ and $\tilde{A}_\varrho$ by $B_\varrho$ and $\tilde{B}_\varrho$.

\begin{proof}
First, the relationship between the marginal variances is given by
$$
1=\var[f(\bX_0)]=p\var[f(\tbX_0)]+p(\tpi f)^2+\pi(f^2\1_{\tXset^\mathsf{c}})\,.
$$
Second, given that $f\in\Ltwoz(\pi)$, $\cov[f(\bX_0),f(\bX_k)]=\esp[f(\bX_0)f(\bX_k)]$. For $k<\rho$,
\begin{align*}
\esp[f(\bX_0)f(\bX_k)\1_{A_\varrho}] &=\int \pi(\d \bx_0) P_1^{(L)}(\bx_0,\d \bx_1)\cdots P_1^{(L)}(\bx_{k-1},\d \bx_k)f(\bx_0)f(\bx_k)\1_{A_\varrho} \cr
&=p\esp[f(\tbX_0)f(\tbX_k)\1_{\tilde{A}_\varrho}]
\end{align*}
because for all $\bx\in\itXset$ and all $B\subset \itXset$, $P_1^{(L)}(\bx,B) = \tP_1^{(L)}(\bx,B)$. Therefore,
\[
\esp[f(\bX_0)f(\bX_k)\1_{A_\varrho}] = p \esp[f(\tbX_0)f(\tbX_k)] -p\esp[f(\tbX_0)f(\tbX_k)\1_{\tilde{A}_\varrho^\mathsf{c}}],
\]
implying that
\begin{align*}
\cov[f(\bX_0),f(\bX_k)]&=p\cov[f(\tbX_0),f(\tbX_k))+p(\tpi f)^2-p\esp[f(\tbX_0)f(\tbX_k)\1_{\tilde{A}_\varrho^\mathsf{c}}] \cr
&\qquad +\esp[f(\bX_0)f(\bX_k)\1_{A_\varrho^\mathsf{c}}]\,.
\end{align*}
We are thus able to conclude the proof with
\begin{align*}
\vara(f, P_1^{(L)}) &= 1 +2 \sum_{k=1}^\infty\cov[f(\bX_0),f(\bX_k)]\\
&=p\var[f(\tbX_0)]+2p\sum_{k=1}^{\varrho-1}\cov[f(\tbX_0), f(\tbX_k)]
+(2\varrho-1)p(\tpi f)^2+\pi(f^2\1_{\tXset^\mathsf{c}})\\
&\quad +2\sum_{k=1}^{\varrho-1}\left[\esp[f(\bX_0)f(\bX_k)\1_{A_\varrho^\mathsf{c}}]
-p\esp[f(\tbX_0)f(\tbX_k)\1_{\tilde{A}_\varrho^\mathsf{c}}]\right]+2\sum_{k=\varrho}^\infty\cov[f(\bX_0),f(\bX_k)]\\
&=p\vara(f, \tP_1^{(L)})+(2\varrho-1)p(\tpi f)^2+\pi(f^2\1_{\tXset^\mathsf{c}})\\
&\quad+2\sum_{k=1}^{\varrho-1}\left[\esp[f(\bX_0)f(\bX_k)\1_{A_\varrho^\mathsf{c}}]
-p\esp[f(\tbX_0)f(\tbX_k)\1_{\tilde{A}_\varrho^\mathsf{c}}]\right] \cr
&\quad+2\sum_{k=\varrho}^{\infty}\left[\esp[f(\bX_0)f(\bX_k)]-p\cov[f(\tbX_0),f(\tbX_k)]\right]\,.
\end{align*}
\end{proof}

\begin{Lemma}
\label{lem:2}
Assume that there exists $0 <\omega \leq 1$ such that $P_{1}^{(L)}(\mathbf{x}, \mathbf{y}) \geq \omega P_{2}^{(L)}(\mathbf{x}, \mathbf{y})$, for all $(\bx,\by) \in \tXset_n^2$ with $\mathbf{x} \neq \mathbf{y}$. For any $f\in\Ltwoz(\pi)$ and $\varrho\in\nset$,
$$
\vara(f, P_1^{(L)}) \leq \frac{\vara(f, P_2^{(L)})}{\omega}+\frac{1-\omega}{\omega}+\Delta_\varrho(f)
$$
with
\begin{align*}
&\Delta_\varrho(f) = 2\sum_{k=1}^{\varrho-1}\left\{\esp[f(\bX_0)f(\bX_k)\1_{A_\varrho^\mathsf{c}}] - \frac{1}{\omega}\esp[f(\bY_0)f(\bY_k)\1_{B_\varrho^\mathsf{c}}]
+\frac{p}{\omega}\esp[f(\tbY_0)f(\tbY_k)\1_{\tilde{B}_\varrho^\mathsf{c}}] \right.\cr
&\qquad\left.-p\esp[f(\tbX_0)f(\tbX_k)\1_{\tilde{A}_\varrho^\mathsf{c}}]\right\}\\
&+2\sum_{k\geq \varrho}\left\{\esp[f(\bX_0)f(\bX_k)]-\frac{1}{\omega}\esp[f(\bY_0)f(\bY_k)] +\frac{p}{\omega}\cov[f(\tbY_0),f(\tbY_t)]-p\cov[f(\tbX_0),f(\tbX_k)]\right\}\,.
\end{align*}
\end{Lemma}

Note that $P_{1}^{(L)}(\mathbf{x}, \mathbf{y}) \geq \omega P_{2}^{(L)}(\mathbf{x}, \mathbf{y})$, for all $(\bx,\by) \in \tXset_n^2$ with $\mathbf{x} \neq \mathbf{y}$, is equivalent to $P_{1}(\mathbf{x}, \mathbf{y}) \geq \omega P_{2}(\mathbf{x}, \mathbf{y})$, for all $(\bx,\by) \in \tXset_n^2$ with $\mathbf{x} \neq \mathbf{y}$. Note also that $P_{1}^{(L)}(\mathbf{x}, \mathbf{y}) \geq \omega P_{2}^{(L)}(\mathbf{x}, \mathbf{y})$, for all $(\bx,\by) \in \tXset_n^2$ with $\mathbf{x} \neq \mathbf{y}$, is equivalent to $\tP_{1}^{(L)}(\mathbf{x}, \mathbf{y}) \geq \omega \tP_{2}^{(L)}(\mathbf{x}, \mathbf{y})$, for all $(\bx,\by) \in \tXset_n^2$ with $\mathbf{x} \neq \mathbf{y}$.

\begin{proof}
We first apply \autoref{lem:1} and obtain:
\begin{align*}
\vara(f,P_1^{(L)})&=p\, \vara(f,\tP_1^{(L)})+\pi (f^2\1_{\tXset^\mathsf{c}})+p(2\varrho-1)(\tpi f)^2\\
&+2\sum_{k=1}^{\varrho-1}\left\{\esp[f(\bX_0)f(\bX_{k})\1_{A_\varrho^\mathsf{c}}]-p\esp[f(\tbX_0)f(\tbX_{k})\1_{\tilde{A}_\varrho^\mathsf{c}}]\right\}\\
&+2\sum_{k\geq \varrho}\left\{\esp[f(\bX_0)f(\bX_k)]-p\cov[f(\tbX_0),f(\tbX_k)]\right\}.
\end{align*}
We now apply Lemma 33 of \cite{andrieu2018uniform} and obtain:
\begin{align*}
\vara(f,P_1^{(L)})&\leq \frac{p}{\omega}\, \vara(f,\tP_2^{(L)}) + \frac{p(1 - \omega)}{\omega} +\pi (f^2\1_{\tXset^\mathsf{c}})+p(2\varrho-1)(\tpi f)^2\\
&+2\sum_{k=1}^{\varrho-1}\left\{\esp[f(\bX_0)f(\bX_{k})\1_{A_\varrho^\mathsf{c}}]-p\esp[f(\tbX_0)f(\tbX_{k})\1_{\tilde{A}_\varrho^\mathsf{c}}]\right\}\\
&+2\sum_{k\geq \varrho}\left\{\esp[f(\bX_0)f(\bX_k)]-p\cov[f(\tbX_0),f(\tbX_k)]\right\}.
\end{align*}
Applying again \autoref{lem:1} yields the result after using that $0 < p \leq 1$ and
\begin{align*}
 \left(1-\frac{1}{\omega}\right)\left[\pi (f^2\1_{\tXset^\mathsf{c}})+p(2\varrho-1)(\tpi f)^2\right] \leq 0.
\end{align*}
\end{proof}

In the next lemma, we establish an upper bound for $\Delta_\varrho(f)$.

\begin{Lemma}\label{lem:3}
For any $\delta>0$, $\varrho\in\nset$ and $f \in \Ltwoz(\pi)$, we have
  \begin{equation*}
 \Delta_\rho(f) \leq \frac{8}{\omega p}\left(\varrho^2\|f\|_{\pi,2+\delta}^2\left[{1}-{\pi(\itXset)}\right]^{\delta / (2 + \delta)}
 + \frac{\exp(-\varrho \ulambda / 2)}{\ulambda / 2}\right)\,.
 \end{equation*}
  \end{Lemma}

  \begin{proof}
First, note that for any $\delta>0$, using H\"older's inequality,
\[
\left|\esp[f(\bX_0)f(\bX_k)\1_{A_\varrho^\mathsf{c}}]\right|\leq \esp[|f(\bX_0)f(\bX_k)\1_{A_\varrho^\mathsf{c}}|]
\leq \left[\esp\left[|f(\bX_0)f(\bX_k)|^{1+\delta/2}\right]\right]^{2/2+\delta}\proba(A_\varrho^\mathsf{c})^{\delta/2+\delta}\,.
\]
Moreover, using Cauchy--Schwarz inequality,
\[
\esp\left[|f(\bX_0)f(\bX_k)|^{1+\delta/2}\right] \leq \left(\esp\left[|f(\bX_0)|^{2+\delta}\right]\right)^{1/2}\left(\esp\left[|f(\bX_k)|^{2+\delta}\right]\right)^{1/2} =\esp\left[|f(\bX_0)|^{2+\delta}\right]\,.
\]
Also,
\[
  \proba(A_\varrho^\mathsf{c})=\proba\left(\bigcup_{k=1}^{\varrho-1} \bX_k\in\partial\tXset\cup\tXset^\mathsf{c}\right)\leq
  \sum_{k=1}^{\varrho-1}\proba(\bX_k\in\partial\tXset\cup\tXset^\mathsf{c})\leq \varrho\pi(\partial\tXset \cup \tXset^\mathsf{c}) = \varrho(1 - \pi(\itXset))\,.
\]
Combining these results yields
$$
\left|\esp[f(\bX_0)f(\bX_k)\1_{A_\varrho^\mathsf{c}}]\right|\leq
\|f\|_{\pi,2+\delta}^2 \left[\varrho(1 - \pi(\itXset))\right]^{\delta/(2+\delta)} \leq \frac{\|f\|_{\pi,2+\delta}^2}{\omega p} \varrho (1 - \pi(\itXset))^{\delta/(2+\delta)}\,,
$$
using that $0 < \delta/(2+\delta) \leq 1$ and $0 < \omega, p \leq 1$.
Similarly, for any $\delta>0$,
$$
\left|\esp[f(\tbX_0)f(\tbX_k)\1_{\tilde{A}_\varrho^\mathsf{c}}]\right|\leq
\|f\|_{\tpi,2+\delta}^2\varrho \left[\tpi(\partial\tXset)\right]^{\delta/(2+\delta)} \leq \frac{\|f\|_{\pi,2+\delta}^2}{\omega p^2}\varrho (1 - \pi(\itXset))^{\delta/(2+\delta)} \,,
$$
using that
$$
p\|f\|_{\tpi,2+\delta}^2=\bigg[\sum_{\bx\in\tXset}f(\bx)^{2+\delta}\pi(\bx)\bigg]^{2/2+\delta}\leq \|f\|_{\pi,2+\delta}^2
$$
and the definition of $\tpi$.

Similar bounds also hold for $\left|\esp[f(\bY_0)f(\bY_k)\1_{B_\varrho^\mathsf{c}}]\right|$ and $\left|\esp[f(\tbY_0)f(\tbY_k)\1_{\tilde{B}_\varrho^\mathsf{c}}]\right|$. Therefore,
\begin{align*}
&2\left|\sum_{k=1}^{\varrho-1}\left\{\esp[f(\bX_0)f(\bX_k)\1_{A_\varrho^\mathsf{c}}] -\frac{1}{\omega}\esp[f(\bY_0)f(\bY_k)\1_{B_\varrho^\mathsf{c}}] \right.\right. \cr
&\qquad\left.\left.+\frac{p}{\omega}\esp[f(\tbY_0)f(\tbY_k)\1_{\tilde{B}_\varrho^\mathsf{c}}] -p\esp[f(\tbX_0)f(\tbX_k)\1_{\tilde{A}_\varrho^\mathsf{c}}]\right\}\right|\\
&\leq
\frac{8}{\omega p} \varrho^2\|f\|_{\pi,2+\delta}^2 (1 - \pi(\itXset))^{\delta/(2+\delta)}\,.
\end{align*}

We now bound the second sum in $\Delta_\varrho(f)$. Using Cauchy--Schwarz inequality and \eqref{eq:contraction},
\[
 |\esp[f(\bX_0)f(\bX_k)]| \leq \|f\|_{\pi,2}\|(P_1^{(L)})^k f\|_{\pi,2}\leq (1-\lambda_1^{(L)})^k \leq (1-\ulambda/2)^k \leq  \frac{1}{\omega}(1-\ulambda/2)^k,
\]
where $\lambda_1^{(L)}$ is the absolute (and right) spectral gap of $P_1^{(L)}$. The result follows from Theorem 2 in \cite{zanella2019informed}, which indicates that, if $P_{1}^{(L)}(\mathbf{x}, \mathbf{y}) \geq (1/2) P_{1}(\mathbf{x}, \mathbf{y})$, for all $(\bx,\by) \in \Xset_n^2$ with $\bx \neq \by$, then the right spectral gap $\lambda_1^{(L)}$ is such that $\lambda_1^{(L)} \geq (1/2) \lambda_1$.

Similarly,
\begin{multline*}
 |\cov[f(\tbX_0),f(\tbX_k)]|\leq \|f-\tpi f\|_{\tpi,2}\|(\tP_1^{(L)})^k (f-\tpi f)\|_{\tpi,2} \cr
  = \|f-\tpi f\|_{\tpi,2}^2\left\|(\tP_1^{(L)})^k \frac{(f-\tpi f)}{\|f-\tpi f\|_{\tpi,2}}\right\|_{\tpi,2} \cr
 \leq \var[f(\tbX_0)](1-\tlambda_1^{(L)})^k \leq \frac{1}{\omega p}(1-\ulambda/2)^k\,,
\end{multline*}
using that $\var[f(\tbX_0)] \leq 1 / p$.

Similar bounds also hold for $\left|\esp[f(\bY_0)f(\bY_k)]\right|$ and $\left|\cov[f(\tbY_0),f(\tbY_k)]\right|$. Therefore,
\begin{multline*}
2\left|  \sum_{k\geq \varrho}\left\{\esp[f(\bX_0)f(\bX_k)]-\frac{1}{\omega}\esp[f(\bY_0)f(\bY_k)]+\frac{p}{\omega} \cov[f(\tbY_0),f(\tbY_k)]- p\cov[f(\tbX_0),f(\tbX_k)]\right\}\right|\\
\leq  \frac{8}{\omega}\sum_{k\geq \varrho}(1-\ulambda/2)^k = \frac{8}{\omega}(1-\ulambda/2)^\varrho\frac{1}{\ulambda/2} \leq \frac{8}{\omega p}\frac{\exp(-\varrho \ulambda/2)}{\ulambda/2}\,,
\end{multline*}
using that $1 - x \leq \exp(-x)$.
\end{proof}

We now turn to the proofs of Theorems \ref{thm:1} and \ref{thm:2}. These theorems are stated and proved under the asymptotic framework presented in \autoref{sec:weak_peskun}. In the proofs, it will thus be important to highlight a dependence on $n$ of the target distribution, state-space, and so on.

\begin{proof}[Proof of \autoref{thm:1}]
 We first apply Lemmas \ref{lem:2} and \ref{lem:3}:
 \begin{align*}
  \vara(f_n, P_{1,  n}^{(L)}) &\leq \frac{1}{\omega(n)} \vara(f_n, P_{2, n}^{(L)}) + \frac{1 - \omega(n)}{\omega(n)}
 \cr
 &\qquad +\frac{8}{\omega(n) p(n)}\left(\varrho(n)^2\|f_n\|_{\pi_n,2+\delta}^2\left[{1}-{\pi_n(\itXset_n)}\right]^{\delta / (2 + \delta)}
 + \frac{\exp(-\varrho(n) \ulambda(n)/2)}{\ulambda(n)/2}\right).
 \end{align*}
 Let $\overline{\omega} > \eps > 0$. Consider that $n > n^*$, a positive integer which will be defined in relation to other positive integers. Under \autoref{ass:1}, we know that there exists $n_1^*$ such that for any $n > n_1^*$,
 \begin{align*}
  \vara(f_n, P_{1,  n}^{(L)}) &\leq \frac{1}{\overline{\omega} - \eps} \vara(f_n, P_{2, n}^{(L)}) + \frac{1 - \overline{\omega}}{\overline{\omega}} + \frac{\eps}{3} \cr
  &\qquad +  \frac{8}{\overline{\omega} - \eps}\left(\varrho(n)^2\|f_n\|_{\pi_n,2+\delta}^2\left[{1}-{\pi_n(\itXset_n)}\right]^{\delta / (2 + \delta)}
 + \frac{\exp(-\varrho(n) \ulambda(n)/2)}{\ulambda(n)/2}\right).
 \end{align*}
 Take $n^* \geq n_1^*$.

 Now, we set $\varrho(n) = \lfloor 1/(1-\pi_n(\itXset_n))^{(\bdelta - \gamma)/2}\rfloor$, where $\lfloor \, \cdot \, \rfloor$ is the floor function and $\bdelta := \delta / (2 + \delta)$, and note that, by \autoref{ass:1} and given that $\bdelta > \gamma > 0$, $\varrho(n) \rightarrow \infty$. By assumption, we know that there exists $n_2^*$ such that for any $n > n_2^*$,
 \[
  \left\|f_n\right\|_{\pi_n, 2 + \delta} \varrho(n)^2 (1 - \pi_n(\itXset_n))^{\bdelta} \leq \left\|f_n\right\|_{\pi_n, 2 + \delta}(1-\pi_n(\itXset_n))^\gamma \leq \frac{\overline{\omega} - \eps}{24}\eps.
 \]
Take $n^* \geq n_2^*$.

 Given that $\ulambda(n)$ is bounded away from zero by assumption, we know that there exists $n_3^*$ such that for any $n > n_3^*$,
 \[
   \frac{\exp\{-\varrho(n) \ulambda(n)/2\}}{\ulambda(n)/2} \leq \frac{\overline{\omega} - \eps}{24}\eps.
 \]
 Take $n^* \geq n_3^*$. This yields
 \[
  \vara(f_n, P_{1,n}^{(L)}) \leq \frac{1}{\overline{\omega}-\eps}\vara(f_n, P_{2,n}^{(L)}) +\frac{1- \overline{\omega}}{\overline{\omega}} + \eps.
 \]
 To conclude the proof, we use Theorem 1 of \cite{deligiannidis2018ergodic}, which indicates that
 \[
  \vara(f_n, P_{i,n}^{(L)}) = 1 + 2 \vara(f_n, P_{i,n}), \quad i = 1, 2.
 \]
\end{proof}

\begin{proof}[Proof of \autoref{thm:2}]
 We follow a similar approach than for the proof of \autoref{thm:1}. Let $\overline{\omega} > \eps > 0$. Consider that $n > n^*$. Under \autoref{ass:1} and using Lemmas \ref{lem:2} and \ref{lem:3}, we know that there exists $n_1^*$ such that for any $n > n_1^*$,
 \begin{align*}
  \vara(f_n, P_{1,  n}^{(L)}) &\leq \frac{1}{\overline{\omega} - \eps} \vara(f_n, P_{2, n}^{(L)}) + \frac{1 - \overline{\omega}}{\overline{\omega}} + \frac{\eps}{3} \cr
  &\qquad +  \frac{8}{\overline{\omega} - \eps}\left(\varrho(n)^2\|f_n\|_{\pi_n,2+\delta}^2\left[{1}-{\pi_n(\itXset_n)}\right]^{\delta / (2 + \delta)}
 + \frac{\exp(-\varrho(n) \ulambda(n))}{\ulambda(n)}\right).
 \end{align*}
 Take $n^* \geq n_1^*$.

 As previously, we set $\varrho(n) = \lfloor 1/(1-\pi_n(\itXset_n))^{(\bdelta - \gamma)/2}\rfloor$, (again with $\bdelta>\gamma>0$) which implies that there exists $n_2^*$ such that for any $n > n_2^*$,
 \[
  \left\|f_n\right\|_{\pi_n, 2 + \delta} \varrho(n)^2 (1 - \pi_n(\itXset_n))^{\bdelta} \leq \left\|f_n\right\|_{\pi_n, 2 + \delta}(1-p(n))^\gamma \leq \frac{\overline{\omega} - \eps}{24}\eps.
 \]
 Take $n^* \geq n_2^*$.

 We consider that $\ulambda(n) \rightarrow 0$; otherwise, we are in the same situation as the previous proof and it has been shown that the result holds. We write
 \begin{align*}
  \frac{\exp\{-\varrho(n) \ulambda(n)/2\}}{\ulambda(n)/2} &= \exp\left\{-\varrho(n) \ulambda(n)/2\left(1 + \frac{\log \ulambda(n)/2}{\varrho(n) \ulambda(n)/2}\right)\right\} \cr
  &= \exp\left\{-\varrho(n) \ulambda(n)/2\left(1 + \frac{[\log \ulambda(n)/2] (\ulambda(n)/2)^{1/2}}{\varrho(n) (\ulambda(n)/2)^{3/2}}\right)\right\}.
 \end{align*}
 Clearly, $[\log \ulambda(n)/2] (\ulambda(n)/2)^{1/2}$ vanishes. Now we establish that $\varrho(n) (\ulambda(n)/2)^{3/2} \rightarrow \infty$ which implies that $\varrho(n) (\ulambda(n)/2) \rightarrow \infty$. By \eqref{eq:ass2_thm3},
 \begin{align*}
  \frac{1 - \pi_n(\itXset_n)}{\ulambda(n)^{3 / (\bdelta - \gamma)}} \rightarrow 0,
 \end{align*}
 which is equivalent to
 \begin{align*}
  \frac{(1 - \pi_n(\itXset_n))^{(\bdelta - \gamma)/2}}{\ulambda(n)^{3 / 2}} \rightarrow 0,
 \end{align*}
 which allows to conclude that $\varrho(n) (\ulambda(n)/2)^{3/2} \rightarrow \infty$. Therefore, there exists $n_3^*$ such that for any $n > n_3^*$,
 \[
   \frac{\exp\{-\varrho(n) \ulambda(n)/2\}}{\ulambda(n)/2} \leq \frac{\overline{\omega} - \eps}{24}\eps.
 \]
 Take $n^* \geq n_3^*$. This concludes the proof by using Theorem 1 of \cite{deligiannidis2018ergodic} as before.
\end{proof}

\begin{proof}[Proof of \autoref{prop_invariance}]
 It suffices to prove that the probability to reach the state $(\mathbf{y}, \nu')$ in one step is equal to the probability of this state under the target:
 \begin{align*}\label{eqn_invariance}
  \sum_{\mathbf{x}, \nu} \pi(\mathbf{x}) \, (1 / 2) \, P((\mathbf{x}, \nu), (\mathbf{y}, \nu')) = \pi(\mathbf{y}) \, (1 / 2).
 \end{align*}
 where $P$ is the transition kernel.

The probability to reach the state $(\mathbf{y}, \nu')$ from some $(\mathbf{x}, \nu)$ is given by:
 \begin{align*}
  P((\mathbf{x}, \nu), (\mathbf{y}, \nu')) &= T_\nu(\mathbf{x}, \Xset) \, Q_{\mathbf{x}, \nu}(\mathbf{y}) \, \ind(\nu = \nu') \cr
  &\qquad + \ind(\nu = -\nu', \mathbf{x} = \mathbf{y}) \left[(\rho_\nu(\mathbf{x}) + T_\nu(\mathbf{x}, \Xset)) - T_\nu(\mathbf{x}, \Xset)\right] \cr
  &\qquad + \ind(\nu = \nu', \mathbf{x} = \mathbf{y}) \left[1 - \rho_\nu(\mathbf{x}) - T_\nu(\mathbf{x}, \Xset)\right] \cr
  &= q_{\mathbf{x}, \nu}(\mathbf{y}) \, \alpha_\nu(\mathbf{x}, \mathbf{y}) \, \ind(\nu = \nu') \cr
  &\qquad + \ind(\nu = -\nu', \mathbf{x} = \mathbf{y}) \, \rho_\nu(\mathbf{x})  \cr
  &\qquad + \ind(\nu = \nu', \mathbf{x} = \mathbf{y}) \left[1 - \rho_\nu(\mathbf{x}) - T_\nu(\mathbf{x}, \Xset)\right].
 \end{align*}

 We have that
  \begin{align*}
   \pi(\mathbf{x}) \, (1 / 2) \, P((\mathbf{x}, \nu), (\mathbf{y}, \nu')) &= (1 / 2) \, \pi(\mathbf{y}) \, q_{\mathbf{y}, -\nu'}(\mathbf{x}) \, \alpha_{-\nu'}(\mathbf{y}, \mathbf{x}) \, \ind(-\nu' = -\nu) \cr
    &\qquad + (1 / 2) \, \pi(\mathbf{y}) \,  \ind(-\nu' = \nu, \mathbf{y} = \mathbf{x}) \, \rho_{-\nu'}(\mathbf{y}) \cr
    &\qquad + (1 / 2) \, \pi(\mathbf{y}) \ind(-\nu' = -\nu, \mathbf{y} = \mathbf{x}) \left[1 - \rho_{-\nu'}(\mathbf{y}) - T_{-\nu'}(\mathbf{y}, \Xset)\right] \cr
    &= (1 / 2) \, \pi(\mathbf{y}) \, T_{-\nu'}(\mathbf{y}, \Xset) \, Q_{\mathbf{y}, -\nu'}(\mathbf{x}) \ind(-\nu' = -\nu) \cr
    &\qquad + (1 / 2) \, \pi(\mathbf{y}) \,  \ind(-\nu' = \nu, \mathbf{y} = \mathbf{x}) \left[(\rho_{-\nu'}(\mathbf{y}) + T_{-\nu'}(\mathbf{y}, \Xset)) - T_{-\nu'}(\mathbf{y}, \Xset)\right] \cr
    &\qquad + (1 / 2) \, \pi(\mathbf{y}) \, \ind(-\nu' = -\nu, \mathbf{y} = \mathbf{x}) \left[1 - \rho_{-\nu'}(\mathbf{y}) - T_{-\nu'}(\mathbf{y}, \Xset)\right],
 \end{align*}
 where we used the definition of $\alpha$ for the first term and that  $\rho_{\nu}(\mathbf{x}) - \rho_{-\nu}(\mathbf{x}) = T_{-\nu}(\mathbf{x}, \Xset) - T_{\nu}(\mathbf{x},\Xset)$ for the third term. Notice the sum on the RHS is equal to the probability to reach some $(\mathbf{x}, -\nu)$, starting from $(\mathbf{y}, -\nu')$:  $(1 / 2) \, \pi(\mathbf{y}) \, P((\mathbf{y}, -\nu'), (\mathbf{x}, -\nu))$.

 Therefore,
 \begin{align*}
   \sum_{\mathbf{x}, \nu} \pi(\mathbf{x}) \, (1 / 2) \, P((\mathbf{x}, \nu), (\mathbf{y}, \nu')) &= \sum_{\mathbf{x}, \nu} (1 / 2) \, \pi(\mathbf{y}) \, P((\mathbf{y}, -\nu'), (\mathbf{x}, -\nu))  \cr
  & = (1 / 2) \, \pi(\mathbf{y}),
 \end{align*}
 which concludes the proof.
\end{proof}

We now present a lemma that will be useful in the next proofs. We define $\bpi :=\pi \otimes \mathcal{U}\{-1, +1\}$ and note that in the following we can assume without loss of generality that $\bpi f = 0$.

\begin{Lemma}\label{lemma:2}
Assume that $\Xset$ is finite. Then, for any function $f: \Xset \times \{-1, + 1\} \to \rset$,
\begin{equation}\label{eqlemma2}
\lim_{\lambda\to 1}\sum_{k>0}\lambda^k\pscal{f}{P_\rho^k f}_{\bpi} = \sum_{k>0}\pscal{f}{P_\rho^k f}_{\bpi}.
\end{equation}
\end{Lemma}

\begin{proof}
 Let us define the sequence of functions $S_N: \lambda \mapsto \sum_{0 < k \leq N}\lambda^k\pscal{f}{P_\rho^k f}$ defined for $\lambda\in[0,1)$ and its limit $S(\lambda) = \sum_{k>0}\lambda^k\pscal{f}{P_\rho^k f}_{\bpi}$ (the dependence of $S_N$ and $S$ on $f$ and $P_\rho$ is implicit).  We now show that the partial sum $S_N$ converges uniformly to $S$ on $[0,1)$, and given that for each $N\in\nset$, the function $\lambda \to \lambda^N\pscal{f}{P_\rho^N f}_{\bpi}$ admits a limit when $\lambda\to 1$,  we have that $S$ admits a limit when $\lambda\to 1$, given by
 $$
 \lim_{\lambda\to 1}S(\lambda)=S(1)=\sum_{k>0}\pscal{f}{P_\rho^k f},
 $$
 which is \eqref{eqlemma2}.

 First, note that
 \begin{align*} 
   \sup_{\lambda\in[0,1)}\left|S_N(\lambda)-S(\lambda)\right| =   \sup_{\lambda\in[0,1)}\left|\sum_{k > N} \lambda^k\pscal{f}{P_\rho^k f}_{\bpi}\right| &\leq
   \sup_{\lambda\in[0,1)} \sum_{k > N}\lambda^k\left|\pscal{f}{P_\rho^k f}_{\bpi}\right| \cr
   &= \sum_{k > N}\left|\pscal{f}{P_\rho^k f}_{\bpi}\right|.
 \end{align*}
Thus, to prove that $\sup_{\lambda\in[0,1)}\left|S_N(\lambda)-S(\lambda)\right|\to 0$, it is sufficient to prove that the series
\[\sum_{k>0}\left|\pscal{f}{P_\rho^k f}_{\bpi}\right|\]
 converges.

By bilinearity of the inner product and by linearity of the iterated operators $P_\rho,P_\rho^2,\ldots$, it can be checked that for any linear mapping $\phi$
\begin{equation}\label{eq:lemma2_equiv}
\sum_{k=1}^\infty \left|\pscalpi{f}{P_\rho^k f}_{\bpi}\right|<\infty\Leftrightarrow
\sum_{k=1}^\infty \left|\pscalpi{\phi(f)}{P_\rho^k \phi(f)}_{\bpi}\right|<\infty\,.
\end{equation}
Given that $\Xset$ is finite, any function $f: \Xset \times \{-1, + 1\} \to \rset$ is such that $\sup|f|<\infty$. As a consequence, we may use $\phi(f) := f / \sup |f|$ (recall that $\bpi f = 0$).  In the following we denote by $\Ltwozstar(\bpi)$ the subset of $\Ltwo(\bpi)$ such that
$$
 \Ltwozstar(\bpi) := \left\{f \in \Ltwo(\bpi): \bpi f =0, \sup |f| \leq 1\right\}.
$$

By \eqref{eq:lemma2_equiv}, we only need to check that the series $\sum_{k>0}\left|\pscal{f}{P_\rho^k f}\right|$ converges for any $f\in\Ltwozstar(\bpi)$. Given that $\Xset$ is finite, $P_\rho$ is uniformly ergodic and there exist constants $\gamma\in(0,1)$ and $C\in(0,\infty)$ such that for any $t\in\nset$,
\begin{equation}\label{eq:lemma2_tv}
\sup_{(\bx, \nu)\in\Xset\times\{-1,+1\}}\|\delta_{\bx,\nu}P_\rho^t-\bpi\|_\tv\leq C\gamma^t,
\end{equation}
where for any signed measure $\mu$, $\|\mu\|_\tv$ denotes its total variation.  Denoting a state of the extended state-space by $\bbx := (\bx, \nu) \in \Xset \times \{-1, +1\}$, we have that, for any $ f\in\Ltwozstar(\bpi)$,
\begin{align*}
 |\pscal{f}{P_\rho^k f}_{\bpi}| = \left|\sum_{\bbx} f(\bbx) \, \bpi(\bbx) \sum_{\bby} f(\bby)P_\rho^k(\bbx, \bby)\right| &\leq \sum_{\bbx} |f(\bbx)| \, \bpi(\bbx) \left|\sum_{\bby} f(\bby)P_\rho^k(\bbx, \bby)\right| \cr
 &= \sum_{\bbx} |f(\bbx)| \, \bpi(\bbx) \left|\sum_{\bby} f(\bby)P_\rho^k(\bbx, \bby) - \bpi f\right| \cr
 &\leq \sum_{\bbx} |f(\bbx)| \,\bpi(\bbx) \sup_{f \in \Ltwozstar(\bpi)} \left|\sum_{\bby} f(\bby)P_\rho^k(\bbx, \bby) - \bpi f\right| \cr
 &\leq \sum_{\bbx} |f(\bbx)| \, \bpi(\bbx) \, 2 \, \|\delta_{\bx,\nu}P_\rho^k-\bpi\|_\tv \cr
 &\leq C \gamma^k,
\end{align*}
using Jensen's inequality, that $\bpi f = 0$, that $\|\mu\|_\tv = (1/2)\sup_{|g| \leq 1}|\mu g|$ (see, e.g., Proposition 3 in \cite{roberts2004general}) and thus \eqref{eq:lemma2_tv}, and finally that $|f| \leq 1$.

Therefore,
\[
 \sum_{k>0}\left|\pscal{f}{P_\rho^k f}_{\bpi}\right| \leq C \sum_{k>0} \gamma^k < \infty.
\]
As a consequence, $S_n$ converges uniformly to $S$ on $[0,1)$ which concludes the proof.
\end{proof}

\begin{proof}[Proof of \autoref{cor_bestrho}]
 The results of Theorem~6 in \cite{andrieu2019peskun} holds in our framework, implying that
 \[
  \vara_\lambda(f, P_{\rho^*}) \leq \vara_\lambda(f, P_{\rho}) \leq \vara_\lambda(f, P_{\rho}^{\text{w}}),
 \]
 where $\vara_\lambda(f, P_\rho) := \var[ f(\mathbf{X}, \nu)] + 2\sum_{k>0} \lambda^k \pscalpi{f}{P_\rho^k f}_{\bpi}$ with $\lambda \in[0, 1)$. \autoref{lemma:2} allows to conclude.
\end{proof}

\begin{proof}[Proof of \autoref{cor_domination_unif}]
 The proof is an application of Theorem 7 in \cite{andrieu2019peskun} which will allow to establish that
 \[
  \vara_\lambda(f_n, P_{\rho}) \leq \vara_\lambda(f_n, P_{\text{MH}}).
 \]
 We will thus be able to conclude using \autoref{lemma:2}.

 In order to apply Theorem 7 in \cite{andrieu2019peskun}, we must verify that
\[
 q_\mathbf{x}(\mathbf{y}) \, \alpha(\mathbf{x}, \mathbf{y}) = (1/2) \, q_{\mathbf{x},+1}(\mathbf{y}) \, \alpha_{+1}(\mathbf{x}, \mathbf{y}) + (1/2) \, q_{\mathbf{x},-1}(\mathbf{y}) \,\alpha_{-1}(\mathbf{x}, \mathbf{y}),
\]
for all $\mathbf{x}$ and $\mathbf{y}$. This is straightforward to verify under the assumptions of \autoref{cor_domination_unif}:
\begin{align*}
  &(1/2) \, q_{\mathbf{x},+1}(\mathbf{y}) \, \alpha_{+1}(\mathbf{x}, \mathbf{y}) + (1/2) \, q_{\mathbf{x},-1}(\mathbf{y}) \,\alpha_{-1}(\mathbf{x}, \mathbf{y}) \cr
  &\qquad= \frac{1}{2}\frac{1}{(|\neigh(\mathbf{x})| / 2)} \left(1 \wedge \frac{\pi(\mathbf{y})}{\pi(\mathbf{x})}\right) \1_{\mathbf{y} \in \neigh_{+1}(\mathbf{x})} + \frac{1}{2}\frac{1}{(|\neigh(\mathbf{x})| / 2)} \left(1 \wedge \frac{\pi(\mathbf{y})}{\pi(\mathbf{x})}\right) \1_{\mathbf{y} \in \neigh_{-1}(\mathbf{x})} \cr
  &\qquad= \frac{1}{|\neigh(\mathbf{x})|} \left(1 \wedge \frac{\pi(\mathbf{y})}{\pi(\mathbf{x})}\right) \left(\1_{\mathbf{y} \in \neigh_{+1}(\mathbf{x})} + \1_{\mathbf{y} \in \neigh_{-1}(\mathbf{x})}\right) = q_\mathbf{x}(\mathbf{y}) \, \alpha(\mathbf{x}, \mathbf{y}). \qedhere
\end{align*}
\end{proof}

\begin{proof}[Proof of \autoref{lemma:omega_unif}]
Let $(\bx,\by)\in\tXset_n^2$, $\bx\neq \by$, $\by\in\neigh_\nu(\bx)$. Since $\bx\in\tXset_n$, we have $2n_\nu(\bx)\in[n-2\beta(n), n+2\beta(n)]$ and thus
$$
P_{\rev,n}(\bx,\by)=\frac{1}{2n_{\nu}(\bx)}\left(1\wedge \frac{\pi_n(\by)}{\pi_n(\bx)}\frac{n_{\nu}(\bx)}{n_{-\nu}(\by)}\right)\geq \left(1+\frac{\beta(n)}{n/2}\right)^{-1}\frac{1}{n}\left(1\wedge \frac{\pi_n(\by)}{\pi_n(\bx)}\frac{n_{\nu}(\bx)}{n_{-\nu}(\by)}\right).
$$
Noting that
$$
\frac{n_{\nu}(\bx)}{n_{-\nu}(\by)}\geq \max\left\{0, \left(1-\frac{\beta(n)}{n/2}\right)\left(1+\frac{\beta(n)}{n/2}\right)^{-1}\right\},
$$
and that for any $a>0$ and $b\in(0,1)$, we have $1\wedge ab\geq b(1\wedge a)$ and thus
$$
P_{\rev,n}(\bx,\by)\geq \frac{1}{n}\left(1\wedge \frac{\pi_n(\by)}{\pi_n(\bx)}\right) \left(1+\frac{\beta(n)}{n/2}\right)^{-1}\max\left\{0, \left(1-\frac{\beta(n)}{n/2}\right)\left(1+\frac{\beta(n)}{n/2}\right)^{-1}\right\}.
$$
This completes the proof since $\beta(n)=o(n)$ implies that for a large enough $n$, $1-{\beta(n)}/(n/2)>0$ and that $P_{\MH,n}(\bx,\by)= (1/n)\left(1\wedge {\pi_n(\by)}/{\pi_n(\bx)}\right)$.
\end{proof}

\begin{proof}[Proof of \autoref{cor_domination_lifted}]
 Analogous to that of \autoref{cor_domination_unif}.
\end{proof}

\begin{proof}[Proof of \autoref{lemma:omega_loc}]
Let $\bx, \by \in \tXset_n$ with $\by\in\neigh_\nu(\bx)$, then
$$
P_{\rev,n}(\bx,\by)=q_{\bx}(\by)\frac{c_n(\bx)}{2c_{n,\nu}(\bx)}\left(1\wedge \frac{c_n(\bx)}{c_n(\by)}\varphi_n(\bx,\by)\right)\,,\qquad
\varphi_n(\bx,\by):=\frac{c_{n,\nu}(\bx)/c_n(\bx)}{c_{n,\nu}(\by)/c_n(\by)}\,.
$$
For any $\bx\in\tXset_n$, $c_{n,\nu}(\bx)/c_n(\bx)\in[1/2-\beta(n)/c_n(\bx),1/2+\beta(n)/c_n(\bx)]$ so that
$$
\varphi_n(\bx,\by)\geq \max\left\{\frac{c_n(\bx)-2\beta(n)}{c_n(\bx)+2\beta(n)c_n(\bx)/c_n(\by)}, 0\right\}\geq \max\left\{\frac{1-2\beta(n)/c_n(\bx)}{1+2\tau_n\beta(n)/c_n(\bx)}, 0\right\}\,.
$$
As in the proof of \autoref{lemma:omega_unif},
$$
P_{\rev,n}(\bx,\by)\geq q_{\bx}(\by)\left(1\wedge\frac{c_n(\bx)}{c_n(\by)}\right)\left(1+\frac{\beta(n)}{c_n(\bx)/2}\right)^{-1}\max\left\{\frac{1-2\beta(n)/c_n(\bx)}{1+2\tau_n\beta(n)/c_n(\bx)}, 0\right\}\,.
$$
%
By assumption $c_n(\bx)\geq\inf\{c_n(\bx)\,:\,\bx\in\tXset_n\}\geq n m$  and we thus have that $\beta(n)/c_n(\bx)\to 0$ since $\beta(n)=o(n)$. Thus for $n$ sufficiently large,
$$
\frac{1-2\beta(n)/nm}{1+2\tau_n\beta(n)/n m}\in(0,1)
$$
so that
$$
P_{\rev,n}(\bx,\by)\geq P_{\MH,n}(\bx,\by)\left(1+\frac{\beta(n)}{nm/2}\right)^{-1}\left(\frac{1-2\beta(n)/nm}{1+2\tau_n\beta(n)/n m}\right)\,.
$$
\end{proof}

\begin{proof}[Proof of \autoref{prop:simple_ising_1}]
We first obtain the bound $P_{\rev, n}(\bx, \by)\geq (1 / 2) P_{\MH, n}(\bx,\by)$, for all $(\bx, \by) \in \Xset_n^2$ with $\bx \neq \by$. Next we prove that it is essentially not possible to obtain a better bound by establishing an essentially matching upper bound for specific transitions.

For any $\bx, \by$ such that $\by \in \neigh(\bx)$ and $\by \neq \bx$,
\begin{align*}
 P_{\rev, n}(\bx, \by) &= \frac{1}{2} \frac{g(\pi_n(\by) / \pi_n(\bx))}{c_{n, \nu}(\bx)} \left(1 \wedge \frac{c_{n, \nu}(\bx)}{c_{n, -\nu}(\by)}\right) \cr
 &= \frac{c_n(\bx)/2}{c_{n, \nu}(\bx)}  \frac{g(\pi_n(\by) / \pi_n(\bx))}{c_n(\bx)} \left(1 \wedge \frac{c_n(\bx)}{c_n(\by)} \frac{c_{n,\nu}(\bx)}{c_n(\bx) / 2} \frac{c_n(\by) / 2}{c_{n, -\nu}(\by)}\right) \cr
 & \geq \frac{g(\pi_n(\by) / \pi_n(\bx))}{c_n(\bx)} \left(1 \wedge \frac{c_n(\bx)}{c_n(\by)} \right) \frac{c_n(\bx)/2}{c_{n, \nu}(\bx)} \left(1 \wedge \frac{c_{n, \nu}(\bx)}{c_n(\bx) / 2} \frac{c_n(\by) / 2}{c_{n, -\nu}(\by)}\right) \cr
 &= P_{\MH, n}(\bx, \by) \frac{c_n(\bx)/2}{c_{n, \nu}(\bx)} \left(1 \wedge \frac{c_{n, \nu}(\bx)}{c_n(\bx) / 2} \frac{c_n(\by) / 2}{c_{n, -\nu}(\by)}\right),
\end{align*}
using that
\[
 \frac{c_n(\bx)}{c_n(\by)} \frac{c_{n, \nu}(\bx)}{c_n(\bx) / 2} \frac{c_n(\by) / 2}{c_{n, -\nu}(\by)} \geq \frac{c_n(\bx)}{c_n(\by)} \left(1 \wedge \frac{c_{n, \nu}(\bx)}{c_n(\bx) / 2} \frac{c_n(\by) / 2}{c_{n, -\nu}(\by)}\right),
\]
and that, for any $a > 0$ and $b \in (0, 1]$, we have $1 \wedge a b \geq b(1 \wedge a)$. Note that if $\by \neq \bx$ and $\by \notin \neigh(\bx)$, $P_{\rev, n}(\bx, \by) = P_{\MH, n}(\bx, \by) = 0$.

Let us analyse the terms that multiply $P_{\MH, n}(\bx, \by)$ above. We have that
\[
 \frac{c_n(\bx)/2}{c_{n, \nu}(\bx)} \left(1 \wedge \frac{c_{n, \nu}(\bx)}{c_n(\bx) / 2} \frac{c_n(\by) / 2}{c_{n, -\nu}(\by)}\right) = \frac{c_n(\bx)/2}{c_{n, \nu}(\bx)} \wedge \frac{c_n(\by) / 2}{c_{n, -\nu}(\by)}.
\]
We consider that we obtain $\by$ from $\bx$ by changing a coordinate, say $x_j$, from $+1$ to $-1$, implying that the transition is associated with $\nu = -1$. We thus necessarily have $|\neigh_{-1}(\bx)| = n_{+1}(\bx) > 0$. We can obtain the same lower bound in the opposite situation (when the transition is associated with $\nu = +1$). We have that
\begin{align*}
 \frac{\pi_n(\by)}{\pi_n(\bx)} = \exp\left(\sum_{i = 1}^n \alpha_i y_i - \sum_{i = 1}^n \alpha_i x_i\right) &=      \exp(-2 \alpha_j).
 \end{align*}
From this, we can deduce that
\begin{align*}
 \frac{c_n(\bx) / 2}{c_{n, -1}(\bx)} &= \frac{\frac{1}{2} \sum_{\bx' \in \neigh_{-1}(\bx)} g(\exp(-2\alpha_j)) + \frac{1}{2} \sum_{\bx' \in \neigh_{+1}(\bx)} g(\exp(2\alpha_j))}{\sum_{\bx' \in \neigh_{-1}(\bx)} g(\exp(-2\alpha_j))} \cr
  &= \frac{1}{2} + \frac{1}{2} \frac{\sum_{\bx' \in \neigh_{+1}(\bx)} g(\exp(2\alpha_j))}{\sum_{\bx' \in \neigh_{-1}(\bx)} g(\exp(-2\alpha_j))}.
\end{align*}
Also,
\begin{align*}
 \frac{c_n(\by) / 2}{c_{n, +1}(\by)} &= \frac{\frac{1}{2} \sum_{\by' \in \neigh_{-1}(\by)} g(\exp(-2\alpha_j)) + \frac{1}{2} \sum_{\by' \in \neigh_{+1}(\bx)} g(\exp(2\alpha_j))}{\sum_{\by' \in \neigh_{+1}(\by)} g(\exp(2\alpha_j))} \cr
  &= \frac{1}{2} + \frac{1}{2} \frac{\sum_{\by' \in \neigh_{-1}(\by)} g(\exp(-2\alpha_j))}{\sum_{\by' \in \neigh_{+1}(\by)} g(\exp(2\alpha_j))}.
\end{align*}

All that implies that
\begin{align*}
 &\frac{c_n(\bx)/2}{c_{n, \nu}(\bx)} \wedge \frac{c_n(\by) / 2}{c_{n, -\nu}(\by)} \cr
  &\qquad = \frac{1}{2} + \frac{1}{2} \frac{\sum_{\bx' \in \neigh_{+1}(\bx)} g(\exp(2\alpha_j))}{\sum_{\bx' \in \neigh_{-1}(\bx)} g(\exp(-2\alpha_j))} \wedge \frac{1}{2} + \frac{1}{2} \frac{\sum_{\by' \in \neigh_{-1}(\by)} g(\exp(-2\alpha_j))}{\sum_{\by' \in \neigh_{+1}(\by)} g(\exp(2\alpha_j))} \cr
  &\qquad = \frac{1}{2} + \left(\frac{1}{2} \frac{\sum_{\bx' \in \neigh_{+1}(\bx)} g(\exp(2\alpha_j))}{\sum_{\bx' \in \neigh_{-1}(\bx)} g(\exp(-2\alpha_j))} \wedge \frac{1}{2} \frac{\sum_{\by' \in \neigh_{-1}(\by)} g(\exp(-2\alpha_j))}{\sum_{\by' \in \neigh_{+1}(\by)} g(\exp(2\alpha_j))} \right) \cr
  &\qquad \geq \frac{1}{2}.
\end{align*}

We can exploit the structure of the model studied in \autoref{sec:simple_Ising} to obtain a more explicit expression of the first term in the parentheses (we can proceed similarly with the second term). For that, it will be useful to introduce notation. Let $\uparrow(\bx)$ be the states $\bx' \in \neigh(\bx)$ that increase the value of $\pi_n$, compared with $\bx$. Analogously, let $\downarrow(\bx)$ be the states $\bx' \in \neigh(\bx)$ that decrease the value of $\pi_n$.  Note that $|\uparrow(\bx)| = d(\bx)$. We have that
\begin{align*}
    \frac{\sum_{\bx' \in \neigh_{+1}(\bx)} g(\exp(2\alpha_j))}{\sum_{\bx' \in \neigh_{-1}(\bx)} g(\exp(-2\alpha_j))} &= \frac{|\neigh_{+1}(\bx) \cap \downarrow(\bx)| g(\exp(-2c)) + |\neigh_{+1}(\bx) \cap \uparrow(\bx)| g(\exp(2c))}{|\neigh_{-1}(\bx) \cap \downarrow(\bx)| g(\exp(-2c)) + |\neigh_{-1}(\bx) \cap \uparrow(\bx)| g(\exp(2c))} \cr
    &= \frac{|\neigh_{+1}(\bx) \cap \downarrow(\bx)| \exp(-2c) + |\neigh_{+1}(\bx) \cap \uparrow(\bx)| }{|\neigh_{-1}(\bx) \cap \downarrow(\bx)| \exp(-2c) + |\neigh_{-1}(\bx) \cap \uparrow(\bx)|} \cr
    &= \frac{(n / 2 - |\neigh_{-1}(\bx) \cap \uparrow(\bx)|) \exp(-2c) + (n / 2 -  |\neigh_{-1}(\bx) \cap \downarrow(\bx)|) }{|\neigh_{-1}(\bx) \cap \downarrow(\bx)| \exp(-2c) + |\neigh_{-1}(\bx) \cap \uparrow(\bx)|} \cr
    &= \frac{n / 2(1 + \exp(-2c)) - (|\neigh_{-1}(\bx) \cap \downarrow(\bx)| + |\neigh_{-1}(\bx) \cap \uparrow(\bx)| \exp(-2c)) }{|\neigh_{-1}(\bx) \cap \downarrow(\bx)| \exp(-2c) + |\neigh_{-1}(\bx) \cap \uparrow(\bx)|}.
\end{align*}
In the second equality, we used that $g(x) / g(1 / x) = x$ for $x > 0$. In the third equality, we used that $n = |\neigh_{+1}(\bx) \cap \downarrow(\bx)| + |\neigh_{+1}(\bx) \cap \uparrow(\bx)| + |\neigh_{-1}(\bx) \cap \downarrow(\bx)| + |\neigh_{-1}(\bx) \cap \uparrow(\bx)|$ with $|\neigh_{+1}(\bx) \cap \uparrow(\bx)| + |\neigh_{-1}(\bx) \cap \downarrow(\bx)| = n/2$ (a result that is proved below), implying that $|\neigh_{+1}(\bx) \cap \downarrow(\bx)| + |\neigh_{-1}(\bx) \cap \uparrow(\bx)| = n/2$. Note that we cannot have $|\neigh_{-1}(\bx)| = 0$ as mentioned above, implying that the denominator is greater than 0.

We used that $|\neigh_{+1}(\bx) \cap \uparrow(\bx)| + |\neigh_{-1}(\bx) \cap \downarrow(\bx)| = n/2$. This follows from the fact that there are $n/2$ indices in the external field with $\alpha_i = c$. Indeed, for each of these $n/2$ indices, either the spin is aligned with the external field (there are $|\neigh_{-1}(\bx) \cap \downarrow(\bx)|$ such indices), or the spin is not aligned with the external field (there are $|\neigh_{+1}(\bx) \cap \uparrow(\bx)|$ such indices).

Also, we have that $0 \leq |\neigh_{-1}(\bx) \cap \downarrow(\bx)| \leq n/2$ and $0 \leq |\neigh_{-1}(\bx) \cap \uparrow(\bx)| \leq n/2$ (but we cannot have $0 = |\neigh_{-1}(\bx) \cap \downarrow(\bx)| =  |\neigh_{-1}(\bx) \cap \uparrow(\bx)|$), and these variables are independent, in the sense that it is possible for example to have $|\neigh_{-1}(\bx) \cap \downarrow(\bx)| = n/2$ and $|\neigh_{-1}(\bx) \cap \uparrow(\bx)| = n/2$ simultaneously. When $|\neigh_{-1}(\bx) \cap \downarrow(\bx)| = n/2$ and $|\neigh_{-1}(\bx) \cap \uparrow(\bx)| = n/2$,
\[
 \frac{n / 2(1 + \exp(-2c)) - (|\neigh_{-1}(\bx) \cap \downarrow(\bx)| + |\neigh_{-1}(\bx) \cap \uparrow(\bx)| \exp(-2c)) }{|\neigh_{-1}(\bx) \cap \downarrow(\bx)| \exp(-2c) + |\neigh_{-1}(\bx) \cap \uparrow(\bx)|} = 0.
\]

Putting all those results together yields
\begin{align*}
 P_{\rev, n}(\bx, \by) \geq P_{\MH, n}(\bx, \by) \left(\frac{1}{2} + \left(\frac{1}{2} \frac{\sum_{\bx' \in \neigh_{+1}(\bx)} g(\exp(2\alpha_j))}{\sum_{\bx' \in \neigh_{-1}(\bx)} g(\exp(-2\alpha_j))} \wedge \frac{1}{2} \frac{\sum_{\by' \in \neigh_{-1}(\by)} g(\exp(-2\alpha_j))}{\sum_{\by' \in \neigh_{+1}(\by)} g(\exp(2\alpha_j))} \right)\right),
\end{align*}
and the minimum value of
\[
 \left(\frac{1}{2} \frac{\sum_{\bx' \in \neigh_{+1}(\bx)} g(\exp(2\alpha_j))}{\sum_{\bx' \in \neigh_{-1}(\bx)} g(\exp(-2\alpha_j))} \wedge \frac{1}{2} \frac{\sum_{\by' \in \neigh_{-1}(\by)} g(\exp(-2\alpha_j))}{\sum_{\by' \in \neigh_{+1}(\by)} g(\exp(2\alpha_j))} \right)
\]
is 0.

To prove that this bound cannot essentially be improved we analyse the probability of transitions when the current state is $\bx = (+1, \ldots, +1)$. In this case, the only possible transitions are to $\by$ with one of the components equal to $-1$. We thus have $\neigh_{-1}(\bx) = \neigh(\bx)$. Therefore,
\begin{align}\label{eq:relation_rev_MH}
 P_{\rev, n}(\bx, \by) &= \frac{1}{2} \frac{g(\pi_n(\by) / \pi_n(\bx))}{c_{n}(\bx)} \left(1 \wedge \frac{c_{n}(\bx)}{c_{n}(\by)}\frac{c_{n}(\by)}{c_{n, -\nu}(\by)}\right) \cr
 &= \frac{1}{2} \frac{g(\pi_n(\by) / \pi_n(\bx))}{c_{n}(\bx)} \left(1 \wedge \frac{c_{n}(\bx)}{c_{n}(\by)}\right) \frac{\left(1 \wedge \frac{c_{n}(\bx)}{c_{n}(\by)}\frac{c_{n}(\by)}{c_{n, -\nu}(\by)}\right)}{ \left(1 \wedge \frac{c_{n}(\bx)}{c_{n}(\by)}\right) } \cr
 &\leq \frac{1}{2} P_{\MH, n}(\bx, \by) \left(1 \wedge \frac{c_{n}(\bx)}{c_{n}(\by)}\right)^{-1}.
\end{align}
The lower bound of $\frac{c_n(\bx)}{c_n(\by)}$ is attained when the coordinate that is modified from $\bx$ to $\by$ yields a decrease in $\pi_n$. This implies that $|\uparrow(\by)| = |\uparrow(\bx)| + 1$ and $|\downarrow(\by)| = |\downarrow(\bx)| - 1$. Therefore,
 \begin{align*}
  \frac{c_n(\bx)}{c_n(\by)} &= \frac{\sum_{\bx' \in \uparrow(\bx)} g(\exp(2c)) + \sum_{\bx' \in \downarrow(\bx)} g(\exp(-2c))}{\sum_{\bx' \in \uparrow(\by)} g(\exp(2c)) + \sum_{\bx' \in \downarrow(\by)} g(\exp(-2c))} \cr
  &= \frac{g(\exp(2c)) |\uparrow(\bx)| + g(\exp(-2c)) |\downarrow(\bx)|}{g(\exp(2c)) |\uparrow(\by)| + g(\exp(-2c)) |\downarrow(\by)|} \cr
  &= 1 - \frac{g(\exp(2c)) - g(\exp(-2c))}{g(\exp(2c)) |\uparrow(\by)| + g(\exp(-2c)) |\downarrow(\by)|} \cr
  &\geq 1 - \frac{g(\exp(2c)) - g(\exp(-2c))}{n g(\exp(-2c))},
 \end{align*}
using that $g(\exp(2c)) \geq g(\exp(-2c))$ and $n = |\uparrow(\by)| + |\downarrow(\by)|$. When $n > \exp(2c) -1$, we have that
\[
 0 < \frac{g(\exp(2c)) - g(\exp(-2c))}{n g(\exp(-2c))} < 1,
\]
using that $g(x) / g(1 / x) = x$ for $x > 0$. Therefore,
\[
 P_{\rev, n}(\bx, \by) \leq \frac{1}{2} P_{\MH, n}(\bx, \by) \left(1 - \frac{g(\exp(2c)) - g(\exp(-2c))}{n g(\exp(-2c))}\right)^{-1}.
\]
\end{proof}

\begin{proof}[Proof of \autoref{prop:simple_ising_2}]
 To prove the result, we study the normalizing constants $c_n(\bx)$ and $c_{n,\nu}(\bx)$, and more precisely, their relation. For that, it will be useful to exploit \eqref{eq:simple_Ising}. When proposing $\bx'$ from $\bx$, a coordinate, say $x_j$, changes from $-1$ to $+1$, or the opposite, implying that there exists $j$ such that
\begin{align}
 \frac{\pi_n(\bx')}{\pi_n(\bx)} = \exp\left(\sum_{i = 1}^n \alpha_i x_i' - \sum_{i = 1}^n \alpha_i x_i\right) &= \exp\left(\alpha_j (x_j' - x_j)\right) \nonumber \\
 &= \begin{cases}
  \exp(2 \alpha_j) \quad \text{if} \quad x_j = -1, \cr
      \exp(-2 \alpha_j) \quad \text{if} \quad x_j = +1.
 \end{cases} \label{eqn:proof_prop_simple_ising_1}
\end{align}
Therefore,
\[
 c_n(\bx) = \sum_{\bx' \in \neigh(\bx)} g\left(\frac{\pi(\bx')}{\pi(\bx)}\right) = \sum_{\bx' \in \neigh_{-1}(\bx)} g(\exp(-2\alpha_j)) +  \sum_{\bx' \in \neigh_{+1}(\bx)} g(\exp(2\alpha_j)).
\]
Also,
\[
 c_{n, -1}(\bx) = \sum_{\bx' \in \neigh_{-1}(\bx)} g\left(\frac{\pi(\bx')}{\pi(\bx)}\right) = \sum_{\bx' \in \neigh_{-1}(\bx)} g(\exp(-2\alpha_j)).
\]

We are now ready to analyse ratio of normalizing constants, such as $c_{n, -1}(\bx) /  (c_n(\bx) / 2)$, and provide bounds. Using similar arguments as below, we obtain the same bounds for the ratio $c_{n, +1}(\bx) /  (c_n(\bx) / 2)$. We have that
\begin{align*}
 \frac{c_{n,-1}(\bx)}{c_n(\bx) / 2} &= \frac{\sum_{\bx' \in \neigh_{-1}(\bx)} g(\exp(-2\alpha_j))}{\frac{1}{2} \sum_{\bx' \in \neigh_{-1}(\bx)} g(\exp(-2\alpha_j)) + \frac{1}{2} \sum_{\bx' \in \neigh_{+1}(\bx)} g(\exp(2\alpha_j))} \cr
 &= 1 + \frac{\frac{1}{2}\sum_{\bx' \in \neigh_{-1}(\bx)} g(\exp(-2\alpha_j)) - \frac{1}{2} \sum_{\bx' \in \neigh_{+1}(\bx)} g(\exp(2\alpha_j))}{\frac{1}{2} \sum_{\bx' \in \neigh_{-1}(\bx)} g(\exp(-2\alpha_j)) + \frac{1}{2} \sum_{\bx' \in \neigh_{+1}(\bx)} g(\exp(2\alpha_j))} \cr
 &= 1 + \frac{\frac{1}{2}\left(\sum_{\bx' \in \neigh_{-1}(\bx)} g(\exp(-2\alpha_j)) - \sum_{\bx' \in \neigh_{+1}(\bx)} g(\exp(2\alpha_j))\right)}{c_n(\bx) / 2}.
\end{align*}

Let us analyse the numerator of the ratio in more detail. We have
\begin{align*}
 &\sum_{\bx' \in \neigh_{-1}(\bx)} g(\exp(-2\alpha_j)) - \sum_{\bx' \in \neigh_{+1}(\bx)} g(\exp(2\alpha_j)) \cr
 &\quad= \sum_{\bx' \in \neigh_{-1}(\bx) \cap \uparrow(\bx)} g(\exp(2c)) + \sum_{\bx' \in \neigh_{-1}(\bx) \cap \downarrow(\bx)} g(\exp(-2c)) \cr
 &\qquad- \sum_{\bx' \in \neigh_{+1}(\bx) \cap \uparrow(\bx)} g(\exp(2c)) - \sum_{\bx' \in \neigh_{+1}(\bx) \cap \downarrow(\bx)} g(\exp(-2c)) \cr
 &\quad = g(\exp(2c)) \left(|\neigh_{-1}(\bx) \cap \uparrow(\bx)| - |\neigh_{+1}(\bx) \cap \uparrow(\bx)|\right) \cr
 &\qquad + g(\exp(-2c)) \left(|\neigh_{-1}(\bx) \cap \downarrow(\bx)| - |\neigh_{+1}(\bx) \cap \downarrow(\bx)|\right) \cr
 &\quad = \left(g(\exp(2c)) - g(\exp(-2c))\right) \left(|\neigh_{-1}(\bx) \cap \uparrow(\bx)| - |\neigh_{+1}(\bx) \cap \uparrow(\bx)|\right) \cr
 &\qquad + g(\exp(-2c)) \left(|\neigh_{-1}(\bx)| - |\neigh_{+1}(\bx)|\right),
 \end{align*}
 using that $|\neigh_{\nu}(\bx) \cap \downarrow(\bx)| = |\neigh_{\nu}(\bx)| - |\neigh_{\nu}(\bx) \cap \uparrow(\bx)|$ (see the proof of \autoref{prop:simple_ising_1} for the definitions of $\uparrow(\bx)$ and $\downarrow(\bx)$).

 Therefore,
 \begin{align*}
  &\left|\sum_{\bx' \in \neigh_{-1}(\bx)} g(\exp(-2\alpha_j)) - \sum_{\bx' \in \neigh_{+1}(\bx)} g(\exp(2\alpha_j))\right| \cr
   &\quad\leq \left(g(\exp(2c)) - g(\exp(-2c))\right) \left||\neigh_{-1}(\bx) \cap \uparrow(\bx)| - |\neigh_{+1}(\bx) \cap \uparrow(\bx)|\right| \cr
   &\qquad +  g(\exp(-2c)) \left||\neigh_{-1}(\bx)| - |\neigh_{+1}(\bx)|\right|.
 \end{align*}

 We now explain how we get a bound in terms of $d(\bx)$. We first explain that
 \begin{align}\label{eqn:proof_prop_simple_ising_2}
  \left||\neigh_{-1}(\bx) \cap \uparrow(\bx)| - |\neigh_{+1}(\bx) \cap \uparrow(\bx)|\right| \leq d(\bx).
 \end{align}
 Let us consider the case where
 \[
  |\neigh_{-1}(\bx) \cap \uparrow(\bx)| \geq |\neigh_{+1}(\bx) \cap \uparrow(\bx)|.
 \]
 The explanation for the other case is analogous. Using that
 \[
  |\neigh_{-1}(\bx) \cap \uparrow(\bx)| = |\uparrow(\bx)| - |\neigh_{+1}(\bx) \cap \uparrow(\bx)| = d(\bx) - |\neigh_{+1}(\bx) \cap \uparrow(\bx)|,
  \]
   we have that
 \begin{align*}
  |\neigh_{-1}(\bx) \cap \uparrow(\bx)| - |\neigh_{+1}(\bx) \cap \uparrow(\bx)| = d(\bx) - 2 |\neigh_{+1}(\bx) \cap \uparrow(\bx)| \leq d(\bx).
 \end{align*}

 We now explain that
 \[
  \left||\neigh_{-1}(\bx)| - |\neigh_{+1}(\bx)|\right| \leq 2 d(\bx).
 \]
 Using that $n = |\neigh_{-1}(\bx)| + |\neigh_{+1}(\bx)|$ and that $|\neigh_{-1}(\bx)| = |\neigh_{-1}(\bx) \cap \uparrow(\bx)| + |\neigh_{-1}(\bx) \cap \downarrow(\bx)|$, we have that
 \begin{align*}
  &\left||\neigh_{-1}(\bx)| - |\neigh_{+1}(\bx)|\right| = \left|2 |\neigh_{-1}(\bx)| - n\right| \cr
  &= \left|2\left(|\neigh_{-1}(\bx) \cap \uparrow(\bx)| - |\neigh_{+1}(\bx) \cap \uparrow(\bx)| + |\neigh_{+1}(\bx) \cap \uparrow(\bx)| + |\neigh_{-1}(\bx) \cap \downarrow(\bx)|\right) - n\right|.
 \end{align*}
  Recall that $|\neigh_{+1}(\bx) \cap \uparrow(\bx)| + |\neigh_{-1}(\bx) \cap \downarrow(\bx)| = n/2$ (see the proof of \autoref{prop:simple_ising_1}). Consequently, following \eqref{eqn:proof_prop_simple_ising_2},
  \begin{align*}
  &\left||\neigh_{-1}(\bx)| - |\neigh_{+1}(\bx)|\right| = \left|2|\neigh_{-1}(\bx) \cap \uparrow(\bx)| - |\neigh_{+1}(\bx) \cap \uparrow(\bx)|\right|\leq 2d(\bx).
 \end{align*}

Therefore,
 \begin{align*}
  &\left|\sum_{\bx' \in \neigh_{-1}(\bx)} g(\exp(-2\alpha_j)) - \sum_{\bx' \in \neigh_{+1}(\bx)} g(\exp(2\alpha_j))\right| \cr
   &\quad\leq \left(g(\exp(2c)) - g(\exp(-2c))\right) d(\bx) + 2 g(\exp(-2c)) d(\bx) = \left(g(\exp(2c)) + g(\exp(-2c))\right) d(\bx),
 \end{align*}
 which concludes the proof.
\end{proof}

\begin{proof}[Proof of \autoref{prop:simple_ising_3}]
    We have that
    \[
     \pi_n(d(\bx) = k) = \frac{{n \choose k} \exp(-2ck)}{(1 + \exp(-2c))^n} = {n \choose k}\left(\frac{\exp(-2c)}{1+\exp(-2c)}\right)^k \left(\frac{1}{1+\exp(-2c)}\right)^{n - k}.
    \]
    Therefore, $d(\bx)$ has a binomial distribution with parameters $n$ and $p = \frac{\exp(-2c)}{1+\exp(-2c)}$. Let us define a sequence of independent random variables $\{Y_i\}_{i=1}^n$ with each of them following a Bernoulli distribution of parameter $p$.

    Consider that $\lfloor\phi(n)\rfloor \leq np$. We have that
    \begin{align*}
     1 - \pi_n(\itXset_n) = \Prob(d(\bx) \geq \lfloor\phi(n)\rfloor) = \Prob\left(\sum_{i=1}^n Y_i \geq \lfloor\phi(n)\rfloor\right) &\geq \Prob\left(\sum_{i=1}^n Y_i \geq np\right) \cr
      &= \Prob\left(\sqrt{n}\left(\frac{1}{n}\sum_{i=1}^n Y_i - p\right) \geq 0\right) \rightarrow \frac{1}{2},
    \end{align*}
    by the central limit theorem. This allows to establish the first part of \autoref{prop:simple_ising_3}.

    Now, consider that $n > \lfloor\phi(n)\rfloor > np$ with $\lfloor\phi(n)\rfloor / n - p$ converging towards a positive constant. Hoeffding's inequality indicates that
    \begin{align*}
        \Prob(d(\bx) \geq \lfloor\phi(n)\rfloor) = \Prob\left(\sum_{i=1}^n Y_i \geq \lfloor\phi(n)\rfloor\right) &= \Prob\left(\sum_{i=1}^n Y_i - np \geq \lfloor\phi(n)\rfloor - np\right) \cr
        &\leq \exp(-2(\lfloor\phi(n)\rfloor - np)^2 /n) \cr
        &= \exp(-2n(\lfloor\phi(n)\rfloor/n - p)^2).
    \end{align*}
 Therefore, $1 - \pi_n(\itXset_n)$ converges to 0 at an exponential rate.
\end{proof}

\begin{proof}[Proof of \autoref{prop:simple_ising_4}]
 We saw in the proof of Proposition 2, that, for any $\bx, \by$ such that $\by \in \neigh(\bx)$ and $\by \neq \bx$,
 \begin{align*}
 P_{\rev, n}(\bx, \by) &\geq P_{\MH, n}(\bx, \by) \left(\frac{1}{2} + \left(\frac{1}{2} \frac{\sum_{\bx' \in \neigh_{+1}(\bx)} g(\exp(2\alpha_j))}{\sum_{\bx' \in \neigh_{-1}(\bx)} g(\exp(-2\alpha_j))} \wedge \frac{1}{2} \frac{\sum_{\by' \in \neigh_{-1}(\by)} g(\exp(-2\alpha_j))}{\sum_{\by' \in \neigh_{+1}(\by)} g(\exp(2\alpha_j))} \right)\right),
\end{align*}
with
 \begin{align*}
    \frac{\sum_{\bx' \in \neigh_{+1}(\bx)} g(\exp(2\alpha_j))}{\sum_{\bx' \in \neigh_{-1}(\bx)} g(\exp(-2\alpha_j))} &=  \frac{n / 2(1 + \exp(-2c)) - (|\neigh_{-1}(\bx) \cap \downarrow(\bx)| + |\neigh_{-1}(\bx) \cap \uparrow(\bx)| \exp(-2c)) }{|\neigh_{-1}(\bx) \cap \downarrow(\bx)| \exp(-2c) + |\neigh_{-1}(\bx) \cap \uparrow(\bx)|},
\end{align*}
when we obtain $\by$ from $\bx$ by changing a coordinate, say $x_j$, from $+1$ to $-1$, implying that the transition is associated with $\nu = -1$ (see the proof of \autoref{prop:simple_ising_1} for the definitions of $\uparrow(\bx)$ and $\downarrow(\bx)$). We can obtain the same lower bound in the opposite situation (when the transition is associated with $\nu = +1$).

We have that $0 \leq |\neigh_{-1}(\bx) \cap \downarrow(\bx)| \leq n/2$ and
\[
 0 \leq |\neigh_{-1}(\bx) \cap \uparrow(\bx)| \leq n/2 \wedge \lfloor \phi(n) \rfloor,
\]
but we cannot have $0 = |\neigh_{-1}(\bx) \cap \downarrow(\bx)| = |\neigh_{-1}(\bx) \cap \uparrow(\bx)|$ as explained in the proof of \autoref{prop:simple_ising_1}. We now show that the lower bound of
\[
 \frac{n / 2(1 + \exp(-2c)) - (|\neigh_{-1}(\bx) \cap \downarrow(\bx)| + |\neigh_{-1}(\bx) \cap \uparrow(\bx)| \exp(-2c)) }{|\neigh_{-1}(\bx) \cap \downarrow(\bx)| \exp(-2c) + |\neigh_{-1}(\bx) \cap \uparrow(\bx)|}
\]
is attained when $|\neigh_{-1}(\bx) \cap \downarrow(\bx)|$ and $|\neigh_{-1}(\bx) \cap \uparrow(\bx)|$ are at their upper bounds. To achieve this, we consider $|\neigh_{-1}(\bx) \cap \downarrow(\bx)|$ and $|\neigh_{-1}(\bx) \cap \uparrow(\bx)|$ as continuous variables, given by $x$ and $y$, respectively, and calculate the derivatives of the log of the function.

We have that
\begin{align*}
 &\frac{\partial}{\partial x} \log\left(\frac{n / 2(1 + \exp(-2c)) - (x + y \exp(-2c)) }{x \exp(-2c) + y}\right) \cr
 &\qquad = -\frac{1}{n / 2(1 + \exp(-2c)) - (x + y \exp(-2c))} - \frac{\exp(-2c)}{x \exp(-2c) + y},
\end{align*}
and
\begin{align*}
 &\frac{\partial}{\partial y} \log\left(\frac{n / 2(1 + \exp(-2c)) - (x + y \exp(-2c)) }{x \exp(-2c) + y}\right) \cr
 &\qquad = -\frac{\exp(-2c)}{n / 2(1 + \exp(-2c)) - (x + y \exp(-2c))} - \frac{1}{x \exp(-2c) + y}.
\end{align*}
We also have that
\[
0 \leq n / 2(1 + \exp(-2c)) - (x + y \exp(-2c))  < n / 2(1 + \exp(-2c)),
\]
given that $0 < x, y \leq n/2$,
with the lower bound of
\[
 n / 2(1 + \exp(-2c)) - (x + y \exp(-2c))
\]
 that is attained when $x = y = n/2$. Therefore, the partial derivatives are strictly negative on $0 < x, y < n/2$ (and when either $x$ or $y$ is equal to $n/2$ and the other variable is smaller than $n/2$), and go to $-\infty$ when $ x, y \rightarrow  n/2$. Recall that
\[
 \frac{n / 2(1 + \exp(-2c)) - (|\neigh_{-1}(\bx) \cap \downarrow(\bx)| + |\neigh_{-1}(\bx) \cap \uparrow(\bx)| \exp(-2c)) }{|\neigh_{-1}(\bx) \cap \downarrow(\bx)| \exp(-2c) + |\neigh_{-1}(\bx) \cap \uparrow(\bx)|} = 0
\]
is attained when $|\neigh_{-1}(\bx) \cap \downarrow(\bx)| = |\neigh_{-1}(\bx) \cap \uparrow(\bx)| = n / 2$.

Consequently, the lower bound on $\tXset_n$ is given by
\[
 \frac{n / 2(1 + \exp(-2c)) - \left(n/2 + \lfloor n\frac{\exp(-2c)}{1 + \exp(-2c)}(1 + \varepsilon) + 1\rfloor \exp(-2c)\right) }{(n / 2) \exp(-2c) + \lfloor n\frac{\exp(-2c)}{1 + \exp(-2c)}(1 + \varepsilon) + 1\rfloor}= \frac{1 - 2  \frac{\lfloor n\frac{\exp(-2c)}{1 + \exp(-2c)}(1 + \varepsilon) + 1\rfloor}{n}}{1 + 2 \frac{\lfloor n\frac{\exp(-2c)}{1 + \exp(-2c)}(1 + \varepsilon) + 1\rfloor}{n \exp(-2c)}}.
\]
Note that we can obtain the same bound for
\[
  \frac{\sum_{\by' \in \neigh_{-1}(\by)} g(\exp(-2\alpha_j))}{\sum_{\by' \in \neigh_{+1}(\by)} g(\exp(2\alpha_j))}.
\]
Therefore,
\begin{align*}
 P_{\rev, n}(\bx, \by) &\geq \omega(n) P_{\MH, n}(\bx, \by),
\end{align*}
with
\[
\omega(n) = \frac{1}{2} + \frac{1}{2} \frac{1 - 2  \frac{\lfloor n\frac{\exp(-2c)}{1 + \exp(-2c)}(1 + \varepsilon) + 1\rfloor}{n}}{1 + 2 \frac{\lfloor n\frac{\exp(-2c)}{1 + \exp(-2c)}(1 + \varepsilon) + 1\rfloor}{n \exp(-2c)}} \rightarrow \frac{1}{2} + \frac{1}{2} \frac{1 - 2  \frac{\exp(-2c)(1 + \varepsilon)}{1 + \exp(-2c)}}{1 + 2 \frac{1 + \varepsilon}{1 + \exp(-2c)}} = \overline{\omega}.
\]
\end{proof}

Before presenting the proof of \autoref{prop:simple_ising_5}, we briefly explain how we proceed, which will motivate the introduction of a lemma. To prove \autoref{prop:simple_ising_5}, we first derive a lower bound on the right (and absolute) spectral gap of a lazy version of Glauber dynamics when used to sample from $\pi_n$ (recall \eqref{eq:simple_Ising}). Denote by $P_{\text{G}, n}^{(\text{L})}$ the Markov kernel associated to the produced Markov chain. Next, we provide an order between $P_{\MH, n}$ and $P_{\text{G}, n}^{(\text{L})}$. This allows to have an order on the right spectral gaps using Theorem 2 in \cite{zanella2019informed}: if $P_{\MH, n}(\bx, \by) \geq \omega(n) P_{\text{G}, n}^{(\text{L})}(\bx, \by)$ for all $(\bx, \by) \in \Xset_n^2$ with $\bx \neq \by$, then the spectral gap of $P_{\MH, n}$ is lower bounded by $\omega(n)$ times that of $P_{\text{G}, n}^{(\text{L})}$. The lower bound of the spectral gap of $P_{\text{G}, n}^{(\text{L})}$ thus allows to characterize the speed decay of the right spectral gap of $P_{\MH, n}$, as well as that of $P_{\rev, n}$ by \autoref{prop:simple_ising_1}.

We introduce the lower bound of the spectral gap of $P_{\text{G}, n}^{(\text{L})}$ in a lemma, but beforehand, we describe the algorithm which is in fact valid whenever $\pi_n$ factorizes, i.e.
\[
 \pi_n(\bx) = \prod_{i = 1}^n f_i(x_i), \quad \bx \in \re^n,
\]
and it is possible to sample from the marginal distributions. Note that it is the case for $\pi_n$ defined in \eqref{eq:simple_Ising}.

\begin{algorithm}[ht]
\caption{Glauber dynamics for factorized target distributions} \label{algo:Glauber}

 \begin{itemize}

  \item[1.] Sample $i \sim \mathcal{U}\{1, \ldots, n\}$.

  \item[2.] Sample $w \sim f_i$.

  \item[3.] Set the next state of the chain, denoted by $\bx'$, as follows: $x_j' = x_j$ for all $j \neq i$, and $x'_i = w$, where $\bx$ denotes the current state.

  \item[4.] Go to Step 1.

 \end{itemize}

\end{algorithm}

We set $P_{\text{G}, n}^{(\text{L})}$ to be the lazy version of the transition kernel of the Markov chain simulated by \autoref{algo:Glauber}.

\begin{Lemma}\label{lemma:Glauber}
$P_{\text{G}, n}^{(\text{L})}$ has a spectral gap with a lower bound that decreases as $n$ increases at a speed of $n \log n$.
\end{Lemma}

\begin{proof}
 To prove the result, we provide a mixing-time bound and use it to derive a bound on the spectral gap. The mixing time is defined as follows:
 \[
  t_{\text{mix}}(\varepsilon) := \inf \left\{t \in \na: \sup_{\bx} \|P_{\text{G}, n}^{(\text{L})}( \bx, \cdot \,)^t - \pi_n\|_{\tv} \leq \varepsilon\right\},
 \]
 with $0 < \varepsilon < 1$. To identify a bound, we use that
 \[
  \|P_{\text{G}, n}^{(\text{L})}( \bx, \cdot \,)^t - \pi_n\|_{\tv} \leq \sup_{\bx_0, \by_0} \Prob_{\bx_0, \by_0}(\bX_t \neq \bY_t),
 \]
 with $\Prob_{\bx_0, \by_0}$ being a joint distribution of $\{\bX_t\}$ and $\{\bY_t\}$ with $\bX_0 = \bx_0$ and $\bY_0 = \by_0$, and $\Prob_{\bx_0, \by_0}(\bX_t \in A) = P_{\text{G}, n}^{(\text{L})}( \bx_0, A)^t$ and $\Prob_{\bx_0, \by_0}(\bY_t \in A) = P_{\text{G}, n}^{(\text{L})}( \by_0, A)^t$, for all $t$ \citep[Corollary 5.5]{levin2017markov}. Define $X_{t,i}$ and $Y_{t,i}$ to be the $i$-th components of $\bX_t$ and $\bY_t$, respectively.

 We now describe a transition of $\{\bX_t\}$ and $\{\bY_t\}$. Let us consider that the current states are $\bX_t = \bx$ and $\bY_t = \by$. With probability $0.5$, set $\bX_{t+1} = \bx$ and $\bY_{t+1} = \by$. With probability $0.5$, sample $i \sim \mathcal{U}\{1, \ldots, n\}$ and $w \sim f_i$, and set $X_{t+1,j} = x_j$ and $Y_{t+1,j} = y_j$, for all $j \neq i$, and $X_{t+1,i} = Y_{t+1,i} = w$.

 We thus have that, marginally, the transitions of $\{\bX_t\}$ and $\{\bY_t\}$ are lazy versions of that in \autoref{algo:Glauber}. Also, once a component index is selected in a transition of $\{\bX_t\}$ and $\{\bY_t\}$, then $X_{t+1,i}$ and $Y_{t+1,i}$ become equal and remain equal in the following transitions. Therefore,
 \begin{align*}
  \Prob_{\bx_0, \by_0}(\bX_t \neq \bY_t) = \Prob_{\bx_0, \by_0}\left(\cup_{j = 1}^n \{X_{t, j} \neq Y_{t, j}\}\right) &\leq \sum_{j = 1}^n \Prob_{\bx_0, \by_0}\left(X_{t, j} \neq Y_{t, j}\right) \cr
  &= \sum_{j = 1}^n \left(1 - \frac{1}{2 n}\right)^t \cr
  &= n \left(1 - \frac{1}{2 n}\right)^t.
 \end{align*}
  To have
  \[
  n \left(1 - \frac{1}{2 n}\right)^t \leq \varepsilon
  \]
  we need to have
  \[
   t \geq \frac{\log(\varepsilon / n)}{\log\left(1 - \frac{1}{2n}\right)} \geq 2n \log(n / \varepsilon),
  \]
  using that $\log\left(1 - \frac{1}{2n}\right) \leq - \frac{1}{2n}$. This implies that $t_{\text{mix}}(\varepsilon) \leq 2n \log(n / \varepsilon)$.

  We conclude the proof using that the spectral gap is lower bounded by
  \[
   \left(2n \log\left(\frac{n}{\varepsilon}\right)\Big/ \log\left(\frac{1}{2 \varepsilon}\right) + 1\right)^{-1},
  \]
  using Proposition 1.1 in \cite{jerison2013general} which provides a lower bound on the mixing time in terms of the absolute spectral gap.
\end{proof}

\begin{proof}[Proof of \autoref{prop:simple_ising_5}]
 To prove this result, we establish that $P_{\MH, n}(\bx, \by) \geq \omega(n) P_{\text{G}, n}^{(\text{L})}(\bx, \by)$ for all $(\bx, \by) \in \Xset_n^2$ with $\bx \neq \by$. This allows to establish a lower bound on the right spectral gap of $P_{\MH, n}$ by Theorem 2 in \cite{zanella2019informed} and using \autoref{lemma:Glauber}. Indeed, combining these results yields a lower bound on the spectral gap of $P_{\MH, n}$ given by
  \[
    \omega(n) \left(2n \log\left(\frac{n}{\varepsilon}\right)\Big/ \log\left(\frac{1}{2 \varepsilon}\right) + 1\right)^{-1},
  \]
  for any $\varepsilon \in (0, 1)$. We obtain a lower bound on the right spectral gap of $P_{\rev, n}$ similarly using Proposition 2.

 Let us consider $\by \neq \bx$, a state that is reachable from $\bx$. Note that the same states $\by$ are reachable from $\bx$ in $P_{\MH, n}$ and $P_{\text{G}, n}^{(\text{L})}$. We have that
 \[
  P_{\text{G}, n}^{(\text{L})}(\bx, \by) \leq \frac{1}{2n} \frac{\ee^c}{\ee^c + \ee^{-c}},
 \]
 using that $f_i(x_i) \leq \frac{\ee^c}{\ee^c + \ee^{-c}}$.

 Now, we find a lower bound on $P_{\MH, n}(\bx, \by)$. We have that
 \[
  P_{\MH, n}(\bx, \by) = \frac{g(\pi_n(\by) / \pi_n(\bx))}{c_n(\bx)} \left(1 \wedge \frac{c_n(\bx)}{c_n(\by)} \right) \geq \frac{g(\exp(-2c))}{c_n(\bx)} \left(1 \wedge \frac{c_n(\bx)}{c_n(\by)} \right),
 \]
 using that $g(\pi_n(\by) / \pi_n(\bx)) \geq g(\exp(-2c))$ as seen in the proof of \autoref{prop:simple_ising_1}. Also,
 \begin{align*}
  c_n(\bx) = \sum_{\bx' \in \uparrow(\bx)} g(\exp(2c)) + \sum_{\bx' \in \downarrow(\bx)} g(\exp(-2c)) \leq n g(\exp(2c)),
 \end{align*}
  with $\uparrow(\bx)$ and $\downarrow(\bx)$ defined as in the proof of \autoref{prop:simple_ising_1}, using that $g(\exp(2c)) \geq g(\exp(-2c))$.

  As shown in the proof of \autoref{prop:simple_ising_1}, the lower bound on $\left(1 \wedge \frac{c_n(\bx)}{c_n(\by)} \right)$ is attained when the coordinate that is modified from $\bx$ to $\by$ yields a decrease in $\pi_n$ and is given by:
  \begin{align*}
  \frac{c_n(\bx)}{c_n(\by)} &= 1 - \frac{g(\exp(2c)) - g(\exp(-2c))}{g(\exp(2c)) |\uparrow(\by)| + g(\exp(-2c)) |\downarrow(\by)|} \cr
  &\geq 1 - \frac{g(\exp(2c)) - g(\exp(-2c))}{n g(\exp(-2c))},
 \end{align*}
using that $g(\exp(2c)) \geq g(\exp(-2c))$ and $n = |\uparrow(\by)| + |\downarrow(\by)|$.

Combining the results above yields
\begin{align*}
 P_{\MH, n}(\bx, \by) &\geq \frac{g(\exp(-2c))}{n g(\exp(2c))} \left(1 - \frac{g(\exp(2c)) - g(\exp(-2c))}{n g(\exp(-2c))} \right) \cr
 & \geq 2\frac{g(\exp(-2c))}{g(\exp(2c))}\left(\frac{\ee^c}{\ee^c + \ee^{-c}}\right)^{-1} \left(1 - \frac{g(\exp(2c)) - g(\exp(-2c))}{n g(\exp(-2c))} \right) P_{\text{G}, n}^{(\text{L})}(\bx, \by),
\end{align*}
which concludes the proof.
\end{proof}

\begin{proof}[Proof of \autoref{prop:simple_ising_6}]
 Let $\bX = (X_1, \ldots, X_n) \sim \pi_n$. We have that $X_1, \ldots, X_n$ are independent random variables. We want to analyse
 \[
  \E\left[\left(\frac{\sum_{i=1}^n X_i - \E\left[\sum_{i=1}^n X_i\right]}{\sqrt{\var\left[\sum_{i=1}^n X_i\right]}}\right)^4\right].
 \]
 We have that
 \[
   \Prob(X_i = +1) = \frac{\ee^{\alpha_i}}{\ee^{\alpha_i} + \ee^{-\alpha_i}} = \frac{\ee^{\alpha_i}}{\ee^{c} + \ee^{-c}} = 1 - \Prob(X_i = -1),
 \]
 using that $\alpha_i$ is either $c$ or $-c$. Therefore, $X_i$ is equal in distribution to $2Y_i - 1$, where $Y_i$ has a Bernoulli distribution with parameter $p_i := \frac{\ee^{\alpha_i}}{\ee^{c} + \ee^{-c}}$. We have that $Y_1, \ldots, Y_n$ are independent random variables.

 Consequently,
 \begin{align*}
  \E\left[\left(\frac{\sum_{i=1}^n X_i - \E\left[\sum_{i=1}^n X_i\right]}{\sqrt{\var\left[\sum_{i=1}^n X_i\right]}}\right)^4\right] &= \E\left[\left(\frac{\sum_{i=1}^n (2Y_i - 1) - \E\left[\sum_{i=1}^n 2Y_i - 1\right]}{\sqrt{\var\left[\sum_{i=1}^n 2Y_i - 1\right]}}\right)^4\right] \cr
  &= \E\left[\left(\frac{\sum_{i=1}^n Y_i - \sum_{i=1}^n\E\left[Y_i\right]}{\sqrt{\sum_{i=1}^n\var\left[Y_i\right]}}\right)^4\right] \cr
  &= \frac{\E[\left(\sum_{i=1}^n Z_i\right)^4]}{\left(\sum_{i=1}^n\var\left[Y_i\right]\right)^2},
 \end{align*}
 where $Z_i := Y_i - \E[Y_i]$.

 We now calculate the numerator and denominator. We have that
 \[
 \sum_{i=1}^n\var\left[Y_i\right] = \sum_{i=1}^n \frac{\ee^{\alpha_i}}{\ee^{c} + \ee^{-c}}\frac{\ee^{-\alpha_i}}{\ee^{c} + \ee^{-c}} = \frac{n}{(\ee^{c} + \ee^{-c})^2}.
 \]
 Also,
 \begin{align*}
  \E\left[\left(\sum_{i=1}^n Z_i\right)^4\right] &= \E\left[\sum_{i,j,k,l} Z_iZ_jZ_kZ_l\right] \cr
  &= \sum_{i=1}^n \E[Z_i^4] + 3 \sum_{i \neq j} \E[Z_i^2] \E[Z_j^2] \cr
  &= \sum_{i=1}^n \E[Z_i^4] + 3n(n-1) \frac{1}{(\ee^{c} + \ee^{-c})^4},
 \end{align*}
 using that $\E[Z_i] = 0$.

 Putting together the results for the numerator and denominator, we have that
 \begin{align*}
  \frac{\E[\left(\sum_{i=1}^n Z_i\right)^4]}{\left(\sum_{i=1}^n\var\left[Y_i\right]\right)^2} &= \frac{\sum_{i=1}^n \E[Z_i^4] + 3n(n-1) \frac{1}{(\ee^{c} + \ee^{-c})^4}}{\frac{n^2}{(\ee^{c} + \ee^{-c})^4}} \cr
  &\rightarrow 3,
 \end{align*}
 using that $0 \leq Z_i^4 \leq 1$. This concludes the proof.
\end{proof}

\begin{proof}[Proof of \autoref{prop_invariance_NRJ}]
  It suffices to prove that the probability to reach the state $\mathbf{y}, \boldsymbol\theta_\mathbf{y}'\in A, \nu'$ in one step is equal to the probability of this state under the target:
 \begin{align}\label{eqn_inv}
  \sum_{\mathbf{x}, \nu}\int \pi(\mathbf{x}, \boldsymbol\theta_\mathbf{x}) \times (1 / 2) \left(\int_{A} P((\mathbf{x}, \boldsymbol\theta_\mathbf{x}, \nu), (\mathbf{y}, \mathrm{d}\boldsymbol\theta_\mathbf{y}', \nu'))\, \right) \, \mathrm{d}\boldsymbol\theta_\mathbf{x}
  &= \int_A \pi(\mathbf{y}, \boldsymbol\theta_\mathbf{y}')\times (1 / 2) \, \mathrm{d}\boldsymbol\theta_\mathbf{y}',
 \end{align}
 where $P$ is the transition kernel. Note that we abuse notation here by denoting the integration variable $ \boldsymbol\theta_\mathbf{y}'$ on the LHS given that we in fact use a vector of auxiliary variables $\mathbf{u}_{\bx \mapsto \by}$ to generate the proposal when switching models, which do not necessarily have the same dimension as $\boldsymbol\theta_\mathbf{y}'$.

  We consider two distinct events: a model switch is proposed, that we denote $S$, and a parameter update is proposed (therefore denoted $S^\mathsf{c}$). We know that the probabilities of these events are $1 - \tau$ and $\tau$, respectively. We rewrite the LHS of \eqref{eqn_inv} as
 \begin{align}\label{eqn_inv_1}
  &\sum_{\mathbf{x}, \nu}\int \pi(\mathbf{x}, \boldsymbol\theta_\mathbf{x}) \times (1 / 2) \left(\int_{A} P((\mathbf{x}, \boldsymbol\theta_\mathbf{x}, \nu), (\mathbf{y}, \mathrm{d}\boldsymbol\theta_\mathbf{y}', \nu'))\, \right) \, \mathrm{d}\boldsymbol\theta_\mathbf{x} \cr
  &\quad = \sum_{\mathbf{x}, \nu} \, (1 - \tau) \int \pi(\mathbf{x}, \boldsymbol\theta_\mathbf{x}) \times (1 / 2) \left(\int_{A} P((\mathbf{x}, \boldsymbol\theta_\mathbf{x}, \nu), (\mathbf{y}, \mathrm{d}\boldsymbol\theta_\mathbf{y}', \nu') \mid S)\right) \, \mathrm{d}\boldsymbol\theta_\mathbf{x} \cr
  &\qquad + \sum_{\mathbf{x}, \nu} \,  \tau \int \pi(\mathbf{x}, \boldsymbol\theta_\mathbf{x}) \times (1 / 2) \left(\int_{A} P((\mathbf{x}, \boldsymbol\theta_\mathbf{x}, \nu), (\mathbf{y}, \mathrm{d}\boldsymbol\theta_\mathbf{y}', \nu') \mid S^\mathsf{c}) \right) \, \mathrm{d}\boldsymbol\theta_\mathbf{x}.
 \end{align}
We analyse the two terms separately. We know that
  \[
   P((\mathbf{x}, \boldsymbol\theta_\mathbf{x}, \nu), (\mathbf{y}, \mathrm{d}\boldsymbol\theta_\mathbf{y}', \nu')\mid S^\mathsf{c}) = \delta_{(\mathbf{x}, \nu)}(\mathbf{y}, \nu')\, P_{S^\mathsf{c}}(\boldsymbol\theta_\mathbf{x}, \mathrm{d}\boldsymbol\theta_\mathbf{y}'),
  \]
  where $P_{S^\mathsf{c}}$ is the transition kernel associated with the method used to update the parameters. Therefore, the second term on the RHS of \eqref{eqn_inv_1} is equal to
  \begin{align*}
   &\tau \sum_{\bx, \nu} \int \pi(\mathbf{x}, \boldsymbol\theta_\mathbf{x}) \times (1 / 2) \left(\int_{A} P((\mathbf{x}, \boldsymbol\theta_\mathbf{x}, \nu), (\mathbf{y}, \mathrm{d}\boldsymbol\theta_\mathbf{y}', \nu') \mid S^\mathsf{c}) \right) \, \mathrm{d}\boldsymbol\theta_\mathbf{x} \cr
   &\quad = \tau \times \pi(\mathbf{y}) \times (1 / 2)\int \pi(\boldsymbol\theta_\mathbf{y} \mid \mathbf{y}) \left(\int_{A} P_{S^\mathsf{c}}(\boldsymbol\theta_\mathbf{y}, \mathrm{d}\boldsymbol\theta_\mathbf{y}') \right) \, \mathrm{d}\boldsymbol\theta_\mathbf{y}.
  \end{align*}
  We also know that $P_{S^\mathsf{c}}$ leaves the conditional distribution $\pi(\, \cdot \mid \mathbf{y})$ invariant, implying that
  \begin{align}\label{eqn_conclusion_up}
   &\tau \times \pi(\mathbf{y}) \times (1 / 2)\int \pi(\boldsymbol\theta_\mathbf{y} \mid \mathbf{y}) \left(\int_{A} P_{S^\mathsf{c}}(\boldsymbol\theta_\mathbf{y}, \mathrm{d}\boldsymbol\theta_\mathbf{y}') \right) \, \mathrm{d}\boldsymbol\theta_\mathbf{y} \cr
   &\quad = \tau \times \pi(\mathbf{y}) \times (1 / 2) \int_A  \pi(\boldsymbol\theta_\mathbf{y}' \mid \mathbf{y}) \, \mathrm{d}\boldsymbol\theta_\mathbf{y}' = \tau \int_A \pi(\mathbf{y}, \boldsymbol\theta_\mathbf{y}')\times (1 / 2) \, \mathrm{d}\boldsymbol\theta_\mathbf{y}'.
  \end{align}

  For the model switching case (the first term on the RHS of \eqref{eqn_inv_1}), we use the fact that there is a connection between $P((\mathbf{x}, \boldsymbol\theta_\mathbf{x}, \nu), (\mathbf{y}, \boldsymbol\theta_\mathbf{y}', \nu')\mid S)$ and the kernel associated to a specific RJ. Consider that in this RJ, $q_{\mathbf{x}}(\mathbf{y}) = (1/2) \, q_{\mathbf{x}, -1}(\mathbf{y}) + (1/2) \, q_{\mathbf{x}, +1}(\mathbf{y})$ for all $\mathbf{x}$ and $\mathbf{y} \in \neigh(\bx)$ and that all other proposal distributions in RJ are the same as in \autoref{algo_NRJ} during model switches. In this case, $\alpha_{\text{RJ}}=\alpha_{\text{NRJ}}$ and it is considered that to go from $\bx$ to $\by$, $q_{\mathbf{x}, \nu}$ is chosen (this happens with probability $1/2$) and, in the reverse move, $q_{\mathbf{y}, -\nu}$ is chosen (which also happens with probability $1/2$).

  We now analyse the first sum on the RHS in \eqref{eqn_inv_1},
  \[
   \sum_{\mathbf{x}, \nu} \, (1 - \tau) \int \pi(\mathbf{x}, \boldsymbol\theta_\mathbf{x}) \times (1 / 2) \left(\int_{A} P((\mathbf{x}, \boldsymbol\theta_\mathbf{x}, \nu), (\mathbf{y}, \mathrm{d}\boldsymbol\theta_\mathbf{y}', \nu') \mid S) \right) \, \mathrm{d}\boldsymbol\theta_\mathbf{x}.
  \]
  First, consider that $\by \in \neigh_\nu(\bx)$, i.e.\ the case of an accepted model switch, thus model $\mathbf{y}$ is reached from model $\mathbf{x} \neq \mathbf{y}$, coming from direction $\nu$ (with $\nu = \nu'$ because the move is accepted). Given the reversibility of RJ, the probability to go from model $\mathbf{x}$ with parameters in $B$ to model $\mathbf{y} \neq \mathbf{x}$ with parameters in $A$ is
  \begin{align}\label{eqn_reversibility_proof}
    &\int_B \pi(\mathbf{x}, \boldsymbol\theta_\mathbf{x}) \left(\int_A P_{\text{RJ}}((\mathbf{x}, \boldsymbol\theta_\mathbf{x}), (\mathbf{y}, \mathrm{d}\boldsymbol\theta_\mathbf{y}')) \right) \, \mathrm{d}\boldsymbol\theta_\mathbf{x} = \int_A \pi(\mathbf{y}, \boldsymbol\theta_\mathbf{y}') \left(\int_B P_{\text{RJ}}((\mathbf{y}, \boldsymbol\theta_\mathbf{y}'), (\mathbf{x}, \mathrm{d}\boldsymbol\theta_\mathbf{x}))\right) \, \mathrm{d}\boldsymbol\theta_\mathbf{y}',
  \end{align}   	
  where $P_{\text{RJ}}$ is the transition kernel of the RJ. Note that
\[
P_{\text{RJ}}((\mathbf{x}, \boldsymbol\theta_\mathbf{x}), (\mathbf{y}, \mathrm{d}\boldsymbol\theta_\mathbf{y}')) = (1/2) \, (1 - \tau) \, P((\mathbf{x}, \boldsymbol\theta_\mathbf{x}, \nu'), (\mathbf{y}, \mathrm{d}\boldsymbol\theta_\mathbf{y}', \nu') \mid S),
\]
  given that the difference between both kernels is that in RJ, it is randomly decided to use $q_{\mathbf{x}, \nu}$; there is thus an additional probability factor of $1/2$. Analogously, we have that $P_{\text{RJ}}((\mathbf{y}, \boldsymbol\theta_\mathbf{y}'), (\mathbf{x}, \mathrm{d}\boldsymbol\theta_\mathbf{x}))= (1/2) \, (1 - \tau) \, P((\mathbf{y}, \boldsymbol\theta_\mathbf{y}', -\nu'), (\mathbf{x}, \mathrm{d}\boldsymbol\theta_\mathbf{x}, -\nu')\mid S)$. Using this and taking $B$ equals the whole parameter (and auxiliary) space in \eqref{eqn_reversibility_proof}, we have
\begin{align}\label{eqn2_reversibility_proof}
 &(1 - \tau) \int \pi(\mathbf{x}, \boldsymbol\theta_\mathbf{x}) \times (1/2) \left(\int_A P((\mathbf{x}, \boldsymbol\theta_\mathbf{x}, \nu'), (\mathbf{y}, \mathrm{d}\boldsymbol\theta_\mathbf{y}', \nu') \mid S) \right) \, \mathrm{d} \boldsymbol\theta_\mathbf{x} \cr
 &\qquad = (1 - \tau) \int_A \pi(\mathbf{y}, \boldsymbol\theta_\mathbf{y}') \times (1/2) \left(\int  P((\mathbf{y}, \boldsymbol\theta_\mathbf{y}', -\nu'), (\mathbf{x}, \mathrm{d}\boldsymbol\theta_\mathbf{x}, -\nu')\mid S) \right) \, \mathrm{d}\boldsymbol\theta_\mathbf{y}'.
\end{align}

Now, consider that $\by = \bx$, i.e.\ a rejected model switch so model $\mathbf{y}$ is reached from model $\mathbf{y}$ and the direction is such that $-\nu = \nu'$. The probability of the transition is
\[
 (1 - \tau) \int_A \pi(\mathbf{y}, \boldsymbol\theta_\mathbf{y}') \times (1/2) \left(1 - \sum_{\mathbf{x} \in \neigh_{-\nu'}(\by)}\int  P((\mathbf{y}, \boldsymbol\theta_\mathbf{y}', -\nu'), (\mathbf{x}, \mathrm{d}\boldsymbol\theta_\mathbf{x}, -\nu')\mid S) \right) \, \mathrm{d}\boldsymbol\theta_\mathbf{y}'.
\]
So, the total probability of reaching $\mathbf{y}, \boldsymbol\theta_\mathbf{y}'\in A, \nu'$ through a model switch is (recalling \eqref{eqn_inv_1}):
\begin{align*}
 &\sum_{\mathbf{x}, \nu} \, (1 - \tau) \int \pi(\mathbf{x}, \boldsymbol\theta_\mathbf{x}) \times (1 / 2) \left(\int_{A} P((\mathbf{x}, \boldsymbol\theta_\mathbf{x}, \nu), (\mathbf{y}, \mathrm{d}\boldsymbol\theta_\mathbf{y}', \nu') \mid S)\right) \, \mathrm{d}\boldsymbol\theta_\mathbf{x} \cr
 &\quad=\sum_{\mathbf{x}: \by \in \neigh_{\nu'}(\bx)} (1 - \tau) \int \pi(\mathbf{x}, \boldsymbol\theta_\mathbf{x}) \times (1/2) \left(\int_A P((\mathbf{x}, \boldsymbol\theta_\mathbf{x}, \nu'), (\mathbf{y}, \mathrm{d}\boldsymbol\theta_\mathbf{y}', \nu') \mid S) \right) \, \mathrm{d}\boldsymbol\theta_\mathbf{x} \cr
 &\qquad + (1 - \tau) \int_A \pi(\mathbf{y}, \boldsymbol\theta_\mathbf{y}') \times (1/2) \left(1 - \sum_{\mathbf{x} \in \neigh_{-\nu'}(\by)}\int  P((\mathbf{y}, \boldsymbol\theta_\mathbf{y}', -\nu'), (\mathbf{x}, \mathrm{d}\boldsymbol\theta_\mathbf{x}, -\nu')\mid S) \right) \, \mathrm{d}\boldsymbol\theta_\mathbf{y}' \cr
 &\quad = \sum_{\mathbf{x} \in \neigh_{-\nu'}(\by)}  (1 - \tau) \int_A \pi(\mathbf{y}, \boldsymbol\theta_\mathbf{y}') \times (1/2) \left(\int  P((\mathbf{y}, \boldsymbol\theta_\mathbf{y}', -\nu'), (\mathbf{x}, \mathrm{d}\boldsymbol\theta_\mathbf{x}, -\nu')\mid S) \right) \, \mathrm{d}\boldsymbol\theta_\mathbf{y}' \cr
 &\qquad + (1 - \tau) \int_A \pi(\mathbf{y}, \boldsymbol\theta_\mathbf{y}') \times (1/2) \left(1 - \sum_{\mathbf{x} \in \neigh_{-\nu'}(\by)}\int  P((\mathbf{y}, \boldsymbol\theta_\mathbf{y}', -\nu'), (\mathbf{x}, \mathrm{d}\boldsymbol\theta_\mathbf{x}, -\nu')\mid S) \right) \, \mathrm{d}\boldsymbol\theta_\mathbf{y}' \cr
 &\quad =  (1 - \tau) \int_A \pi(\mathbf{y}, \boldsymbol\theta_\mathbf{y}') \times (1/2) \, \mathrm{d}\boldsymbol\theta_\mathbf{y}',
\end{align*}
using \eqref{eqn2_reversibility_proof} and that if $\bx$ allows to reach $\by$ using the direction $\nu'$, then $\bx \in \neigh_{-\nu'}(\by)$. Combining this result with \eqref{eqn_conclusion_up} allows to conclude the proof.
\end{proof}

\section{Supplementary material}\label{sec:supp_mat}

We present in \autoref{example:1} a model such that \eqref{eq:ideal_unif} is satisfied. We next provide an example (\autoref{ex:path}) to illustrate how a careful application of \autoref{thm:2} can allow to conclude that the lifted Markov chain is more efficient than the MH one in certain situations, provided that $n$ is sufficiently large.

\begin{Example}
\label{example:1}
 Let $\pi_n$ be such that
\begin{equation}\label{eq:example1_1}
\pi_n\left\{1<i<j<n\,:\,\inf_{i\leq k\leq j}x_k=1\;,\; \sup_{k\not\in\{i,\ldots,j\}}x_k=-1\;\right\}=1\,.
\end{equation}
By construction, a random variable $\bX\sim \pi_n$ consists of a series of (at least one) $-1$ component(s) followed by a series of (at least two) $+1$ component(s) and then a series of (at least one) $-1$ component(s), $\pi_n$-almost surely. For $i\in\{1,\ldots,n\}$, let  $R_i:\Xset_n\to \Xset_n$  be the operator that flips the $i$-th coordinate, formally defined as $R_i(\bx)=\bx-2x_i\boldsymbol{\delta}_i$, where $\boldsymbol{\delta}_i$ is the Kronecker symbol, i.e. the vector of $\{0,1\}^n$ that has $1$ at entry $i$ and $0$ elsewhere. For $\bx\in\Xset_n$ such that $\{i,j\}$ are as in \eqref{eq:example1_1}, define $\neigh(\bx)$ as $
\neigh(\bx)=\{R_{i-1}(\bx),R_i(\bx),R_{j}(\bx),R_{j+1}(\bx)\}$. By definition, the neighbourhood of $\bx\in\Xset_n$ is made of states obtained by extending or shortening the series of $+1$ components of $\bx_n$.  To split $\neigh(\bx)$ into two directional neighbourhoods, the partial ordering on $\Xset$ is defined through the set
$$
\R=\left\{(\bx,\by)\in\Xset^2\,:\,\inf_{1\leq i\leq n}(y_i-x_i)\geq 0\right\}\,.
$$
Given this partial ordering, $\neigh(\bx)$ is split into $\neigh_{+1}(\bx)=\{R_{i-1}(\bx),R_{j+1}(\bx)\}$ and $\neigh_{-1}(\bx)=\{R_{i}(\bx),$ $R_{j}(\bx)\}$, where $\{i,j\}$ are as in \eqref{eq:example1_1}. Clearly for any $\pi_n$ which satisfies \eqref{eq:example1_1}, we have for $\pi_n$-almost all $\bx\in\Xset_n$, $|\neigh_{+1}(\bx)|=|\neigh_{-1}(\bx)|=2=|\neigh(\bx)|/2$ and we are in the context of \autoref{cor_domination_unif}. A specific distribution $\pi_n$ which verifies  \eqref{eq:example1_1} is defined as follows: let $I:=\inf\{i\,:\,\bX_i=1\}$ and $L:=\sum_{k=1}^n\1_{\bX_k=1}$ follow a truncated geometric distribution with parameters $\lambda_1\in(0,1)$ and $\lambda_2\in(0,1)$ respectively such that $\pi_n\{1<I<n-1\;\cap\;1<L<n-2\;\cap\;L+I<n\}=1$. Since everything is tractable in this example, asymptotic variances $\vara(f,P)$ can be calculated exactly for a given Markov kernel $P$ and a test function $f$. The right panel of \autoref{fig:ex1} shows the ratio of asymptotic variances $\vara(f,P_{\MH,n})/\vara(f,P_{\rho,n})$ for three different functions $f$. Here, the simplest switching rate function was used $\rho\equiv \rho_{\nu}^w$, i.e. $\rho_\nu(\bx)=1-T_\nu(\bx,\Xset_n)$. As anticipated by \autoref{cor_domination_unif}, these ratios are always larger than one. However, this experiment shows that they can indeed be \textit{much} larger than one and increase with $n$, hence justifying the lifted approach. Intuitively, the mild variations of $\pi_n$ over neighbouring states (see left panel of \autoref{fig:ex1}) explain why the lifted Markov chain outperforms significantly the MH one in this example: the persistent nature of the lifted chain increases (or decreases) consistently the length of the $+1$ series until an unlikely rejection occurs (since $a_\nu(\bx,\by)\approx 1$, $\by\in\neigh_\nu(\bx)$) or that the boundary of the support is reached.
\begin{figure}
\includegraphics[width=0.49\textwidth]{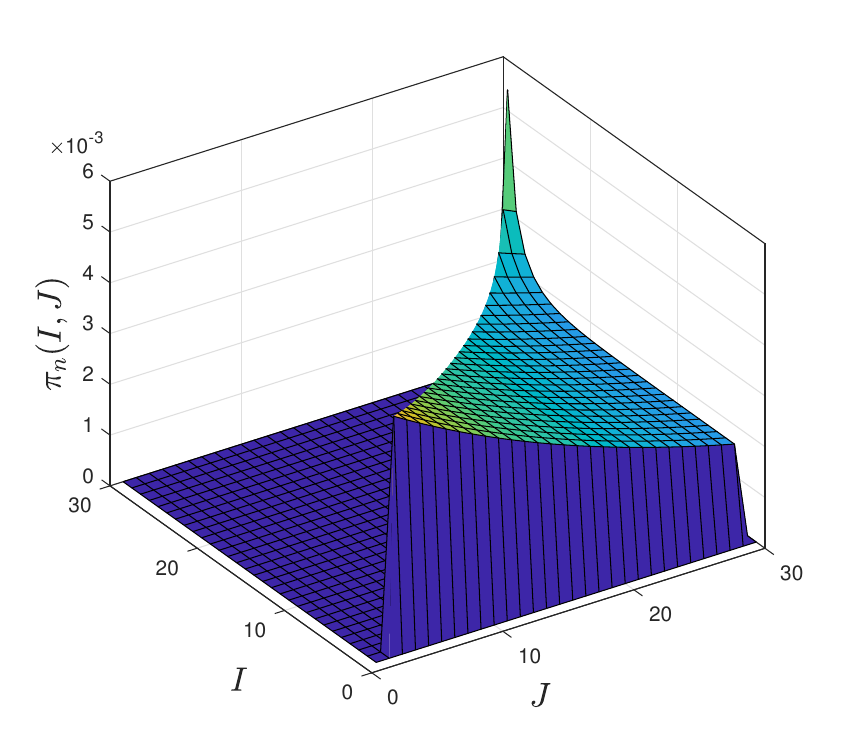}\includegraphics[width=0.49\textwidth]{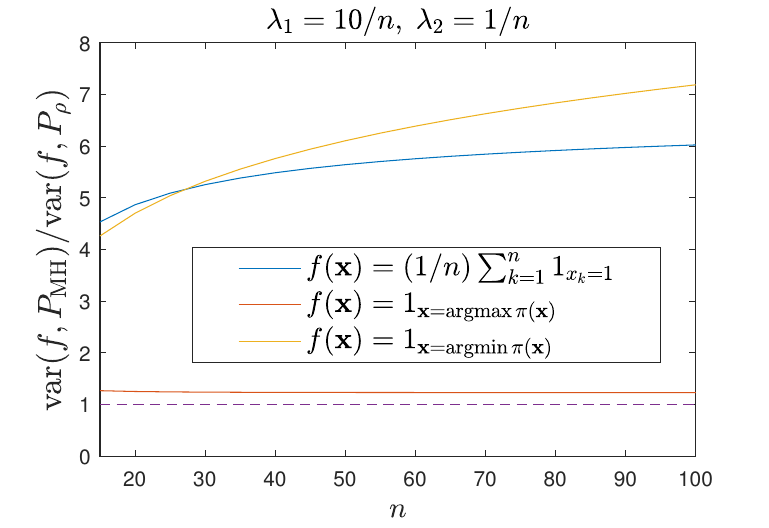}
\caption{Left: illustration of the distribution when $(I,J=I+L)$ with $n=30$, $\lambda_1=1/3$ and $\lambda_2=1/30$. Right: ratio of the asymptotic variances $\vara(f,P_{\MH,n})/\vara(f,P_{\rho,n})$ for three different functions $f$. As established in \autoref{cor_domination_unif}, we always have $\vara(f,P_{\MH,n})/\vara(f,P_{\rho,n})\geq1$ and for some functions, in addition of being several times larger than 1, that ratio increases significantly with $n$.\label{fig:ex1}}
\end{figure}
\end{Example}

\begin{Example}
\label{ex:path}
Let $\Xset_n=\{-1,+1\}^{n}$. In this example, in addition to the state space dimension, $n\in\nset$ also characterizes the geometrical features of $\pi_n$:  as $n\to\infty$ more and more probability mass is put on a structure that can be seen as a path $\mathfrak{P}_n$ within $\Xset_n$ defined as
$$
\mathfrak{P}_n:=\left\{\bx\in\Xset_n\;:\;\sup_{i\leq j}x_i=-1\,,\;\inf_{i> j}x_i=+1\,,\quad \text{for some}\,j\in\{1,\ldots,n-1\}\right\}\cup\{-1\}^{n}\cup\{+1\}^{n}\,.
$$
The states belonging to the path are denoted $\{\bx_i\}_{i=1}^{n+1}$ such that $\{\bx_{1}\prec \bx_{2}\prec \cdots\prec \bx_{n+1}\}$. Therefore, $\bx_1=\{-1\}^n$, $\bx_2=\{-1,\ldots,-1,+1\}$, $\bx_3=\{-1,\ldots,-1,+1,+1\}$ and so on. Moreover, we define the subset $\mathfrak{P}_n^{\circ}=\mathfrak{P}_n\backslash\{\{-1\}^n\cup\{+1\}^n\}$. In order to set up a context that resemble applications, the natural neighborhood structure defined as $\neigh(\bx)=\{\by\in\Xset_n\,:\,\sum_{i=1}^n|x_i-y_i|=1\}$ is slightly modified as follows:
\begin{itemize}
  \item if $\bx\not\in\mathfrak{P}_n$, $\tneigh(\bx)=\neigh(\bx)\backslash\mathfrak{P}_n^{\circ}$ and  if $\bx\in\mathfrak{P}_n^{\circ}$, $\tneigh(\bx)=\neigh(\bx)\cap \mathfrak{P}_n^{\circ}$,
  \item if $\bx\in\{-1\}^n\cup\{+1\}^n$, $\tneigh(\bx)=\neigh(\bx)\cup\{-\bx\}$.
\end{itemize}
The first modification is designed so as to account that in many applications,  the part of the state space on which $\pi_n$ concentrates has a certain depth (it is not possible to exit this subset in a one-step transition for the vast majority of states in that subset),  such as in the Ising model of \autoref{sec:Ising}. The second modification provides the state space with a torus-like feature, since the extreme states $\bx_{1}$ and $\bx_{n}$ are neighbors. The neighborhood structure $\tneigh$ induces the following mapping on $\Xset_n\times\Xset_n$:
\begin{multline*}
d_n(\bx,\by)=
\inf\{i\in\nset\,:\,\exists\{\bz_1,\ldots,\bz_{i-1}\}\in\Xset_n^{i-1} \quad \text{such that}\\
  \bigcup_{\nu\in\{-1,+1\}}\{\bz_1 \in\tneigh_\nu(\bx)\}\cap_{k=1}^{i-2} \bz_{k}\in\tneigh_\nu(\bz_{k+1}) \cap \{\bz_{i-1} \in\tneigh_\nu(\by)\}\}\,,
\end{multline*}
which defines a distance on $\Xset_n$.  This distance is the smallest number of transitions in the same direction required to go from $\bx$ to $\by$, meaning that  $\{\bx \succ \bz_1\succ\cdots\succ \bz_{d_n(\bx,\by)-1}\succ \by\}$ or $\{\bx \prec \bz_1\prec\cdots\prec \bz_{d_n(\bx,\by)-1}\prec \by\}$. In the following, this distance is used to measure the distance from an arbitrary state:
\begin{itemize}
\item to the center of the path,
$$
\bx\in\Xset_n\,,\qquad \mathfrak{d}_n(\bx)=d_n(\bx,\bx_{n/2+1})\1_{\{n\,\text{is even}\}}+(d_n(\bx,\bx_{(n+1)/2})\wedge d_n(\bx,\bx_{(n+3)/2})\1_{\{n\,\text{is odd}\}}\,,
$$
\item to the  path,
$$
\bx\in\Xset_n\,,\qquad \mathsf{d}_n(\bx)=\left[d_n(\bx,\{-1\}^n)\wedge d_n(\bx,\{+1\}^n)\right]\1_{\{\bx\not\in\mathfrak{P}_n\}}\,.
$$
\end{itemize}
This allows to define the distribution $\pi_n$, parameterized by $a\in(2/3,1)$, as
\begin{equation}\label{eq:model_suppmat}
\bx\in\Xset_n\,,\qquad \pi_{n}(\bx)\propto \left(\frac{1}{3}\right)^{n \mathsf{d}_n(\bx)} a^{\mathfrak{d}_n(\bx)}\,.
\end{equation}
Intuitively, the mode of $\pi_n$ is at the centre of the path, the mass decays geometrically along the path and beyond the path, the mass is further shrunk in the tails.
\end{Example}

Consider the problem of sampling from $\pi_n$ of \autoref{ex:path} using the locally-balanced version of either the lifted or the MH algorithms. A careful application of \autoref{thm:2} shows that, provided that $n$ is sufficiently large, the lifted Markov chain is more efficient than the MH one for, at least, a certain function of interest, namely $f_n(\bx)=\sum_{i=1}^n x_i/n$.

We first show that $\mathfrak{P}_n$ can be used as a control subset, thus setting $\tXset_n =\mathfrak{P}_n$. By construction, the interior and boundary of $\tXset_n$ are defined as
$$
\itXset_n=\mathfrak{P}_n\backslash\{\bx_1\cup \bx_{n+1}\}\,,\qquad
\partial\tXset_n=\bx_1\cup \bx_{n+1}\,.
$$
When $n\to\infty$, $\pi_n$ concentrates on $\mathfrak{P}_n$ exponentially fast. In particular, it can be checked that for all $r<1/2$,
\begin{equation}\label{eq:rate_path_concentrate}
\lim_{n\to\infty}\left(\frac{1}{a^{r}}\right)^n\left(1-\pi_n(\itXset_n)\right)= 0\,.
\end{equation}
Indeed, considering that $n$ is  odd (a similar derivation holds if $n$ is even), denoting by $\bar{\pi}_n$ the unnormalized probability, we note that since $0<\bar\pi_n(\Xset_n\backslash \mathfrak{P}_n)\leq 2^n a^{(n+1)/2}/3^n$ and $\bar{\pi}_n(\mathfrak{P}_n)=2(1+a+\cdots+a^{(n-1)/2})$
\begin{multline*}
1-\pi_n(\itXset_n)=\frac{\bar{\pi}_n(\Xset_n\backslash \mathfrak{P}_n)+\bar{\pi}_n(\partial\tXset_n)}{\bar{\pi}_n(\Xset_n\backslash \mathfrak{P}_n)+\bar{\pi}_n(\mathfrak{P}_n)}
\leq \frac{2^n a^{(n+1)/2}/3^n+2a^{(n-1)/2}}{2(1-a^{(n+1)/2})/(1-a)}%
\leq \left(a^{1/2}\right)^n\frac{(1-a)}{\sqrt{a}(1-a^{n/2})}\,.
\end{multline*}
Recall that the locally-balanced function $g$ used in the locally-balanced proposal presented at \autoref{sec:zanella} is such that $g(x) = x/(1+x)$ and note that, by direct calculation, we obtain that for each $(\bx,\by)\in\mathfrak{P}_n$, with $\bx\neq \by$ and any $n\in\nset$,
$$
\frac{P_{\rev,n}(\bx,\by)}{P_{\MH,n}(\bx,\by)}\in\left\{\frac{1+a}{2},\frac{a/2}{1+a\kappa_n}\left(\frac{1}{2}+\frac{3}{2a}+\kappa_n\right),\frac{1}{2}+\frac{1}{1+a}\left(1+\kappa_n/a)\right)\right\}\,,\quad \kappa_n=(n-1)a/(a+3^n)\,.
$$
It can be checked that for $n$ sufficiently large (in fact for those $n$ satisfying $n  3^{-n}\leq (1-a)/(2a^3)$), we have
\begin{equation}
\label{eq:suppmat0}
P_{\rev,n}(\bx,\by)\geq \frac{1+a}{2} P_{\MH,n}(\bx,\by)\,,\qquad(\bx,\by)\in\tXset_n^2,\,\bx\neq \by\,.
\end{equation}
Equation \eqref{eq:suppmat0} shows that \autoref{ass:1} of \autoref{thm:2} holds with $\omega(n)=(1+a)/2$, for all $n\in\nset$.

Based  on numerical results carried out on a computer (see Figure \ref{fig:gaps}), the spectral gap of $P_{\rev,n}$, $\tP_{\rev,n}$, $P_{\MH,n}$ and $\tP_{\MH,n}$ are surmised to vanish at a quadratic rate, that is $\underline{\lambda}(n)\sim \lambda/n^2$. Notice that the result for $\tP_{\MH,n}$ is in line with some well known results in \cite{diaconis2000analysis} on polynomially mixing Markov chains, see also \cite{diaconis2013some}. Assuming that this result holds, we have that for any $\delta>0$ and $\gamma<\bar{\delta}=\delta/(2+\delta)$, $1-\pi_n(\itXset_n)=o(\underline{\lambda}(n)^{3/(\bar{\delta}-\gamma)})$, that is the assumption \eqref{eq:ass2_thm3} of \autoref{thm:2} holds for any choice of $\delta>0$ and $\gamma<\bar{\delta}$.

\begin{figure}
\centering
\includegraphics[scale=1]{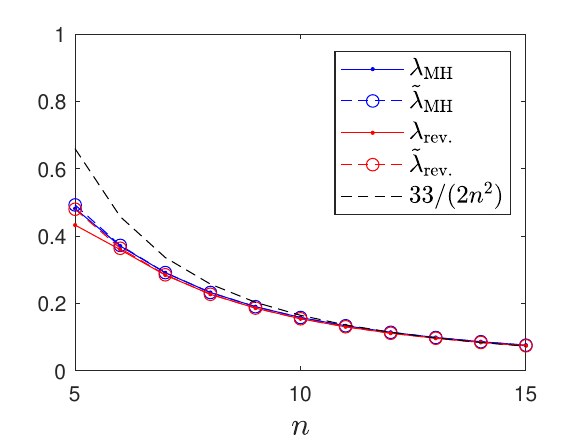}
\caption{Illustration of the spectral gap of $P_{\MH,n}, P_{\rev,n}$ and their restricted version  obtained from a calculation carried out on a computer for $n\in\{5,15\}$. Beyond the dimension $n=15$, calculation of the transition matrices failed. \label{fig:gaps} }
\end{figure}

Consider the function ${f}_n(\bx)=(1/n)\sum_{i=1}^n x_i$. One can check that, because $f_n$ is odd and the mass function $\pi_n$ is even, $\pi_n{f}_n=0$. 
Moreover, noting that $\|f_n\|_{\pi_n,2}>0$ and $\|f_n\|_{\pi_n,2}\neq 1$,  ${f}_n$ needs to be  normalized to remain in the framework of \autoref{thm:2}. We thus seek to compare $\vara(P_{\MH,n},\bar{f}_n)$ and its lifted counterpart $\vara(P_{\rho,n},\bar{f}_n)$. The last check to apply \autoref{thm:2} is that $\|\bar{f}_n\|_{\pi_n,2+\delta}$ does not grow to fast relatively to $1/(1-\pi_n(\itXset))$. We note that because $|f_n|\in[0,1]$,
\begin{equation}
\label{eq:ex22}
\|\bar{f}_n\|_{\pi_n,2+\delta}\leq \frac{1}{\|{f}_n\|_{\pi_n,2}}\,.
\end{equation}
We conclude using the following rough, but not too rough, lower bound on $\|f_n\|_{\pi_n,2}$. Consider $n$ even, then
\begin{equation*}
  \|f_n\|_{\pi_n,2+\delta}^2\geq \sum_{\bx\in\tXset} f_n(\bx)^2\pi_n(\bx)=\frac{8}{n^2 Z_n}W_n\,,\qquad W_n:=\sum_{k=1}^{n/2}k^2a^k\,,
\end{equation*}
where $Z_n$ is the normalizing constant in the definition of $\pi_n$ given at Eq. \eqref{eq:model_suppmat}. We have that for any $\upsilon>0$,
$$
 \frac{1}{n^{1+\upsilon}}\|\bar{f}_n\|_{\pi_n,2+\delta}\leq \frac{1}{n^{1+\upsilon}\|{f}_n\|_{\pi_n,2}}\leq \frac{1}{n^\upsilon}\frac{1}{2\sqrt{2}}\left[\frac{Z_n}{W_n}\right]^{1/2}\,.
$$
A similar derivation holds whenever $n$ is odd. Now, noting that
$$
Z_n\geq 2\left[1+a+a^2+\cdots+a^{n/2}\right]=2\frac{1-a^{n/2+1}}{1-a}\sim 2\frac{1}{1-a}\,,
$$
and that the series $W_n$ converges to $a(a+1)/(1-a)^3$, we have that $\|\bar{f}_n\|_{\pi_n,2+\delta}=o(n^{1+\upsilon})$.  Together with \eqref{eq:rate_path_concentrate}  show that $\{\bar{f}_n\}$ verifies  \eqref{eq:ass1_thm3} of \autoref{thm:2}. Therefore, for any $\epsilon>0$, we have for a sufficiently large $n$ that, using $\bar{\omega}=(1+a)/2$ (see Eq. \eqref{eq:suppmat0})
$$
\vara(\bar{f}_n, P_{\rev,n}) \leq \frac{1}{(a+1)/2-\eps}\vara(\bar{f}_n, P_{\MH,n}) +\frac{1}{2}\left(\frac{1}{(a+1)/2-\eps} + \frac{1}{(a+1)/2}\right) - 1 + \eps/2\,.
$$
For instance, taking $\eps=(1-a) / 2$, for $n$ sufficiently large,
\begin{equation}\label{eq:suppmat1}
\vara(\bar{f}_n,P_{\rho,n})\leq \frac{1}{a}\vara(\bar{f}_n, P_{\MH,n}) +\frac{1}{2}\left(\frac{1}{a} + \frac{1}{(a+1)/2}\right) - 1 + \frac{1-a}{4}\,.
\end{equation}
It is remarkable that such a result can be obtained since, for \textit{all} $n\in\nset$, it does not hold that $P_{\rev,n}(\bx,\by)-P_{\MH,n}(\bx,\by)\geq 0$ for all $(\bx,\by)\in\Xset^2$, $\bx\neq \by$. Hence, $P_{\rev,n}$ does not dominate $P_{\MH,n}$ in the usual Peskun sense and the result of \cite{andrieu2019peskun} does not allow to compare $\vara(\bar{f}_n,P_{\rho,n})$ and $\vara(\bar{f}_n, P_{\MH,n})$. Indeed, one can check that taking $\bz\in\partial\tXset$ and $\bx\not\in\tXset$ such that $\bz\in\neigh_\nu(\bx)$, we have that
$$
P_{\MH,n}(\bz,\bx)=\frac{g(3^n/a)}{1/2+g(1/a)+(n-a)g(a/3^n)}\,,\qquad P_{\rev,n}(\bz,\bx)=1/2\,,
$$
and, for all $n\in\nset$,
$$
\frac{P_{\MH,n}(\bz,\bx)}{P_{\rev,n}(\bz,\bx)}\geq \frac{(1+a)3^n}{1+n+3^n}>1\,.
$$
This means that, if, hypothetically, one were able to establish a quantitative Peskun ordering between $P_{\rev,n}$ and $P_{\MH,n}$, they would obtain something like $P_{\rev,n}(\bx,\by)\geq \omega(n)P_{\MH,n}(\bx,\by)$ for each $(\bx,\by)\in\Xset^2$ with $\omega(n)\leq   ({3^{-n}+n3^{-n}+1})/(1+a)$ which decreases with $n$ to $1/(1+a)$. Assuming the best case scenario with $\omega(n)=1/(1+a)$, the comparison between the lifted and MH asymptotic variances would then be
\begin{equation}
\label{eq:suppmat2}
\vara(\bar{f}_n,P_\rho)\leq (1+a)\vara(\bar{f}_n, P_{\MH,n}) +a\,,\qquad \text{for all }n\in\nset\,.
\end{equation}
Taking the constant $a$ arbitrarily close to one, one can compare (for a large enough $n$) the difference in tightness offered by the two bounds of Eqs. \eqref{eq:suppmat1} and \eqref{eq:suppmat2}. Thus, not only the weaker Peskun ordering introduced in this paper allows one to establish an ordering with greater ease, but the asymptotic variances inequality can also be much tighter.

\end{document}